\documentclass[aps,english,prx,floatfix,amsmath,superscriptaddress,tightenlines,twocolumn,longbibliography,nofootinbib]{revtex4-2}
\usepackage{hyperref}
\usepackage{graphicx}
\usepackage{amsmath}
\usepackage{tikz}
\usepackage{soul,xcolor}
\usepackage{amssymb}
\usepackage{amsthm}
\usepackage{thmtools} 
\usepackage{lipsum, babel}

\usepackage[shortlabels]{enumitem}
\usepackage{tikz-cd}
\usepackage{breakurl}

\setstcolor{red}
\usetikzlibrary{positioning}
\usetikzlibrary{patterns}
\usetikzlibrary{arrows.meta}
\usetikzlibrary{spy}
\tikzset{invclip/.style={clip,insert path={{[reset cm]
				(-1638 pt,-1638 pt) rectangle (1638 pt,1638 pt)}}}}

\makeatother

\theoremstyle{definition}
\newtheorem{definition}{Definition}[section]
\newtheorem*{remark}{Remark}
\theoremstyle{theorem}
\newtheorem{theorem}{Theorem}[section] 
\theoremstyle{corollary}

\theoremstyle{lemma} 
\newtheorem{lemma}[theorem]{Lemma}
\theoremstyle{Proposition} 
\newtheorem{Proposition}[theorem]{Proposition}

\definecolor{IKblue}{RGB}{25,25,125}
\definecolor{BSorange}{RGB}{140,50,0}
\newcommand{\merge}[0]{\Join}

\newcommand{\axiomone}[2] %I made the line thiner for this command.
{
\begin{scope}[xshift=#1cm, yshift=#2cm]
    \draw[line width=0.6pt, fill=red!10!white, opacity=0.5] (0,0) circle (0.3cm);
    \draw[line width=0.6pt] (0,0) circle (0.3cm);
    \draw[line width=0.6pt] (0,0) circle (0.15cm);
\end{scope}
}

\newcommand{\Rom}[1]{\uppercase\expandafter{\romannumeral #1\relax}}

\newcommand{\calC}{{\mathcal C}}
\newcommand{\calD}{{\mathcal D}}
\newcommand{\calE}{{\mathcal E}}

\newcommand{\calH}{{\mathcal H}}

\newcommand{\Tr}{{\rm Tr}}
\newcommand{\conv}{{\rm conv}}

\begin{document}
\title{Entanglement bootstrap approach for gapped domain walls} 	
\author{Bowen Shi}
\affiliation{Department of Physics, University of California at San Diego, La Jolla, CA 92093, USA}
\affiliation{Department of Physics, The Ohio State University, Columbus, OH 43210, USA}
\author{Isaac H. Kim}
\affiliation{3 Centre for Engineered Quantum Systems, School of Physics, University of Sydney, Sydney, NSW 2006, Australia}
\affiliation{Stanford University, Stanford, CA 94305, USA}
\date{\today}

\begin{abstract}
We develop a theory of gapped domain wall between topologically ordered systems in two spatial dimensions. We find a new type of superselection sector -- referred to as the parton sector -- that subdivides the known superselection sectors localized on gapped domain walls. Moreover, we introduce and study the properties of composite superselection sectors that are made out of the parton sectors. We explain a systematic method to define these sectors, their fusion spaces, and their fusion rules, by deriving nontrivial identities relating their quantum dimensions and fusion multiplicities. We propose a set of axioms regarding the ground state entanglement entropy of systems that can host gapped domain walls, generalizing the bulk axioms proposed in [B. Shi, K. Kato, and I. H. Kim, Ann. Phys. 418, 168164 (2020)]. Similar to our analysis in the bulk, we derive our main results by examining the self-consistency relations of an object called information convex set. As an application, we define an analog of topological entanglement entropy for gapped domain walls and derive its exact expression. 
\end{abstract}

\maketitle

\tableofcontents

\section{Introduction}
\label{sec:intro}
One of the fundamental discoveries in physics is topologically ordered phases of matter~\cite{Wen2004}. These are gapped phases of quantum many-body systems that possess low-energy excitations with fractional statistics~\cite{Leinaas1977,PhysRevLett.48.1144,PhysRevLett.53.722,Kalmeyer1987,Moore1991}. A prominent experimental example is the well-known fractional quantum Hall states~\cite{Tsui1982}. 

While these systems already exhibit a rich set of phenomena in the bulk of the material, more new physics can appear on their boundaries. The existence of a robust gapless boundary mode is well-known~\cite{1997PhRvB..5515832K,Kitaev2006solo}. The nontrivial effects of gapped boundary conditions on the topological ground state degeneracy~\cite{Bravyi1998} and low-energy excitations~\cite{Beigi2010} have also been studied. 

More generally, there can be \emph{gapped domain walls} between two different topologically ordered mediums~\cite{KitaevKong2012,2013PhRvX...3b1009L,Barkeshli2013a,Kong2014,Lan2015, Hung2015}. Gapped domain walls are not just of theoretical interest. When used in conjunction with the low-energy excitations, the domain walls can complete a universal set of topologically protected quantum gates~\cite{Cong2017}. Therefore, studies of gapped domain walls may lead to new means of building a fault-tolerant quantum computer~\cite{Kitaev2003}.

While there have been a number of beautiful prior works that studied gapped domain walls in various contexts~\cite{Beigi2010,KitaevKong2012,2011NuPhB.845..393K,2013CMaPh.321..543F,2013PhRvX...3b1009L,Barkeshli2013a,Bais2009,Kong2014,2015PhRvB..91l5124W,Lan2015, Hung2015,Hung2015a,Neupert2016,Neupert2016a,Cong2017a,2018JHEP...01..134H,2019PhRvB..99o5134M,2019arXiv190108285S,2019JHEP...05..110H,2019arXiv190706692B,Lan2019,2020arXiv200811187L}, there are still many unknowns. For one, less is known about the order parameters that characterize gapped domain walls. In the bulk of a topologically ordered system, entanglement-based measures~\cite{Kitaev2006,Levin2006,Zhang2012} are useful for characterizing the underlying topological order~\cite{Isakov2011,Jiang2012,Cincio2013}. However, similar measures for gapped domain walls are not known to the best of our knowledge.

Moreover, while a theory of gapped domain wall has been proposed already~\cite{Kong2014}, this theory is based on an assumption about the condensation algebra~\cite{Bais2009,KitaevKong2012}, which abstracts away the microscopic physics of the underlying many-body quantum system. The abstractness of this theory is both a blessing and a curse. It allows us to identify the fundamental data that characterize the gapped domain wall without ever dealing with the microscopic physics. However, the downside is that it is not always clear how to extract these data directly from the original many-body system. Moreover, one may contest that the rules set out in this theory may not constitute a complete theory of gapped domain walls. While this is a sentiment that we do not necessarily share, it will still be desirable to derive these rules from a more microscopic assumption about the underlying physical system. 

To address these issues, we applied a recently discovered approach to studying topological order~\cite{SKK2019} to systems separated by a gapped domain wall. In Ref.~\cite{SKK2019}, we derived the axioms of the fusion theory of anyon and the expression for the topological entanglement entropy -- defined as the subleading contribution to the ground state entanglement entropy -- from a set of simple assumptions on ground-state entanglement. In this paper, we extend this analysis to systems that possess a gapped domain wall, by relaxing the set of assumptions used in Ref.~\cite{SKK2019} appropriately; see Fig.~\ref{fig:axioms_all} for the summary of these assumptions.

\begin{figure}
  \centering
  \includegraphics[scale=1]{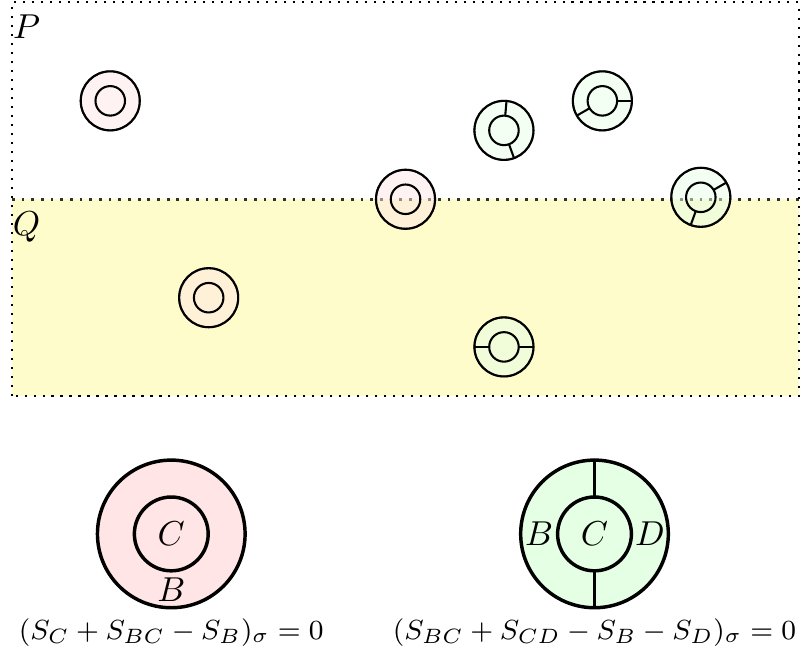}
    \caption{The starting assumptions of this paper. Here, $\sigma$ is the ground state and $S_{A}$ is the entanglement entropy of a subsystem $A$ with respect to the state $\sigma$. For details, see Section~\ref{sec:fusion_from_entanglement}. (Top) We consider topologically ordered mediums $P$ (upper half) and $Q$ (lower half) which are separated by a gapped domain wall. (Bottom) We assume that the ground state $\sigma$ locally obeys two types of entropic constraints. These constraints are imposed on balls of bounded radius. While these constraints hold everywhere in the bulk, the second constraint is relaxed on the domain wall.}
    \label{fig:axioms_all}
\end{figure}

From these assumptions, we were able to identify a new set of superselection sectors localized at the domain wall. These sectors, which we refer to as the \emph{parton sectors}, will be the main subject of this paper. These are ``parton-like'' in the sense that other superselection sectors are composite objects made from these sectors. One example of such a composite sector is the superselection sectors of point excitations on the domain wall, which have been studied in Refs.~\cite{KitaevKong2012,Kong2014}. However, there are other types of composite sectors that are new to the best of our knowledge. 

Both the parton and the composite sectors can be ``fused'' together like the superselection sectors appearing in the bulk of the topological phase. However, the ordinary rule of fusion does not always apply. When we say fusion, we usually mean that there are two sectors, say $a$ and $b$, that fuses into $c$. The state space in which $a$ and $b$ fusing into $c$ is isomorphic to the state space of some Hilbert space. However, when we fuse parton sectors, the state space in which two parton sectors fuse into another parton sector may not be isomorphic to any such state space. We refer to this phenomenon as \emph{quasi-fusion} and later explain how this difference arises.

Another strange thing about the parton sectors is that they should not be viewed as low-energy excitations. Generally, a single parton by itself cannot completely specify an excitation. Instead, parton labels should be considered as quantum numbers that partially determine the excitation.

Despite their bizarre nature, parton sectors are actual physical objects. There are operators localized on the $N$- and $U$-shaped regions in Fig.~\ref{fig_UN_intro} that can measure these sectors. More concretely, for every parton sector, there is an operator that can unambiguously detect the presence of that sector. As such, parton sectors should be treated as fundamental objects in any theory of gapped domain walls.

\begin{figure}
  \centering
  \includegraphics[scale=1]{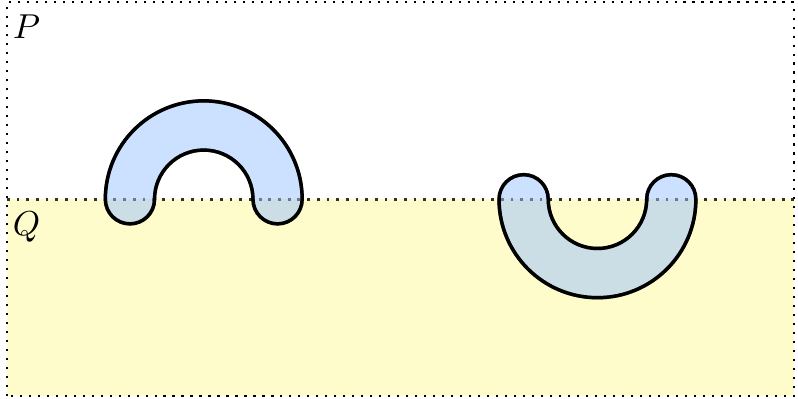}
    \caption{For every parton sector, there is an operator acting either on the $N$-shaped(left) or $U$-shaped(right) region that can unambiguously detect the sector.}
    \label{fig_UN_intro}
\end{figure}

\begin{figure}[h]
  \centering
  \includegraphics[scale=1]{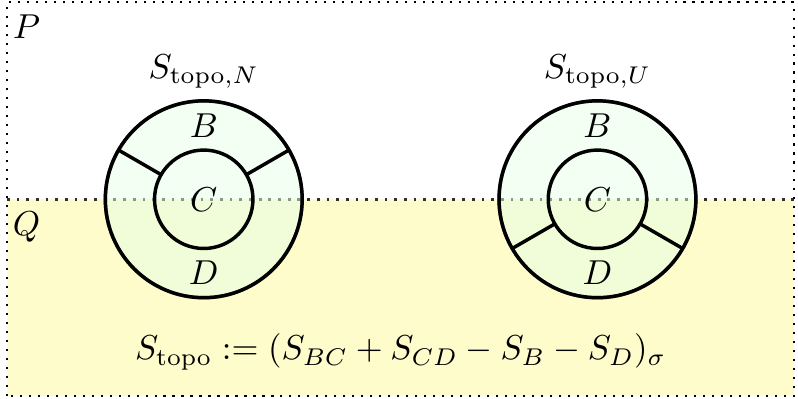}
    \caption{Subsystems involved in the calculation of the domain wall topological entanglement entropy.}
    \label{fig:tee_configs}
\end{figure}

To examine whether a given microscopic system can host parton sectors, calculating ground state entanglement can be a fruitful approach. We prove, starting from a set of assumptions summarized in Fig.~\ref{fig:axioms_all}, that the linear combination of entanglement entropy in Fig.~\ref{fig:tee_configs} must be equal to
\begin{equation}
    \begin{aligned}
    S_{\text{topo}, N} &= 2 \ln \mathcal{D}_N,\\
    S_{\text{topo}, U} &= 2 \ln \mathcal{D}_U, 
    \end{aligned}
\end{equation}
where  $\mathcal{D}_N=\sqrt{\sum_n d_n^2}$ and $\mathcal{D}_U = \sqrt{\sum_u d_u^2}$ are the total quantum dimension of two different types of parton sectors referred to as $U$- and $N$-sectors. In analogy with the topological entanglement entropy~\cite{Kitaev2006,Levin2006}, we refer to these ``order parameters'' as \emph{domain wall topological entanglement entropies}. More discussion on this order parameter will appear in our companion paper~\cite{EntanglementBootstrap_short}.

Notwithstanding the rich physics of parton sectors, perhaps the most remarkable fact of all is that all of these results followed entirely from Fig.~\ref{fig:axioms_all}. No assumption on the parent Hamiltonian was necessary. The notion of superselection sectors was derived, instead of being imposed. The existence of fusion spaces was, again, derived. These facts compel us to name our approach as \emph{entanglement bootstrap,} in analogy with the conformal bootstrap program~\cite{Ferrara1973,Polyakov1974}. 

While there are many conclusions one can make from this work, the following two stand out. First, in the presence of gapped domain walls, there is a new type of superselection sector called parton sector. Parton sectors are more fundamental than the other sectors in the sense that they subdivide the other sectors. These findings suggest that there is more to be understood about gapped domain walls than previously thought. 

The second lesson is somewhat philosophical. We often do physics by beginning with a specific Lagrangian/Hamiltonian in mind and then computing various properties of the theory from those objects. Alternatively, one may write a set of consistency equations coming from the underlying symmetry of the theory~\cite{Polyakov1974}. Our work shows that there is a third possibility, namely a possibility to study the theory from the properties of ground state entanglement. Let us again emphasize that, in our study, we did not invoke any assumption about the action or the symmetry. All that was required was the set of consistency equations coming from the property of the ground state entanglement. The fact that a new physics can be uncovered this way is, in our opinion, surprising and certainly warrants further exploration. 

The rest of this paper is structured as follows. In Section~\ref{sec:fusion_from_entanglement}, we review Ref.~\cite{SKK2019}, focusing on the key ideas that are used in this work. In Section~\ref{sec:parton}, we explain how the assumptions used in Ref.~\cite{SKK2019} are modified in the presence of gapped domain walls. In particular, we deduce the existence of the parton sectors, which are the central objects of this paper. In Section~\ref{sec:composite_sectors}, we study the composite superselection sectors that are made out of the parton sectors. We begin with a few examples and conclude with the general lessons. In Section~\ref{sec:fusion}, we introduce a method to construct the fusion spaces of these sectors. In Section~\ref{sec:fusion_rules}, we study the fusion rules. In particular, we derive a number of nontrivial identities relating the fusion multiplicities to the quantum dimensions. In Section~\ref{Sectionquasi_fusion}, we discuss the \emph{quasi-fusion rules} of the parton sectors, which generalize the ordinary fusion rules. In general, more than one fusion space is needed to describe a quasi-fusion process, even if both the parton sectors before and after the quasi-fusion are completely specified. In Section~\ref{sec:tee}, we derive various expressions for the topological entanglement entropies of domain walls. In Section~\ref{sec:string}, we discuss the properties of the string-like operators that can create the superselection sectors we have studied in this paper. We conclude in Section~\ref{sec:discussion}, listing some open problems and directions to pursue.

\section{Fusion rules from entanglement}
\label{sec:fusion_from_entanglement}
Our theory of gapped domain walls rests on our recent work on anyons~\cite{SKK2019}. Before this study, the theory of anyons was based on a mathematical framework called unitary modular tensor category theory~\cite{Kitaev2006solo}. However, in Ref.~\cite{SKK2019}, many basic rules of that framework emerged from a generic property of entanglement in gapped ground states. In this section, we provide a brief overview of this work, focusing on the parts relevant to this paper.

To start with, we explain an important concept called \emph{information convex set}~\cite{SKK2019,Kim2015sydney,Shi2018}. The information convex set is essential in understanding Ref.~\cite{SKK2019} because the key physical objects of interest emerge from this definition. To explain this concept, let us consider a subsystem of a two-dimensional lattice, denoted as $\Omega$. Let $\Omega' \supset \Omega$ be a subsystem obtained by enlarging $\Omega$ along its boundary by an amount large compared to the correlation length. The information convex set is defined as follows:
\begin{equation}
    \Sigma(\Omega) := \left\{\Tr_{\Omega'\setminus \Omega}\left( \rho_{\Omega'}\right)| \rho_b = \sigma_b \quad \forall b\in \mathcal{B}(r), b\subset \Omega' \right\},
\end{equation}
where $\sigma$ is some fixed \emph{global reference state}. It is helpful to think of this state as a ground state of some gapped Hamiltonian, although we do not make use of that fact. Here, $\mathcal{B}(r)$ is the set of balls of bounded radius $r=\mathcal{O}(1)$, where $r$ is chosen to be large compared to the correlation length. 

As it stands, aside from the fact that it is convex, the information convex set does not have any particularly noteworthy structure. However, much more can be said about this set once we incorporate physically motivated axioms on the reference state $\sigma$. To that end, Ref.~\cite{SKK2019} advocated two physical axioms. Specifically, the axioms state that
\begin{equation}
\begin{aligned}
    (S_C + S_{BC} - S_B)_{\sigma} &= 0 \\
    (S_{BC} + S_{CD} - S_B - S_D)_{\sigma} &=0
\end{aligned}
\label{eq:axioms_bulk}
\end{equation}
over a set of subsystems depicted in Fig.~\ref{fig:bulk_axiom_to_fusion_data}, where $S(\rho)=-\Tr(\rho \ln \rho)$ is the von Neumann entropy of $\rho$. Here, we specified $\sigma$ in the subscript of the parenthesis because the underlying global state is the same for all the entanglement entropies in the linear combination. The subscript of $S$ represents the relevant subsystem. For instance, $S_B$ appearing in an expression like $(\ldots + S_B+\ldots)_{\sigma}$ represents $S(\sigma_B)$.

\begin{figure}[h]
	\centering
        \includegraphics[width=0.95\columnwidth]{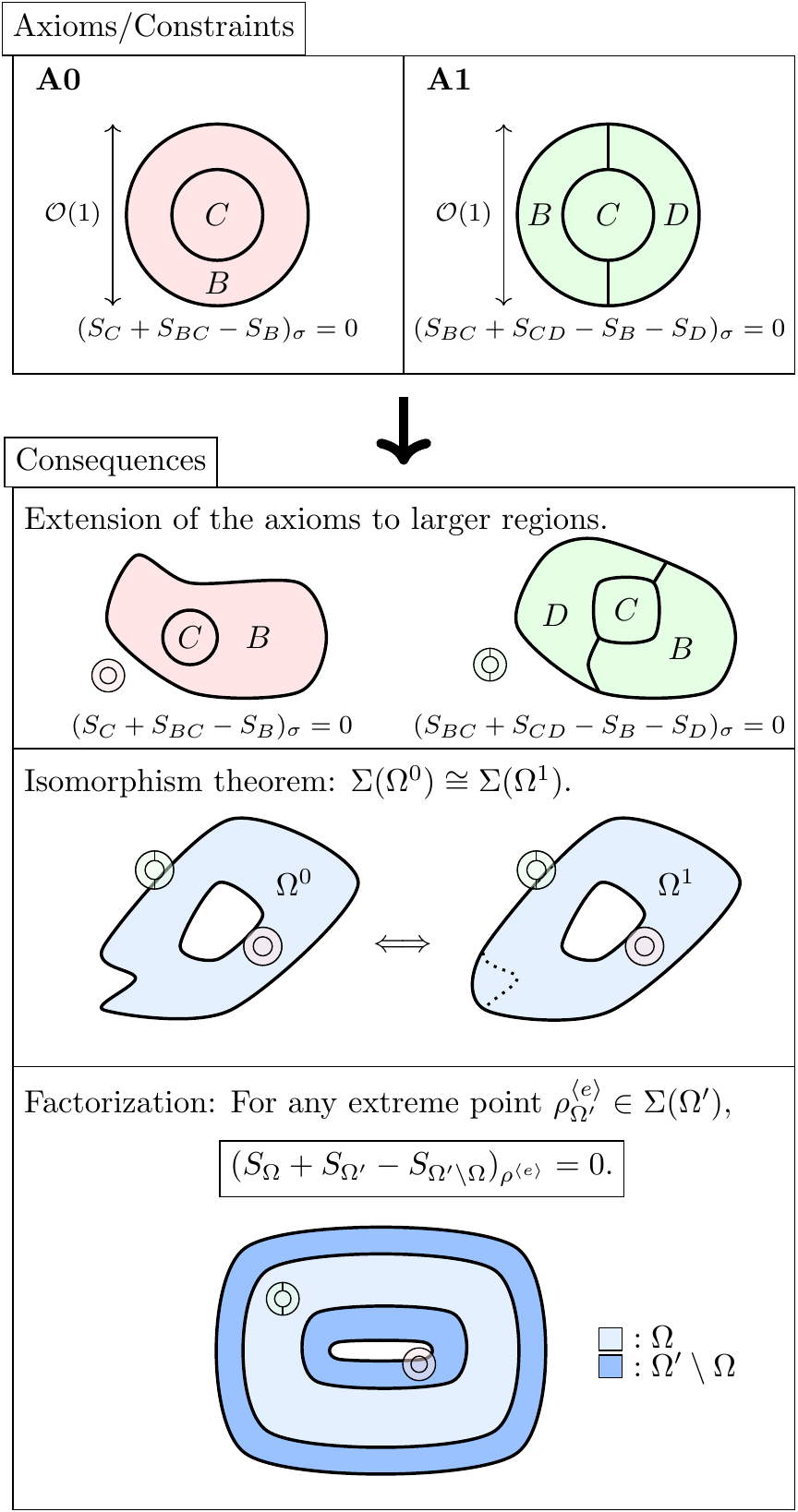}
	\caption{The axioms \textbf{A0} and \textbf{A1} of Ref.~\cite{SKK2019} and their consequences. These axioms, which are defined on a region of size $\mathcal{O}(1)$, imply that the same entropic conditions hold at larger length scales; see the first figure in the ``consequences.'' The isomorphism and the factorization property holds if the subsystems' thicknesses are larger than $2r$. Here, $r$ is the radius of the disks on which the axioms are imposed. While we only depicted annuli in this figure, the same consequences apply to \emph{any} sufficiently smooth subsystems. }
	\label{fig:bulk_axiom_to_fusion_data}
\end{figure}

Equation.~\eqref{eq:axioms_bulk} is a reasonable assumption because it follows from the well-known expression for the ground state entanglement entropy of gapped systems~\cite{Kitaev2006,Levin2006}:
\begin{equation}
    S(\sigma_A) = \alpha |\partial A| - \gamma + \ldots, \label{eq:tee}
\end{equation}
where $A$ is a simply connected subsystem, $\alpha$ is a non-universal constant, $\gamma$ is the topological entanglement entropy, and the ellipsis is the subleading term that vanishes in the $|\partial A| \to \infty$ limit.\footnote{While Eq.~\eqref{eq:axioms_bulk} must be assumed to hold exactly in Ref.~\cite{SKK2019}, we expect the arguments of the paper to go through even if we the conditions only hold approximately.} In the absence of subsystem symmetries~\cite{Williamson2019}, Eq.~\eqref{eq:tee} is expected to hold. Therefore, the fact that Eq.~\eqref{eq:axioms_bulk} follows from Eq.~\eqref{eq:tee} justifies the physical relevance of our axioms.

\begin{figure}[h]
  \centering
  \includegraphics[scale=0.95]{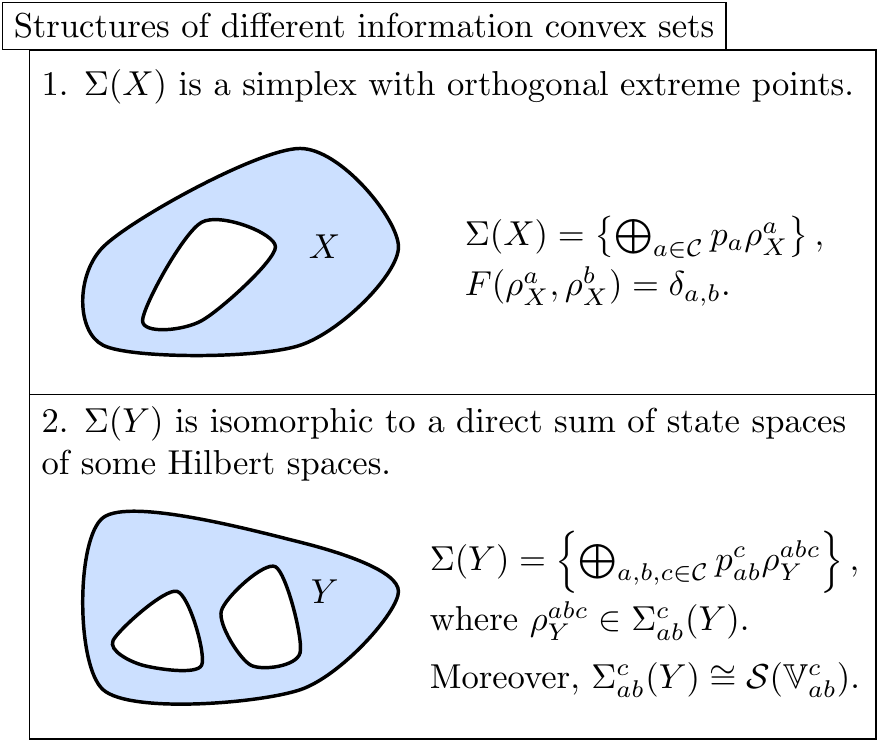}
	\caption{The technical consequences in Fig.~\ref{fig:bulk_axiom_to_fusion_data} lead to the structural statements about the information convex sets. The proof of the first statement is reproduced in Section~\ref{sec:intro_supsec_fusion}. For the proof of the second statement, see Ref.~\cite{SKK2019}. Here, $X$ is an annulus and $Y$ is a two-hole disk. $\mathcal{S}(\mathbb{V}_{ab}^c)$ is the set of density matrices acting on a finite dimensional Hilbert space  $\mathbb{V}_{ab}^c$.}
	\label{fig:iso_fac_to_fusion_data}
\end{figure}

These axioms lead to three important consequences, summarized in Fig.~\ref{fig:bulk_axiom_to_fusion_data}. We will focus on discussing their meanings, referring Ref.~\cite{SKK2019} for the proof.

The first consequence is that Eq.~\eqref{eq:axioms_bulk} holds at larger length scales. Recall that the axioms only apply to balls of bounded radius. The same set of constraints hold on arbitrarily large subsystems.

The second consequence is the isomorphism theorem.
\begin{theorem}
[Isomorphism theorem~\cite{SKK2019}]
 If $\Omega^0$ and $\Omega^1$ are connected by a path $\{\Omega^t\}_{t\in[0,1]}$, there is an isomorphism $\Phi$ between $\Sigma(\Omega^0)$ and $\Sigma(\Omega^1)$ uniquely determined by the path. Moreover, this isomorphism preserves the distance and entropy difference between two elements of the information convex sets: for any $\rho,\lambda \in \Sigma(\Omega^0)$,
\begin{equation}
\boxed{
\begin{aligned}
    D(\rho, \lambda) &= D(\Phi(\rho), \Phi(\lambda)) \\
    S(\rho) - S(\lambda) &= S(\Phi(\rho)) - S(\Phi(\lambda)),
\end{aligned}
}
\end{equation}
where $D(\cdot, \cdot)$ is any distance measure that is non-increasing under completely-positive trace preserving maps.
\label{thm:isomormphism_bulk}
\end{theorem}
\noindent
Here we say that a path exists between two subsystems if they can be smoothly deformed into each other without changing the topology of the subsystem.\footnote{In order to not run into any pathological counterexamples, it is convenient to only consider subsystems whose thicknesses are at least a few times larger compared to $r$.} This theorem implies that there are ``conserved quantities'' which remain invariant under deformations of the subsystems. These quantities include the distance between two states in the information convex set and their entropy difference. 

The third consequence concerns the factorization property of the extreme points. Let $\Omega$ be an arbitrary subsystem. Consider a subsystem $\Omega'\supset \Omega$ that can be smoothly deformed into $\Omega$, where  $\Omega'\setminus \Omega$ is a ``shell'' that covers the boundary of $\Omega$. We shall refer to $\Omega' \setminus \Omega$ as the \emph{thickened boundary} of $\Omega'$. This will be an important concept that will be used throughout this paper.
Let  $\rho^{\langle e\rangle}_{\Omega'} \in \Sigma(\Omega')$ be an extreme point. Then we have
\begin{equation}
\boxed{
    (S_{\Omega} + S_{\Omega'} - S_{\Omega' \setminus \Omega})_{\rho^{\langle e\rangle}}=0.
}    
    \label{eq:factorization_extreme_points}
\end{equation}
To see why we refer to Eq.~\eqref{eq:factorization_extreme_points} as the factorization property, consider a purification of $\rho^{\langle e\rangle}_{\Omega'}$, which we denote as $|\rho^{\langle e\rangle}\rangle_{\Omega' \Omega'^c}$ where $\Omega'^c$ is the purifying system of $\Omega'$. By using the fact that the von Neumann entropy of a state is equal to that of its purifying space, we can conclude
\begin{equation}
    \begin{aligned}
    I(\Omega: \Omega'^c)_{|\rho^{\langle e\rangle}\rangle_{\Omega' \Omega'^c}} &=  (S_{\Omega} + S_{\Omega'} - S_{\Omega' \setminus \Omega})_{\rho^{\langle e\rangle}} \\
    &=0,
    \end{aligned}
\end{equation}
where $I(A:B)_{\rho} := (S_A + S_B - S_{AB})_{\rho}$ is the mutual information between $A$ and $B$ over a state $\rho$. In other words, $|\rho^{\langle e\rangle}\rangle_{\Omega' \Omega'^c}$, upon tracing out $\Omega' \setminus \Omega$, becomes a product state over $\Omega$ and $\Omega'^c$.

These three consequences are the main workhorses of our theory. Below, we will see the power of these consequences in action, by deriving several nontrivial facts about anyon theory. We urge the readers to carefully digest the ensuing material before moving to the rest of the paper, as the key ideas remain the same while the setup becomes more intricate as we move forward.

\subsection{Superselection sectors}\label{sec:intro_supsec_fusion}
In the theory of anyon, a superselection sector is a topological charge associated with a point-like excitation. In our theory, we define the superselection sectors as the extreme points of an information convex set over an annulus. In this section, we justify this definition by showing that different extreme points are orthogonal to each other. In particular, different extreme points can be perfectly distinguished from each other by some physical experiment. 

To prove this fact, we set up the notation as follows. Consider an annulus $X$ and two additional annuli $X'$ and $X''$ such that $X' \supset X$ and $X'X''$ is again an annulus; see Fig.~\ref{fig:superselection_sketch}. Without loss of generality, consider two extreme points of $\Sigma(X)$, denoted as $\rho_X$ and $\rho'_X$.

The key idea is to map these extreme points to the extreme points in $\Sigma(XX'')$ by using Theorem~\ref{thm:isomormphism_bulk}. Then, we apply the factorization property of extreme points to argue that these extreme points must factorize over $X$ and $X''$. Our claim will be an immediate consequence of this last fact.

As a first step, note that any distance measure between $\rho_X$ and $\rho'_X$ is invariant under the isomorphism associated with a smooth deformation of $X$; see Theorem~\ref{thm:isomormphism_bulk}. In particular, the fidelity between the two satisfies the following identity
\begin{equation}
    F(\rho_{X}, \rho_{X}') = F(\rho_{X'X''}, \rho_{X' X''}'),
\end{equation}
where $\rho_{X'X''}$ and  $\rho_{X' X''}'$ are the extreme points of $\Sigma(X'X'')$ obtained from the isomorphism. Moreover, because fidelity is non-decreasing under partial trace, we conclude~\footnote{While fidelity is not a distance measure, one can relate it to a distance measured called Bures distance, defined as $\sqrt{1-F(\cdot, \cdot)}$.}
\begin{equation}
\begin{aligned}
    F(\rho_{X}, \rho_{X}') &= F(\rho_{X'X''}, \rho_{X'X''}') \\
    &\leq F(\rho_{X X''}, \rho_{X X''}'). \label{eq:fedlity_le_X}
\end{aligned}
\end{equation}

\begin{figure}[h]
  \centering
  
  \includegraphics[scale=1]{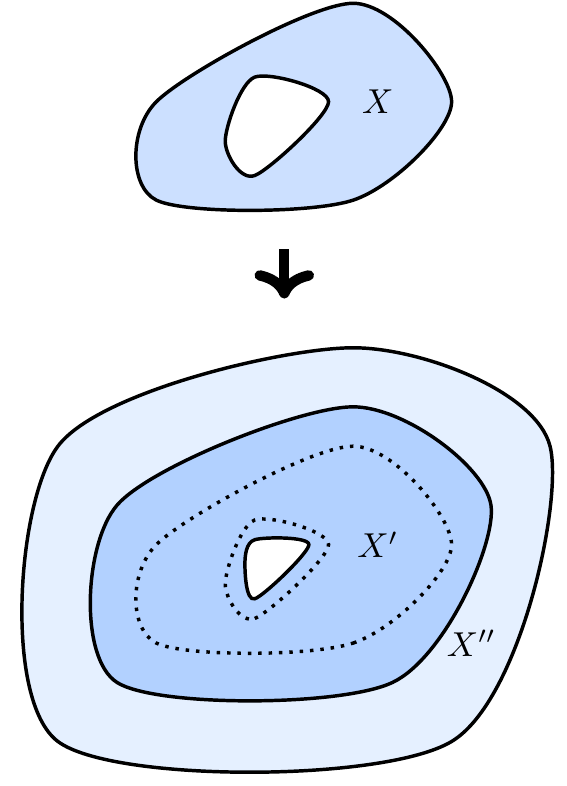}
	\caption{Subsystems involved in the proof of the orthogonality of extreme points in $\Sigma(X)$. The annulus enclosed in the dotted region of the second figure represents the blue annulus $X$ in the top figure. By applying the isomorphism theorem (Theorem~\ref{thm:isomormphism_bulk}), the annulus $X$ is deformed into a larger annulus $X'\supset X$ and the union of $X'$ (dark blue) with $X''$ (light blue).}
	\label{fig:superselection_sketch}
\end{figure}

Secondly, by the factorization of the extreme points, we have 
\begin{equation}
    (S_{X} + S_{X'} - S_{X'\setminus X})_{\rho} = 0.
\end{equation}
By using the strong subadditivity of entropy (SSA)~\cite{Lieb1973}, we get
\begin{equation}
\begin{aligned}
    I(X: X'')_{\rho} &\leq (S_{X} + S_{X'} - S_{X'\setminus X})_{\rho} \\
    &=0.
\end{aligned}
\end{equation}
Therefore, $\rho_{X'X''}$, upon restricting to $XX''$, becomes a product state over $X$ and $X''$. The same conclusion applies to $\rho'_{X'X''}$ because $\rho'$ is an extreme point too.

Combining these two observations, we conclude
\begin{equation}
\begin{aligned}
    F(\rho_{X}, \rho_{X}') \leq F(\rho_X, \rho_X') F(\rho_{X''}, \rho'_{X''}).
\end{aligned}
\end{equation}
Again using the isomorphism theorem (Theorem~\ref{thm:isomormphism_bulk}), we get
\begin{equation}
    F(\rho_{X}, \rho_{X}')\leq F(\rho_{X}, \rho_{X}')^2.
\end{equation} 
Since $F(\rho_{X}, \rho_{X}') \in [0,1]$ by the definition of fidelity, the only allowed values are $F(\rho_{X}, \rho_{X}')\in \{0, 1\}$. Therefore, given any two extreme points, they must be either orthogonal to each other or exactly equal to each other, thus proving our claim.

Therefore, without loss of generality, we can characterize $\Sigma(X)$ as a simplex with orthogonal extreme points.\footnote{Here, the orthogonality means that the Hilbert-Schmidt inner product of two-density matrices is $0$. Namely, $\Tr(\rho^{\dagger} \sigma)=0$.} Specifically, we have
\begin{equation}
   \Sigma(X) =\left\{ \bigoplus_{a} p_a \rho_X^a  : \sum_a p_a = 1, p_a\ge 0 \right\},
\end{equation}
where different extreme points $\{ \rho_X^a \}$ are supported on orthogonal subspaces. Provided that the underlying Hilbert space is finite-dimensional, $a$ belongs to a finite set 
\begin{equation}
  \calC=\{ 1, a, b, \cdots \}, 
\end{equation} 
where ``1'' is the vacuum sector, the extreme point of which is obtained by restricting $\sigma$ to $X$. 

To each of the extreme points, we can define a notion of \emph{quantum dimension}. Let $\rho_X^a$ be an extreme point of $\Sigma(X)$. We \emph{define} the quantum dimension of the superselection sector $a$ as
\begin{equation}
\boxed{
    d_a :=\exp \left(\frac{S(\rho_X^a) - S(\rho_X^1)}{2} \right), \label{eq:d_a_definition}
    }
\end{equation}
Note that, even though we have not specified the annulus here, this definition is still well-defined because of the isomorphism theorem (Theorem~\ref{thm:isomormphism_bulk}). It follows from this definition that $d_1=1$ and $d_a >0$ for all $a\in \mathcal{C}$. 

Our definition of the quantum dimension is not standard. However, this definition is equivalent to the more widely-held definition~\cite{SKK2019}. We prove this in Section~\ref{sec:intro_qd_fm} by showing that our quantum dimensions are completely determined by the fusion multiplicities, as is the case in the more conventional theory of anyon~\cite{Kitaev2006solo}.

\subsection{Fusion multiplicities}\label{sec:intro_qd_fm}
In this section, we derive the following fundamental equation
\begin{equation}
\boxed{
	d_ad_b = \sum_{c\in \mathcal{C}} N_{ab}^c d_c,
}\label{eq:fusion_identity_bulk}
\end{equation}
where $N_{ab}^c$ is the dimension of the fusion space and $d_a$ is defined in the previous section~\cite{Shi2018,SKK2019}. This equation implies that the quantum dimensions \emph{are} the quantum dimensions of anyons. 

\begin{figure}[h]
	\centering
	\includegraphics[scale=1]{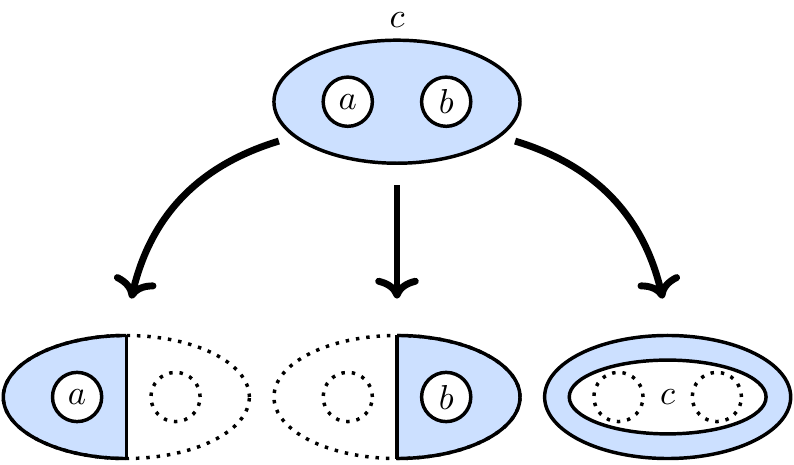}
	\caption{By restricting the elements of $\Sigma_{ab}^c(Y)$ to three different annuli, we obtain extreme points of $\Sigma(X)$ for some annulus $X$. These extreme points correspond to the superselection sectors $a, b,$ and $c$. }
	\label{fig:info_convex_two_hole}
\end{figure}

We begin by briefly explaining what we mean by a fusion space, deferring the proof of Eq.~\eqref{eq:fusion_identity_bulk} for the moment. The fusion space is defined in terms of the information convex set of a two-hole disk, say $Y$. Reference~\cite{SKK2019} completely characterized this set, under the same set of assumptions we have used so far. Specifically,
\begin{equation}
\Sigma(Y) =\left\{ \bigoplus_{a,b,c} p_{ab}^c \rho_Y^{abc}  : \rho_Y^{abc} \in \Sigma_{ab}^c(Y) \right\},
\end{equation}
where $\{p_{ab}^c\}$ is a probability distribution and $\Sigma_{ab}^c(Y)$ is a set of states in $\Sigma(Y)$ whose reduced density matrices on the three annuli are the extreme points associated with superselection sectors $a, b,$ and $c$; see Fig.~\ref{fig:info_convex_two_hole}. Importantly, $\Sigma_{ab}^c(Y)$ is isomorphic to the state space of some finite-dimensional Hilbert space. This is our definition of the fusion space. The dimension of this fusion space is $N_{ab}^c$.

Below, we focus on the derivation of Eq.~\eqref{eq:fusion_identity_bulk}, by first explaining the \emph{merging technique}, and then applying this technique to our setup. 

\subsubsection{Merging}
Equation~\eqref{eq:fusion_identity_bulk} follows from an extremely useful technique called \emph{merging}. The merging technique addresses the following problem. Let $\rho \in \Sigma(\Omega)$ and $\lambda \in \Sigma(\Omega')$ be the elements belonging to two different information convex sets. Can we construct a state in $\Sigma(\Omega \cup \Omega')$ that is consistent with both $\rho$ and $\lambda$? Obviously, this is not always possible because such a state may not even exist.\footnote{As a simple example, let $\rho$ be a maximally entangled state between two subsystems, say $A$ and $B$, and $\lambda$ be a maximally entangled state between $B$ and $C$. By the monogamy of entanglement, there cannot be any tripartite state over $ABC$ that is consistent with both $\rho$ and $\lambda$.} Moreover, even if there exists a state consistent with both $\rho$ and $\lambda$, that state may not belong to $\Sigma(\Omega \cup \Omega')$. With the merging technique, we can ensure both.

The following two statements are the key. First, for general quantum states, we have the following \emph{merging lemma.}
\begin{lemma}[Merging Lemma \cite{Kato2016}]\label{lemma:merging_lemma}
	If there is a pair of quantum states $\rho_{ABC}$ and $\lambda_{BCD}$ satisfying
	$\rho_{BC}=\lambda_{BC}$
	and $I(A:C\vert B)_{\rho}= I(B:D\vert C)_{\lambda}=0$,  there exists a unique quantum state $\tau_{ABCD}$ such that 
	\begin{eqnarray}
	\Tr_D \tau_{ABCD} &=& \rho_{ABC} \nonumber\\
	\Tr_A \tau_{ABCD} &=& \lambda_{BCD} \nonumber\\
	I(A: CD\vert B)_{\tau} &=& I(AB: D\vert C)_{\tau}=0. \nonumber
	\end{eqnarray}
\end{lemma}
Here $I(A:C\vert B)_{\rho}:= (S_{AB} + S_{BC} - S_B - S_{ABC})_{\rho}$ is the \emph{conditional mutual information}. 

Second, with an additional assumption, elements of the information convex sets are ``closed'' under the merging operation. Specifically, the density matrices belonging to information convex sets can merge into an element of another information convex set. We refer to this fact as the \emph{merging theorem}.
\begin{theorem}[Merging Theorem \cite{SKK2019}]\label{thm:merging_info_convex_set}
	Consider two density matrices  $\rho_{ABC}\in {\Sigma}(ABC)$ and  $\lambda_{BCD}\in {\Sigma}(BCD)$. Consider the following three conditions:
\begin{enumerate}
	\item $\rho_{BC}= \lambda_{BC}$ and  $I(A:C\vert B)_{\rho}= I(B:D\vert C)_{\lambda}=0$.
	\item There exists a partition $B'C'=BC$, such that no disk of radius $r$ overlaps with both $AB'$ and $CD$.
	\item $I(A:C'\vert B')_{\rho}= I(B':D\vert C')_{\lambda}=0$.
	\end{enumerate} 
If these three conditions hold, the resulting density matrix generated by $\rho_{ABC}$ and $\lambda_{BCD}$ using the merging lemma (Lemma~\ref{lemma:merging_lemma}) belongs to ${\Sigma(ABCD)}$.
\end{theorem}
\noindent

In this paper, to ensure the merged state is in an information convex set, we shall exclusively use the merging theorem. If the conditions in the merging theorem are satisfied, we shall denote the \emph{merged state} of $\rho$ and $\sigma$ as:
\begin{equation}
\rho \merge \lambda.
\end{equation}

Merged states are useful because they satisfy the following nontrivial identities: 
\begin{equation}
\boxed{
	\begin{aligned}
	I(A:CD|B)_{\rho \merge \lambda} &= I(AB:D|C)_{\rho \merge \lambda} = 0,\\
	I(A:C'D|B')_{\rho \merge \lambda} &= I(AB':D|C')_{\rho \merge \lambda} = 0, \\
	I(A:D|BC)_{\rho\merge \lambda} &=0,
	\end{aligned}
}\label{eq:merge_identities1}
\end{equation}
which implies that $\rho \merge \lambda$ is the \emph{maximum-entropy state} consistent with both $\rho$ and $\lambda$. This fact follows from SSA~\cite{Lieb1973}. Moreover,
\begin{equation}
\boxed{
	S(\rho \merge \lambda) - S(\rho' \merge \lambda) = S(\rho) - S(\rho'),
}\label{eq:merge_identities2}
\end{equation}
where we implicitly assumed that both $\rho$ and $\rho'$ can be merged with $\lambda$. 

To explain the utility of the merging theorem, we discuss a simple example. We will discuss merging two density matrices in a toy setup, focusing on the logic behind why they can be merged. 

Consider an annulus $ABC$ and a disk-like region $BCD$ that overlap on a disk-like region; see Fig.~\ref{fig:merge_example}. We will consider density matrices $\rho\in \Sigma(ABC)$ and $\lambda \in \Sigma(BCD)$.\footnote{In fact, $\lambda_{BCD}=\sigma_{BCD}$ on disk $BCD$. This is because any element of the information convex set is indistinguishable with the reference state on any disk~\cite{SKK2019}.} These two density matrices have identical reduced density matrix on disk $BC$. Provided that the overlapping region is sufficiently thick so that the distance between $A$ and $D$ is large, the requisite conditions in the merging theorem (Theorem~\ref{thm:merging_info_convex_set}) can be satisfied with an appropriate choice of $B$ and $C$. 

The first condition can be satisfied by partitioning the overlapping region as in Fig.~\ref{fig:merge_example}. We can use SSA and the extensions of axioms to derive the conditional independence condition. For example, in order to prove $I(A:C|B)_{\rho}=0$, consider an auxiliary subsystem $E$ introduced in Fig.~\ref{fig:conditional_independence_annulus}. By using the isomorphism theorem, one can extend $\rho$ to a density matrix  $\rho'$ in $\Sigma(ABCE)$. Such density matrix on a disk-like region $BCE$ is consistent with the reference state $\sigma$. By the extension of axiom, one can thus see that
\begin{equation}
\begin{aligned}
I(A:C|B)_{\rho} &\leq (S_{BC} +S_{CE} - S_B - S_E)_{\rho'} \\
&=0.
\end{aligned}
\end{equation}
Conditional independence of other sets of subsystems can be obtained in a similar way. 

\begin{figure}[h]
  \centering
      \includegraphics[scale=1]{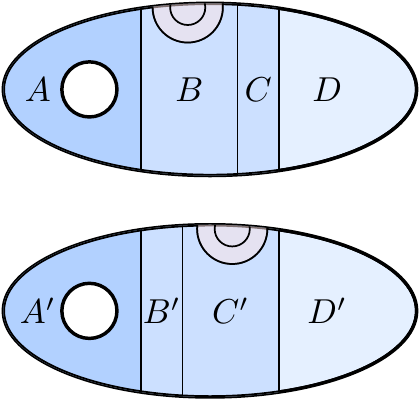}
	\caption{Two ways of partitioning the overlapping subsystem ($BC$). Because the overlapping region was chosen so that $A$ and $D$ are sufficiently far apart from each other, we can choose $B'$ and $C$ to be separated by more than $2r$, where $r$ is the radius of the disk on which our axioms are postulated (red).}
	\label{fig:merge_example}
\end{figure}

\begin{figure}[h]
  \centering
        \includegraphics[scale=1]{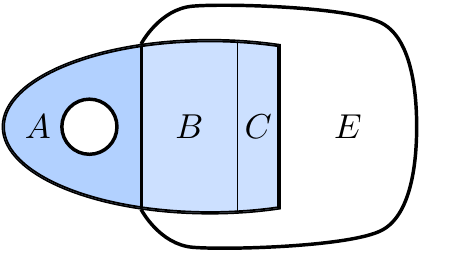}
	\caption{A partition of an annulus (blue) into $A,B,$ and $C$. Here, $E$ is an auxiliary subsystem used in the proof of the vanishing conditional mutual information. }
	\label{fig:conditional_independence_annulus}
\end{figure}

Within this example, the conditions in Theorem~\ref{thm:merging_info_convex_set} can be satisfied because the overlapping region separates the non-overlapping parts sufficiently far apart from each other. This observation holds quite generally, as we shall repeatedly see throughout this paper.

\subsubsection{Derivation}
Armed with the merging technique, we are now in a position to derive Eq.~\eqref{eq:fusion_identity_bulk}. To do so, we will merge two density matrices with supports overlapping on a disk-like region. Partition these annuli into $ABC$ and $BCD$, similar to the partition we had before; see Fig.~\ref{fig:merging_two_annuli}.

\begin{figure}[h]
  \centering
          \includegraphics[scale=1]{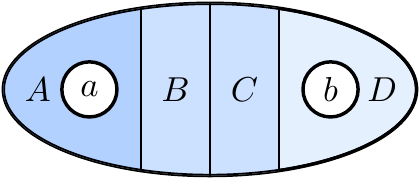}
	\caption{The set of subsystems used in merging the extreme points of $\Sigma(ABC)$ and $\Sigma(BCD)$. Here, $a$ and $b$ are the superselection sectors that the respective extreme points represent. While we also require $BC$ to be partitioned into $B'C'$ such that the requisite conditions in Theorem~\ref{thm:merging_info_convex_set} are satisfied, once these conditions are verified, we will not make use of this partition. This is why we did not describe these subsystems in this figure. }
	\label{fig:merging_two_annuli}
\end{figure}

We can merge the density matrices in $\Sigma(ABC)$ with the density matrices in $\Sigma(BCD)$ provided that the distance between $A$ and $D$ is sufficiently large. To see why, first note that these density matrices have identical density matrices on the overlapping region. Secondly, one can prove the requisite conditions on the conditional mutual information, again by utilizing the auxiliary subsystem introduced in Fig.~\ref{fig:conditional_independence_annulus}.

While one can merge any pair of density matrices from $\Sigma(ABC)$ and $\Sigma(BCD)$, we will merge the extreme points. Without loss of generality, let $\rho_{ABC}^a \in \Sigma(ABC)$ and $\lambda_{BCD}^b \in \Sigma(BCD)$ be a pair of extreme points associated with the superselection sectors $a$ and $b$. The merged state, 
\begin{equation}
\tau_{ABCD}^{a\merge b} := \rho^a_{ABC} \merge \lambda^b_{BCD},
\end{equation}
according to Eq.~\eqref{eq:merge_identities1}, obeys the following identity:
\begin{equation}
\begin{aligned}
(S_{ABCD})_{\tau^{a\merge b}} &= (S_{ABC} + S_{BCD} - S_{BC})_{\tau^{a\merge b}} \\
&=2\ln (d_a d_b) + (S_{ABC} + S_{BCD} - S_{BC})_{\sigma}
\end{aligned}
\end{equation}
In the first line, we used the property of the merged state. In the second line, we used the definition of the quantum dimensions. The second term of the second line can be interpreted as the entropy of the merged state $\tau_{ABCD}^{1\merge 1}$, which is equal to the reference state restricted to $ABCD$.\footnote{This is a fact proved in Ref.~\cite{SKK2019}.} Therefore,
\begin{equation}
S(\tau_{ABCD}^{a\merge b}) = S(\sigma_{ABCD}) + 2\ln (d_a d_b). \label{eq:1st_calculation}
\end{equation}

Note that $\tau_{ABCD}^{a\merge b}$ is the maximum-entropy state of $\Sigma(ABCD)$ that is consistent with the density matrices of the two annuli. We can solve this maximization problem directly, which, by definition, must agree with Eq.~\eqref{eq:1st_calculation}.

For this derivation, we use the structure of the information convex set of a two-hole disk summarized at the bottom of Fig.~\ref{fig:iso_fac_to_fusion_data}.\footnote{The proof of this statement also follows from the axioms in Fig.~\ref{fig:bulk_axiom_to_fusion_data}; see Ref.~\cite{SKK2019} for more detail.} We shall refer to this two-hole disk as $Y:=ABCD$. Because $\tau_{ABCD}^{a\merge b} \in \Sigma(Y)$, without loss of generality, 
\begin{equation}
\tau_{Y}^{a\merge b} = \bigoplus_c p_c \, \rho^{abc}_Y,\quad \rho^{abc}_Y \in \Sigma_{ab}^c(Y), \label{eq:tau_form}
\end{equation}
for some probability distribution $\{ p_c \}$. Density matrices in different $\Sigma_{ab}^c(Y)$ are mutually orthogonal to each other. Because $\tau_{Y}^{a\merge b}$ maximizes the entropy, its entropy is 
\begin{equation}
\begin{aligned}
S(\tau_{Y}^{a\merge b}) &= \max_{\substack{\{p_c\},\\ \{\rho^{abc}_Y \}}}\left(H\left(\{p_c\}\right) + \sum_c p_c S(\rho^{abc}_Y) \right) \\
&= \max_{\{p_c\}}\left(H\left(\{ p_c\}\right) + \sum_c p_c \max_{\rho^{abc}_Y}S(\rho^{abc}_Y)\right) 
\end{aligned} \label{eq:maximizing_entropy_ab}
\end{equation}
where $H(\{ p_c\}) = -\sum p_c \ln p_c$ is the Shannon entropy of the probability distribution $\{p_c \}$.

This is the key identity:
\begin{equation}
\max_{\rho^{abc}_Y \in \Sigma_{ab}^c(Y)} S(\rho^{abc}_Y) = S(\sigma_Y) + \ln N_{ab}^c + \ln (d_ad_b d_c), \label{eq:entropy of center abc}
\end{equation}
which we derive in two steps. First, we show that the extreme points within $\Sigma_{ab}^c(Y)$ have identical entropies. Because this space is isomorphic to the state space of a $N_{ab}^c$-dimensional Hilbert space, the maximum is attained by taking $\rho_Y^{abc}$ to be a uniform mixture of $N_{ab}^c$ orthogonal extreme points within $\Sigma_{ab}^c(Y)$. Second, we show that the entropy of the extreme points are equal to $S(\sigma_Y) + \ln (d_ad_bd_c)$.

For the first step, we use the factorization property of the extreme points. Recall that any extreme point $\rho_{Y'}^{\langle e\rangle} \in \Sigma_{ab}^c(Y')$ satisfies 
\begin{equation}
(S_Y + S_{Y'} - S_{Y'\setminus Y})_{\rho^{\langle e\rangle}} = 0,
\end{equation}
where $Y' \supset Y$ is a two-hole disk that is expanded along the boundary of $Y$ by an amount large compared to the correlation length. By using the fact that the reduced density matrices of the elements in $\Sigma_{ab}^c(Y')$ on $Y'\setminus Y$ are identical and the fact that the entropy difference over $Y$ and $Y'$ are equal, we can conclude that the entropy of extreme points are identical.

In the second step, we seek to prove
\begin{equation}
S(\rho^{\langle e\rangle abc}_Y) = S(\sigma_Y) + \ln (d_ad_bd_c) \label{eq:2-hole_quantum_dim}
\end{equation}
for any extreme point $\rho^{\langle e\rangle abc}_Y \in \Sigma_{ab}^c(Y)$. We can derive this fact by comparing the entropy of $\rho_Y^{\langle e\rangle abc}$ to $\sigma_Y$. More specifically, again consider $Y'\supset Y$ which is obtained by enlarging $Y$ along its boundary by a large enough amount.\footnote{The thickness of $Y'\setminus Y$ should be greater than $2r$ so that the simplex theorem applies to the three annuli subsystems of $Y'\setminus Y$.} The extreme points of $\Sigma_{ab}^c(Y)$ are extended into the extreme points of $\Sigma_{ab}^c(Y')$ by using the isomorphism theorem. By the factorization of the extreme points (Eq.~\eqref{eq:factorization_extreme_points}), we obtain:
\begin{equation}
\begin{aligned}
S(\rho^{\langle e\rangle abc}_Y ) + S(\rho^{\langle e\rangle abc}_{Y'} ) - S(\rho^{\langle e\rangle abc}_{Y'\setminus Y} ) &= 0, \\
S( \sigma_Y) + S( \sigma_{Y'} ) - S( \sigma_{Y'\setminus Y} ) &= 0,
\end{aligned}
\label{eq:shell_entanglement}
\end{equation}
where in the second line we used the fact that $\sigma_Y$ is an extreme point of $\Sigma_{11}^1(Y)$.\footnote{This is a fact discussed in Ref.~\cite{SKK2019}. Our axioms imply that $\Sigma_{11}^1(Y)$ contains a unique element. Since $\sigma_Y$ trivially belongs to $\Sigma_{11}^1(Y)$, it must be an extreme point.} By subtracting the second equation from the first, we obtain
\begin{equation}
( S_Y + S_{Y'})_{\rho^{\langle e\rangle abc}} - ( S_Y + S_{Y'})_{\sigma} = 2 \ln (d_a d_b d_c). \label{eq:temp28}
\end{equation}
The nontrivial part lies on obtaining the right hand side of Eq.~\eqref{eq:temp28}. This expression can be derived by noting that $Y'\setminus Y$ is a union of three annuli and the fact that the reduced density matrix of any element of $\Sigma_{ab}^c(Y')$ over $Y'\setminus Y$ is a tensor product of the extreme points associated with definite superselection sectors $a$, $b$, and $c$.\footnote{Technically speaking, one also needs to prove this fact. We gloss over this subtlety here, referring the readers to Ref.~\cite{SKK2019} for the rigorous proof.} The first two terms in the left-hand side of Eq.~\eqref{eq:temp28} are actually equal to each other due to the isomorphism theorem. Thus, we have proved Eq.~(\ref{eq:2-hole_quantum_dim}).

Plugging Eq.~(\ref{eq:entropy of center abc}) into Eq.~(\ref{eq:maximizing_entropy_ab}) we obtain
\begin{equation}
\begin{aligned}
S(\tau_{Y}^{a\merge b}) - S(\sigma_Y) =
& \max_{\{ p_c\}}\left( H\left(\{p_c \}\right) + \sum_c p_c \ln (N_{ab}^c d_c) \right)\\
& +\ln(d_a d_b).
\end{aligned}
\end{equation}
This maximization problem can be solved by minimizing the free energy of a fictitious Hamiltonian $H(c):=-\ln (N_{ab}^c d_c)$ that depends on the superselection sector $c$ with respect to a ``temperature'' of $T=1$. The partition function of this fictitious Hamiltonian is $Z=\sum_c N_{ab}^c d_c$. Therefore, the free energy, defined as $F= \sum_c p_c H(c) - T H(\{p_c\})$, is minimized as $\min_{\{p_c\}} F=- T\ln Z= -\ln (\sum_c N_{ab}^c d_c)$. The minimum is obtained when $p_c = \frac{N_{ab}^c d_c}{d_a d_b}$.
\footnote{A mathematically equivalent fact is that relative (Shannon) entropy is nonnegative. For two probability distributions $\{p_i\}$ and $\{q_i\}$, $\min_{ \{p_i\} } \left(\sum_i  p_i \ln(p_i/q_i) \right)= 0$. The minimum is obtained if and only if the two probability distributions are identical.}  Therefore, we obtain
\begin{equation}
S(\tau_{Y}^{a\merge b})  = S(\sigma_{ABCD}) + \ln(d_ad_b) + \ln(\sum_{c}N_{ab}^c d_c). \label{eq:2nd_calculation}
\end{equation}

By comparing Eq.~(\ref{eq:1st_calculation}) to Eq.~(\ref{eq:2nd_calculation}), we conclude that
\begin{equation}
\ln(d_ad_b) + \ln(\sum_{c}N_{ab}^c d_c) = 2 \ln(d_ad_b),
\end{equation}
which, after rearranging the terms, becomes Eq.~\eqref{eq:fusion_identity_bulk}.

To conclude, we have sketched the proof of Eq.~\eqref{eq:fusion_identity_bulk}. The key idea was to merge density matrices associated with two superselection sectors. The entropy of the merged state can be calculated in two different ways, one that is obtained from the entropy of the reduced density matrices over smaller regions and another obtained by directly maximizing the entropy. The gist of the second calculation follows from Eq.~\eqref{eq:factorization_extreme_points}.

The ideas sketched above are, in fact, powerful enough to derive a whole slew of consistency relations, such as~\cite{SKK2019}
\begin{equation}
\begin{aligned}
N_{ab}^c &= N_{ba}^c \\
N_{1a}^c &= N_{a1}^c = \delta_{a,c} \\
\forall a, \exists !\, \bar{a} \,\,\text{ s.t. }\,\, N_{ab}^1 &= \delta_{b,\bar{a}} \\
N_{ab}^c &= N_{\bar{b}\bar{a}}^{\bar{c}} \\
\sum_i N_{ab}^i N_{ic}^d &= \sum_j N_{aj}^d N_{bc}^j.
\end{aligned}
\label{eq:axioms_bulk_ch2}
\end{equation} 
A generalization of these identities will be discussed in Section~\ref{sec:fusion_rules}.

\section{Gapped domain walls: a tale of parton sector} \label{sec:parton}

The results we sketched in Section~\ref{sec:fusion_from_entanglement} follow from the axioms described in Fig.~\ref{fig:bulk_axiom_to_fusion_data}. However, these assumptions become inadequate in the vicinity of gapped domain walls. On a domain wall, we need to relax these assumptions appropriately.

\begin{figure}[h]
	\centering
          \includegraphics[scale=1]{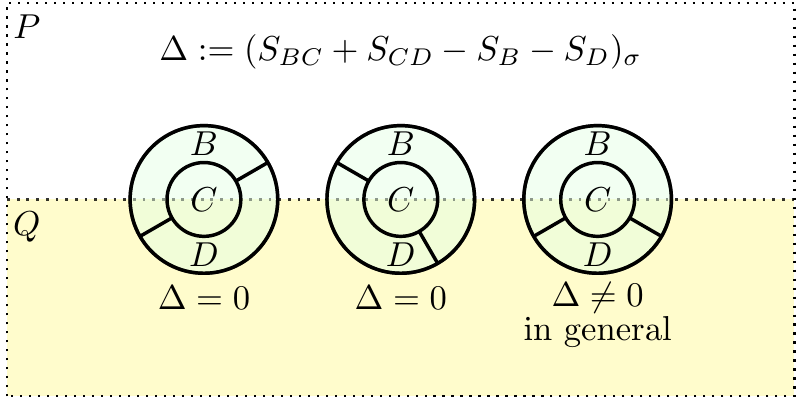}
	\caption{If we use the ansatz for the entanglement entropy in Eq.~\eqref{eq:ee_ansatz}, the linear combination of entanglement entropy $(S_{BC} + S_{CD} - S_B - S_D)_{\sigma}$ becomes $0$ for the first two choices. However, the same conclusion does not generally hold for the third choice.}
	\label{fig:entorpy_condition_wall2}
\end{figure}

Here is a heuristic discussion on this issue. Consider two topologically ordered mediums that are separated by a gapped domain wall. We shall refer to the bulk phases lying on different sides of the domain wall as $P$ and $Q$.  Suppose that the entanglement entropy of a region $A$ has the following form:
\begin{equation}
S(\sigma_A) = \alpha |\partial A | - \gamma(A)+\ldots, \label{eq:ee_ansatz}
\end{equation}
where the first term is the leading area law term that can be canceled from an appropriate linear combination, the second term $\gamma(A)$ is a constant that depends on $A$, and the ellipses represent the subleading correction that vanishes in the $|\partial A| \to \infty$ limit. Based on the study of entanglement entropy in the bulk~\cite{SKK2019}, we can make a somewhat speculative but reasonable assumption about $\gamma(A)$: that it is invariant under smooth deformations of $A$.\footnote{Unlike the subleading contribution in the bulk~\cite{Kitaev2006,Levin2006}, it is unclear if $\gamma(A)$ can be always obtained from a linear combination of entanglement entropies.  Therefore, it is unclear whether individual $\gamma(A)$ has a well-defined physical meaning. However, as we shall see in Section~\ref{sec:tee}, certain linear combinations of entanglement entropies do have clear physical meanings.}  By smooth deformation, we mean any deformation that retains the topology of $A$ and its restrictions to $P$ and $Q$. 

Once we accept this hypothesis, we can immediately verify that $(S_{BC} + S_{CD} - S_B - S_D)_{\sigma} = 0$ for the first two choices of subsystems described in Fig.~\ref{fig:entorpy_condition_wall2}. This is because $\gamma(BC)=\gamma(CD)=\gamma(B)=\gamma(D)$ for these subsystem choices. However, this hypothesis does not imply that the same linear combination of entanglement entropy vanishes for the third choice. For the third choice, $\gamma(BC)$ and $\gamma(B)$ are allowed to take different values since $B$ cannot be smoothly deformed into $BC$. 

This observation motivates a relaxed set of axioms to study gapped domain walls, summarized in  Fig.~\ref{fig:entropy_condition_wall}. We emphasize that the boundary between $B$ and $D$ can deform arbitrarily so long as they do not cross the domain wall. We will not make any assumption about the value of $(S_{BC} + S_{CD} - S_B - S_D)_{\sigma}$ for the rightmost subsystems in Fig.~\ref{fig:entorpy_condition_wall2}. Remarkably, its value is highly constrained, as we explain in Section~\ref{sec:tee}.

\begin{figure}[h]
  \centering
  \includegraphics[scale=1]{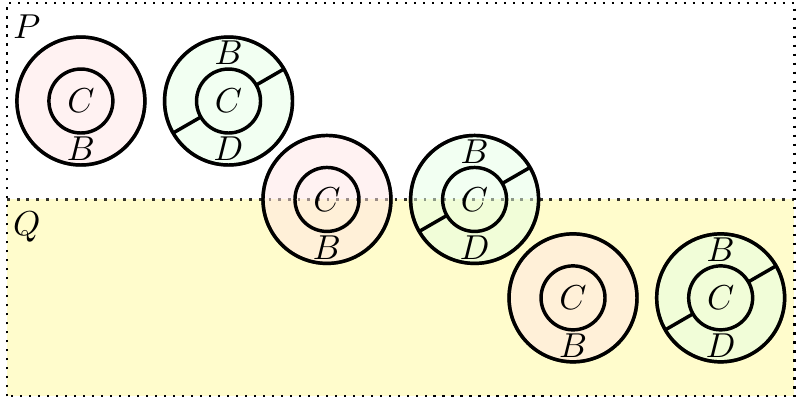}
    \caption{On top of the bulk axioms in $P$ and $Q$, we assume that $(S_C + S_{BC} - S_B)_{\sigma}=0$ (red) and $(S_{BC} + S_{CD} - S_B - S_D)_{\sigma}=0$ (green) on the domain wall. The subsystems are allowed to be deformed as long as the boundaries between $B$ and $D$ do not cross the domain wall.}
    \label{fig:entropy_condition_wall}
\end{figure}

The new axioms in Fig.~\ref{fig:entropy_condition_wall} directly lead to a definition of \emph{parton sectors.} This is a new type of superselection sector that is localized on either side of the domain wall. We refer to these sectors as parton sectors because they subdivide the known superselection sectors of point-like excitations on the domain wall \cite{KitaevKong2012,Kong2014}. As in the discussion in Section~\ref{sec:intro_supsec_fusion}, the properties of the parton sectors can be derived from three important consequences: extension of axioms, isomorphism theorem, and factorization of extreme points.

Let us formally state these consequences below, deferring the proofs to Appendixes~\ref{appendix:extensions_of_axioms}, \ref{appendix:isomorphism}, and \ref{appendix:extreme_points_factorization}. First, the axioms can be extended to arbitrarily large regions. Secondly, a generalization of the isomorphism theorem holds. 
\begin{theorem}[Isomorphism Theorem]
Consider a reference state for which the axioms in Fig.~\ref{fig:entropy_condition_wall} apply.
 If $\Omega^0$ and $\Omega^1$ are connected by a path $\{\Omega^t\}_{t\in [0,1]}$, there is an isomorphism $\Phi$ between $\Sigma(\Omega^0)$ and $\Sigma(\Omega^1)$ uniquely determined by the path. Moreover, this isomorphism preserves the distance and entropy difference between two elements of the information convex sets: for $\rho,\lambda\in \Sigma(\Omega^0)$,
\begin{equation}
\boxed{
\begin{aligned}
    D(\rho, \lambda) &= D(\Phi(\rho), \Phi(\lambda)), \\
    S(\rho) - S(\lambda) &= S(\Phi(\rho)) - S(\Phi(\lambda)),
\end{aligned}
}
\end{equation}
where $D(\cdot, \cdot)$ is any distance measure that is non-increasing under completely-positive trace preserving maps.
\label{thm:isomorphism}
\end{theorem}
\noindent
In order to understand this theorem, it is important to understand what a modified definition of the ``path'' means in the presence of a domain wall. This is most convenient to understand in the continuum limit. We say that two subsystems are connected by a path if they can be continuously deformed into each other. Specifically, let $\mathcal{M}$ be the manifold, which is divided into $\mathcal{M}= \mathcal{M}_P \cup \mathcal{M}_Q$, where $\mathcal{M}_P$ is the part that hosts the topological phase $P$ and $\mathcal{M}_Q$ is the part that hosts the phase $Q$. We say that there is a \emph{path} between $\Omega^0 \subset \mathcal{M}$ and $\Omega^1\subset \mathcal{M}$ if there is a one-parameter family of homeomorphism $\phi_t: \Omega^0 \hookrightarrow \mathcal{M}$ such that $\phi_t(\Omega^0)$ restricted to $P$ and $Q$ are both homeomorphisms for $t\in [0,1]$, $\phi_0(\Omega^0)=\Omega^0$, and $\phi_1(\Omega^0)=\Omega^1$; see Fig.~\ref{fig:path_connected_wall}. Note that this definition of path is a refinement of that in the bulk.

\begin{figure}[h]
  \centering
    \includegraphics[scale=1]{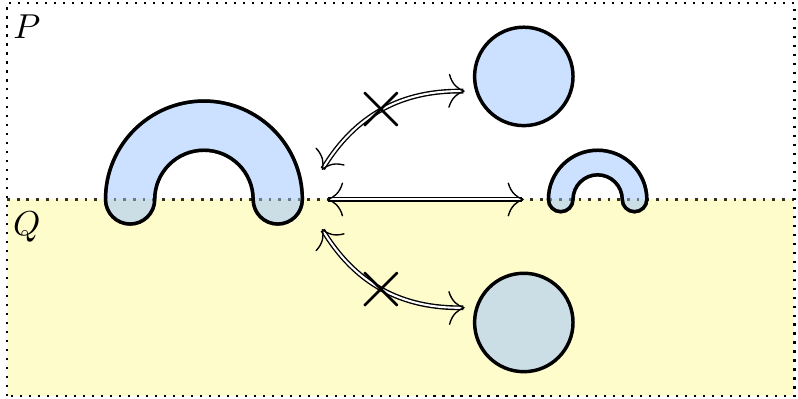}
    \caption{The subsystem on the left side is homeomorphic to a disk. However, it is not connected to a disk in $P$ (or $Q$) by any path because its part that lies on $Q$ is a union of two disks. There is no homeomorphism that maps those disks to either a single disk or an empty set.}
    \label{fig:path_connected_wall}
\end{figure}

Third, we can extend Eq.~\eqref{eq:factorization_extreme_points} to the gapped domain wall. Specifically, the exact same equation holds, irrespective of the presence of the domain wall. We restate this fact here for readers' convenience. Given an extreme point of an information convex set over $\Omega$, let $\Omega'\supset \Omega$ be a region obtained from $\Omega$ by enlarging\footnote{An enlargement is associated with a path connecting $\Omega$ and $\Omega'$. By the isomorphism theorem, there is an isomorphism between $\Sigma(\Omega)$ and $\Sigma(\Omega')$.} $\Omega$ along the boundary, such that $\Omega'\setminus \Omega$ is the thickened boundary of $\Omega'$. Then, for any extreme point $\rho^{\langle e\rangle}_{\Omega'} \in \Sigma(\Omega')$, we have:
\begin{equation}
\boxed{
    (S_{\Omega} + S_{\Omega'} - S_{\Omega' \setminus \Omega})_{\rho^{\langle e\rangle}}=0.
}
\end{equation}

\subsection{Parton sectors}
Now we are ready to define the parton sectors, the fundamental objects of our theory. To define these sectors, we choose the subsystems described in Fig.~\ref{fig:n_and_u}. We will refer to the left diagram as an $N$-shaped region and the right diagram as a $U$-shaped region. Their information convex sets form simplices with orthogonal extreme points, each of which labels a parton sector.

There are two types of parton sectors, one associated with the $N$-shaped region and the one associated with the $U$-shaped region. We shall refer to the former as an $N$-type superselection sector and the latter as a $U$-type superselection sector to evoke the shape of the underlying regions. 

Let us first explain why these sectors are well-defined. We focus only on the $N$-type superselection sector. Because the same argument applies to the $U$-type superselection sector, we omit that discussion.

\begin{figure}[h]
  \centering
      \includegraphics[scale=1]{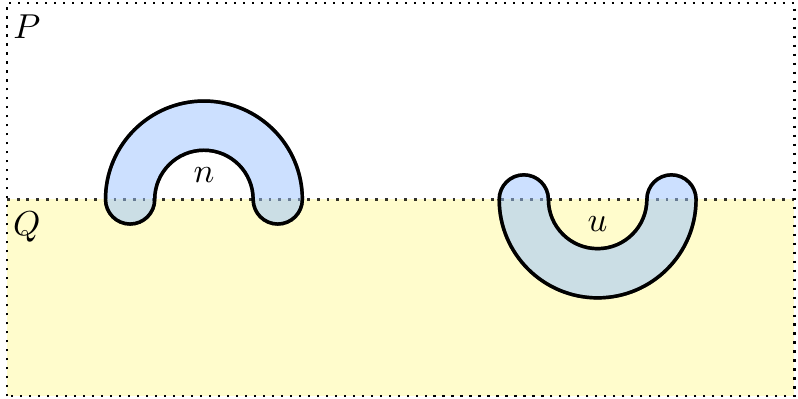}
    \caption{Subsystems associated with the $N$-type (left) and $U$-type (right) superselection sectors. We use $n$ and $u$ denote the respective sectors.}
    \label{fig:n_and_u}
\end{figure}

Our proof is based on the following choice of subsystems, which we depict in Fig.~\ref{fig:simplex_n_sector}. Without loss of generality, consider an $N$-shaped region $N$. Let $N' \supset N$ be a subsystem obtained by enlarging $N$ along its boundary. Let $N''$ be another $N$-shaped region disjoint from $N'$ such that $N'N''$ is again an $N$-shaped region. 

\begin{figure}[h]
  \centering
        \includegraphics[scale=1]{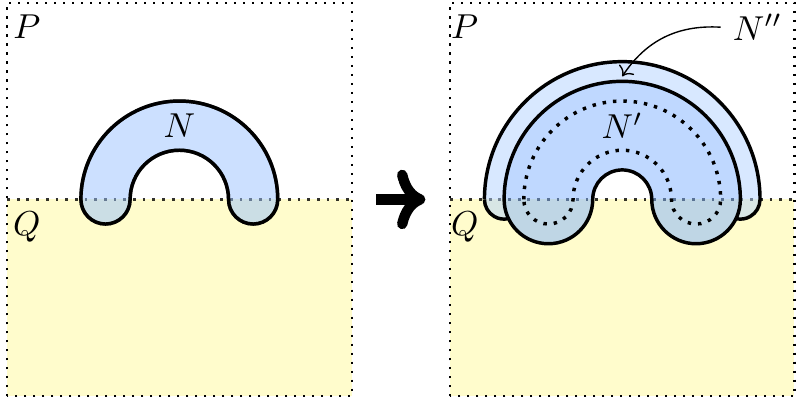}
    \caption{Subsystems involved in the proof of the orthogonality of the extreme points in $\Sigma(N)$. Here, $N'\supset N$ (dark blue) is obtained from $N$ by enlarging $N$ along the boundary. $N''$ (light blue) is chosen in such a way that $N'N''$ is again a \emph{N}-shaped subsystem. On the right, the dotted region represents $N$.}
    \label{fig:simplex_n_sector}
\end{figure}

Consider a pair of extreme points $\rho_N, \rho'_N \in \Sigma(N)$. As discussed in Section~\ref{sec:intro_supsec_fusion}, the key idea is to map these extreme points to the extreme points in $\Sigma(NN'')$ by using the isomorphism theorem (Theorem~\ref{thm:isomorphism}). These extreme points must be factorized over $N$ and $N''$, which immediately implies our claim.

Let $\rho_{N'N''}$ and $\rho_{N'N''}'$ be the extreme points of $\Sigma(N'N'')$ obtained by applying the isomorphism theorem (Theorem~\ref{thm:isomorphism}) between $\Sigma(N)$ and $\Sigma(N'N'')$ to $\rho_N$ and $\rho_N'$ respectively. Their fidelity is equal to the fidelity between $\rho_N$ and $\rho_N'$.
\begin{equation}
    F(\rho_N, \rho_{N}') = F(\rho_{N'N''}, \rho_{N'N''}').
\end{equation}
Because fidelity is non-decreasing under partial trace, we have
\begin{equation}
\begin{aligned}
F(\rho_{N},\rho_{N}') &= F(\rho_{N'N''}, \rho_{N'N''}')\\
&\leq F(\rho_{N N''}, \rho_{N N''}'). \label{eq:fidelity_le_N}
\end{aligned}
\end{equation}

Because both $\rho_{N N''}$ and $\rho_{N N''}'$ are extreme points, they factorize over $N$ and $N''$.
\begin{equation}
    F(\rho_{NN''}, \rho_{NN''}') = F(\rho_{N}, \rho_{N}') F(\rho_{N''}, \rho_{N''}').
\end{equation}
In particular, by the isomorphism theorem, we get
\begin{equation}
    F(\rho_{N}, \rho_{N}')\leq F(\rho_{N}, \rho_{N}')^2.
\end{equation}
Therefore $F(\rho_{N}, \rho_{N}')$ is either $0$ or $1$. 

Therefore we can characterize $\Sigma(N)$ as 
\begin{equation}
    \Sigma(N) = \left\{ \bigoplus_n p_n \rho^n_N : \sum_n p_n=1, p_n\geq 0  \right\},
\end{equation}
where different extreme points $\{ \rho^n_N \}$ are supported on orthogonal subspaces. The same argument applies to $\Sigma(U)$:
\begin{equation}
    \Sigma(N) = \left\{ \bigoplus_u p_u \rho^u_U : \sum_u p_u=1, p_u\geq 0  \right\}.
\end{equation}

We shall formally denote these superselection sectors as 
\begin{equation}
\begin{aligned}
    \calC_N&=\{ 1, n,\cdots \},\\
    \calC_U&=\{ 1, u,\cdots \}.
\end{aligned}
\end{equation}

Like in the bulk, we will define the quantum dimensions of the parton sectors as
\begin{equation}
\boxed{
\begin{aligned}
    d_n &:= \exp \left( \frac{S(\rho_{N}^n) - S(\rho_{N}^1)}{2} \right), \\
    d_u &:= \exp \left( \frac{S(\rho_{U}^u) - S(\rho_{U}^1)}{2} \right),
\end{aligned}
}
\end{equation}
where $N$ and $U$ are $N$- and $U$-shaped regions respectively, and the ``1'' in the superscript means that the density matrix is obtained by tracing out all but the region in the subscript over the global reference state $\sigma$.

% Folding:
Readers may wonder whether the parton sectors we identify can be understood using the ``folding technique"~\cite{Beigi2010,KitaevKong2012,Kong2014}. The folding technique turns a system $P$ and $Q$ separated by a domain wall into a system $P\otimes \bar{Q}$ (upper half-plane) and the vacuum (lower half-plane) separated by a gapped boundary. (Here $\bar{Q}$ is $Q$ reflected along the domain wall.) Historically, the folding technique was useful in understanding the superselection sectors of point excitations on the gapped domain wall (i.e., the $O$-type sectors we shall discuss below in Sec.~\ref{sec:O_type}) by identifying them to with the superselection sectors of point excitations on the gapped boundary of the folded non-chiral system. However, we are not aware of any way to understand the parton sectors by directly looking at the folded system. The intuitive reason behind this is that the $N$-shaped subsystem cannot be obtained by unfolding any subsystem.

Let us comment on the physical interpretation of the parton sector. We emphasize that a parton sector generally does not specify a localized excitation. Specifically, if the reference state is a ground state of some local Hamiltonian, its low-energy excitation is not always uniquely determined by a single parton sector. Often extra information is required to specify such an excitation, as we explain in Section~\ref{sec:composite_sectors}.

Instead, it is better to view them as ``quantum numbers'' that partially specify excitations. Because the extreme points of $\Sigma(N)$ (as well as $\Sigma(U)$) are orthogonal to each other,  there is a set of projectors that project out a unique sector. In principle, one should be able to measure these projectors, thereby obtaining these ``quantum numbers.''

\section{Composite sectors}
\label{sec:composite_sectors}

In this section, we will study the \emph{composite superselection sectors}. These are superselection sectors localized on the domain wall that can carry multiple parton labels:
\begin{equation}
\calC_{\textrm{composite}}  = \bigcup_{\substack{n_1, n_2,\cdots\in \mathcal{C}_N \\
		u_1,u_2,\cdots \in \mathcal{C}_U}}   \calC_{\textrm{composite}}^{[n_1,u_1, n_2,u_2,\cdots]}.
\end{equation}
As before, a superselection sector is associated with some region. This region may contain $N$- and $U$-shaped regions as its subsystems, which contain partial information about the composite sectors. Specifically, recall that there are projectors localized on $N$- and $U$-shaped subsystems that can measure the $N$- and $U$-type superselection sectors. One can measure those projectors to determine the parton labels.

There can be multiple composite sectors that carry the same parton labels. In other words, the collection of parton sectors do not uniquely specify a composite sector. This is actually not a strange phenomenon. If the domain wall is trivial, the parton sectors are also trivial. Because the underlying subsystem is topologically a disk, its information convex set has a unique element~\cite{SKK2019}. However, we can instead consider an annulus, which clearly has \emph{N}- and \emph{U}-shaped regions as its subsystems. The information convex sets of these subsystems are trivial, but the information convex set of an annulus is not; see Section~\ref{sec:intro_supsec_fusion} for the discussion. Therefore, even after specifying the parton sectors, there is a leftover degree of freedom that remains unspecified.  

This ``composition rule'' of superselection sectors is somewhat mundane in the bulk. However, in the presence of a gapped domain wall, we can have a much richer structure. In Section~\ref{sec:O_type}, we shall study a composite sector that can be identified with the point-like excitations studied in Refs.~\cite{KitaevKong2012,Kong2014}. However, we shall see in Section~\ref{sec:snake} and \ref{sec:bbN_bbU} that there are other types of composite sectors as well. They are new to the best of our knowledge. While we do not believe that we have an exhaustive list of composite sectors, we expect to be able to characterize any reasonable composite superselection sectors by using the general observations summarized in Section~\ref{sec:sectors_general}.

Before we delve into these details, let us make a remark on our convention. We will frequently use the following short-hand notation for the merged state:
\begin{equation}
    \tau^{a \protect\merge b} := \rho^a \protect\merge \lambda^b, \label{eq:merge_convention}
\end{equation}
where $\rho^a$ and $\lambda^b$ are associated with superselection sectors $a$ and $b$. Both $\rho^a$ and $\lambda^b$ are elements of some information convex sets. The choice of these sets will depend on the context. 

\subsection{$O$-type sectors}
\label{sec:O_type}
The first of these composite sectors is the $O$-type superselection sector. These sectors correspond to the extreme points of an annulus on the gapped domain wall; see Fig.~\ref{fig:o_type_sector}. These extreme points are orthogonal to each other because the exposition in Section~\ref{sec:intro_supsec_fusion} applies here as well. Physically, these sectors are the superselection sectors of the point-like excitations on the gapped domain wall, studied in Refs.~\cite{KitaevKong2012,Kong2014}. 
\begin{figure}[h]
  \centering
          \includegraphics[scale=1]{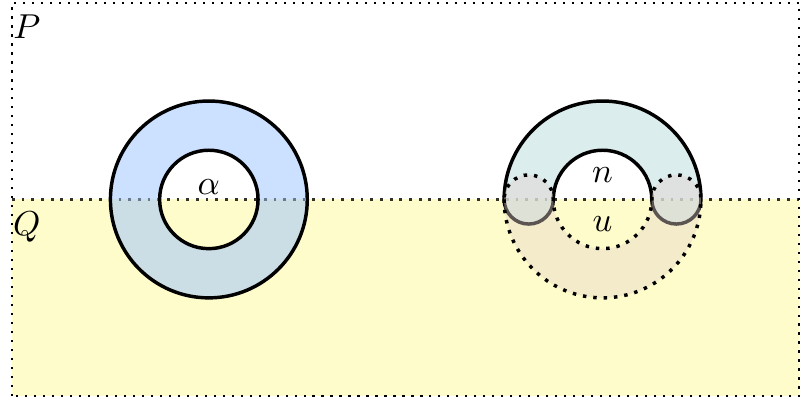}
   \caption{(Left) A subsystem associated with the $O$-type superselection sector. (Right) Upon tracing out the annulus, one can obtain a \emph{N}-shaped and \emph{U}-shaped subsystems, which must lie on the information convex set of $\Sigma(N)$ and $\Sigma(U)$, respectively.}
    \label{fig:o_type_sector}
\end{figure}
\noindent
We shall denote the set of $O$-type superselection sectors as 
\begin{equation}
   \mathcal{C}_O=\{ 1,\alpha,\beta,\cdots\},
\end{equation}
where we use Greek letters starting from $\alpha$ to denote these sectors. 

Let us explain in what sense the $O$-type sectors are composite. Consider an extreme point on the annulus that represents the sector $\alpha \in \calC_O$. Upon tracing out a disk-like region on $Q$, we get a density matrix over an $N$-shaped subsystem. Similarly, by tracing out a disk-like region on $P$, we obtain a density matrix over a \emph{U}-shaped subsystem. Moreover, these density matrices are elements of $\Sigma(N)$ and $\Sigma(U)$, respectively. 

The elements we obtain this way are not just any element; they are extreme points. To see why, consider an extreme point on the annulus that represents a $O$-type sector. As we discussed in Section~\ref{sec:intro_supsec_fusion}, if we extend an annulus to a thicker annulus and trace out the middle of the thicker annulus to obtain two annuli, these two annuli are decoupled. Importantly, subsystems of the two annuli must be also decoupled.
\begin{figure}[h]
    \centering
          \includegraphics[scale=1]{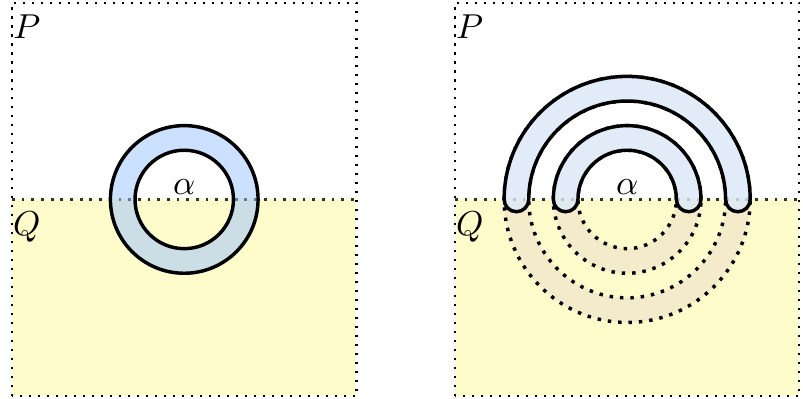}
    \caption{(Left) We begin with an extreme point corresponding to $\alpha \in \mathcal{C}_O$. (Right) By using the isomorphism theorem, the density matrix is extended to a larger annulus. Upon tracing out the middle of this larger annulus, we obtain two annuli. By the factorization of the extreme points, the state over these two annuli is a product state. Consequently, the \emph{N}-shaped subsystems in these two annuli must also be in a product state. }
    \label{fig:decoupling_n_shaped}
\end{figure}

In particular, the state over the two $N$-shaped regions in these annuli is factorized; see Fig.~\ref{fig:decoupling_n_shaped}. This is the key reason why the state is an extreme point. Let the density matrix in one of these two $N$-shaped subsystems (say $N$) to be 
\begin{equation}
    \rho_N = \bigoplus_{n\in\mathcal{C}_N} p_n \rho^n_{N}.
\end{equation}
By the isomorphism theorem, the density matrix over $NN'$ is
\begin{equation}
    \rho_{NN'} = \bigoplus_{n\in\mathcal{C}_N} p_n \rho^n_{N} \otimes \rho^n_{N'},
\end{equation}
where $N'$ is the other $N$-shaped subsystem in Fig.~\ref{fig:decoupling_n_shaped} separated from $N$. Therefore, the mutual information between the two regions is
\begin{equation}
    I(N:N') = H(\{ p_n\}).
\end{equation}
This has to be zero because the underlying state is a product state. The only possibility is that $p_n$ must be equal to $1$ for some $n$ and $0$ for other elements in $\mathcal{C}_N$. Therefore, the reduced density matrix over $N$ is an extreme point of $\Sigma(N)$. Similarly, the reduced density matrix over $U$ is an extreme point of $\Sigma(U)$.

Therefore, $\mathcal{C}_O$ must be a disjoint union of the following form:
\begin{equation}
    \mathcal{C}_O = \bigcup_{\substack{n\in \mathcal{C}_N \\
    u\in \mathcal{C}_U}} \mathcal{C}_O^{[n,u]}, \label{eq:E^nu}
\end{equation}
where $\mathcal{C}_O^{[n,u]} \subset \mathcal{C}_O$ is a subset in which the $N$- and $U$-type superselection sectors are fixed to $n$ and $u$.

The quantum dimension of this sector, which we define as
\begin{equation}
\boxed{
    d_{\alpha} := \exp\left(\frac{S(\rho^{\alpha}_{O}) - S(\rho^{1}_{O})}{2}\right),
}
\end{equation}
has a nontrivial relation with the quantum dimension of the parton sectors. Specifically, 
\begin{equation}
\boxed{
    d_n^2 d_u^2 = \frac{\sum_{\alpha \in \mathcal{C}_O^{[n,u]}} d_{\alpha}^2}{\sum_{\alpha \in \mathcal{C}_O^{[1,1]}} d_{\alpha}^2}.
}\label{eq:qd_parton_composite}
\end{equation}

\begin{figure}[h]
  \centering
            \includegraphics[scale=1]{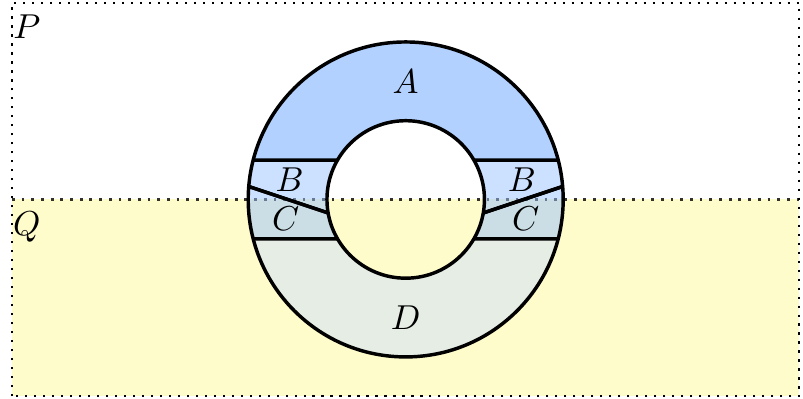}
	\caption{Merging a pair of parton sectors to obtain an $O$-type composite sector. The subsystem $N=ABC$ carries $n\in\calC_N$ and the subsystem $U=BCD$ carries $u\in\calC_U$. In the main text, $O$ is defined to be $ABCD$ in this figure.}
	\label{fig:merging_into_composite_point}
\end{figure}

To derive this relation, we use the merging technique used in Section~\ref{sec:intro_qd_fm}. Specifically, we merge extreme points of $\Sigma(N)$ and $\Sigma(U)$ to obtain an element in $\Sigma(O)$, where $O$ is an annulus on the domain wall; see Fig.~\ref{fig:merging_into_composite_point}. Without loss of generality, let us refer to these extreme points as $\rho^n_N$ and $\rho^u_U$. For the merged state $\tau^{n\merge u}_{O} := \rho^n_N \merge \rho_U^u$, its entropy is equal to 
\begin{equation}
    S(\tau^{n\merge u}_O) = \ln (d^2_n d^2_u) + (S_{N} + S_{U} - S_{N\cap U})_{\sigma}.
\end{equation}
On the other hand, we can directly obtain the maximum entropy consistent with the given extreme points in $\Sigma(N)$ and $\Sigma(U)$:
\begin{equation}
\begin{aligned}
    S(\tau^{n\merge u}_O) &= \max_{\{p_{\alpha}\} } \left(H(\{p_{\alpha}\}) + \sum_{\alpha \in \mathcal{C}_O^{[n,u]}} p_{\alpha} \ln d^2_{\alpha} \right) + S_{O}(\sigma) \\
    &= \ln \left(\sum_{\alpha \in \mathcal{C}_O^{[n,u]}} d_{\alpha}^2 \right) + S(\sigma_O).
\end{aligned}
\end{equation}
Let $\tau^{1\merge 1}_O$ be the state merged from extreme points $1\in\mathcal{C}_N$ and $1\in \mathcal{C}_U$. We get
\begin{equation}
\begin{aligned}
    S(\tau^{n\merge u}_O) - S(\tau^{1\merge 1}_O) &= \ln \left(d_n^2d_u^2 \right) \\
    &= \ln \left(\frac{\sum_{\alpha\in \mathcal{C}_O^{[n,u]}} d_{\alpha}^2}{\sum_{\alpha\in \mathcal{C}_O^{[1,1]}} d_{\alpha}^2} \right),
\end{aligned}
\end{equation}
which leads to Eq.~\eqref{eq:qd_parton_composite}.

\subsection{Snake sectors}
\label{sec:snake}
While the $O$-type sector has appeared in the literature already, there are other composite sectors that are new to the best of our knowledge. One such example is the \emph{snake sector}, or alternatively, a \emph{S}-type sector. This is a superselection sector associated with the ``snake''-shaped regions, e.g., $S$ and $S'$ in Fig.~\ref{fig:snake}. The information convex sets of these subsystems are isomorphic to a simplex with a finite number of orthogonal extreme points. These snake sectors are again composite sectors of $N$-type and $U$-type sectors, and therefore many of the discussions about $\calC_O$ in Section~\ref{sec:O_type} apply here as well.

\begin{figure}[h]
  \centering
              \includegraphics[scale=1]{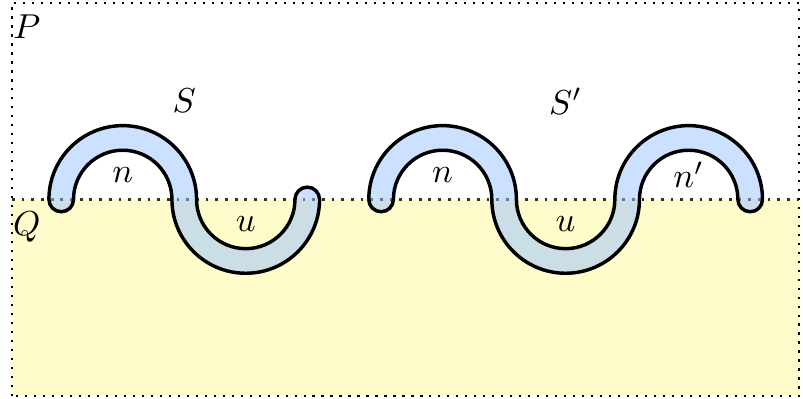}
	\caption{Snake-shaped subsystems.}
	\label{fig:snake}
\end{figure}

Let $S$ be the simplest snake-shaped region in Fig.~\ref{fig:snake}. The set of snake sectors is a disjoint union of the following form.
 \begin{equation}
 \calC_S = \bigcup_{\substack{n\in \mathcal{C}_N \\ u\in \mathcal{C}_U}} \calC_S^{[n,u]},
 \end{equation}
where the extreme points associated with $\calC_S^{[n,u]}$ carry $n\in \mathcal{C}_N$ and $u\in \mathcal{C}_U$.

Also, we can define the quantum dimensions as follows
\begin{equation}
\boxed{
d_{s} :=   \exp\left(\frac{S(\rho_{S}^{s}) - S(\rho^1_{S})}{2}\right), 
}
\end{equation}
where $\rho_S^s$ is an extreme point of $\Sigma(S)$.

There is a nontrivial identity between $\{d_s\}$ and the quantum dimension of the parton sectors:
\begin{equation}
\boxed{
 d_n^2 d_u^2 = \sum_{s\in \calC_S^{[n,u]} } d_s^2.
}\label{eq:d_s_consistency}
\end{equation}
The proof is essentially the same as the proof of Eq.~\eqref{eq:qd_parton_composite}, the only difference being that $\sum_{s \in \calC_S^{[1,1]}} d_s^2=1$. 

This last fact follows from the fact that  $|\calC_S^{[n,1]}| = |\calC_S^{[1,u]}|=1$.
That $\calC_S^{[n,1]}$ has a unique element follows from the observation that an element of $\Sigma(S)$ that carries parton sector $u=1$ is the reduced density matrix of a certain element in $\Sigma(SD)$, where $SD$ is an $N$-shaped subsystem; see Fig.~\ref{fig:fill_the_slot}. Here, $D$ is a disk on the domain wall, which fills the ``slot'' in Fig.~\ref{fig:fill_the_slot} and turns the  $U$-shaped arc into a disk on the domain wall. (In more detail, we need to divide $S$ into $ABC$ in an obvious way, in which $BCD$ is another disk on the domain wall. Then, we use the merging theorem. Note that the merging is possible because $u=1$.) The proof of $|\calC_S^{[1,u]}|=1$ is analogous.

\begin{figure}[h]
	\centering
              \includegraphics[scale=1]{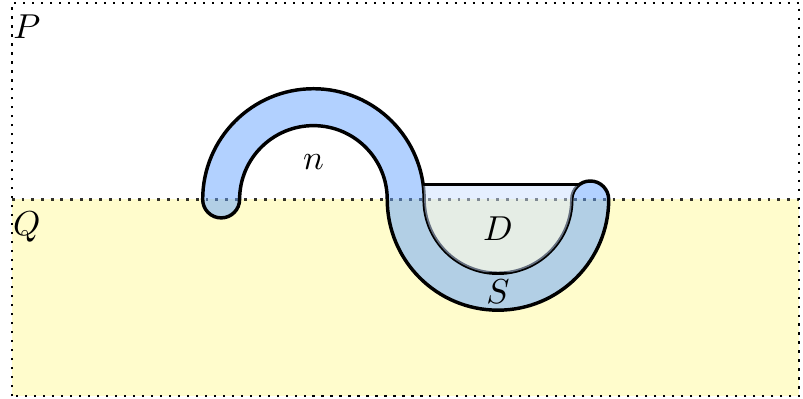}
	\caption{For the proof of $|\calC_S^{[n,1]}|=1 $. A snake-shaped subsystem $S$ and an $N$-shaped subsystem $SD$. Here, $D$ is a disk on the domain wall and it fills a slot.}
	\label{fig:fill_the_slot}
\end{figure}

\subsection{$\mathbb{N}$- and $\mathbb{U}$-type sectors}\label{sec:bbN_bbU}
There are composite sectors that play a crucial role in studying the fusion space of the aforementioned superselection sectors. These are the $\mathbb{N}$- and $\mathbb{U}$-type sectors; see Fig.~\ref{fig:NNUU}. The underlying subsystems are annuli on the domain wall which are not path-connected to any $O$-shaped subsystem. It should be obvious -- from the discussion about the bulk superselection sectors and the parton sectors -- that the information convex set associated with this subsystem is also isomorphic to a simplex formed by a finite number of mutually orthogonal extreme points. Moreover, these are composite sectors in a sense that, upon tracing out the appropriate subsystems, one can obtain two $N$- and $U$-shaped subsystems. Moreover, the argument that leads to Eq.~\eqref{eq:E^nu} also applies here, which implies that $\mathcal{C}_{\mathbb{N}}$ is a disjoint union of sets labeled by $n, n',u,$ and $u'$, where $n,n' \in \mathcal{C}_{N}$ and $u,u' \in \mathcal{C}_U$. However, we will not use this fact in this paper. 
\begin{figure}[h]
	\centering
              \includegraphics[scale=1]{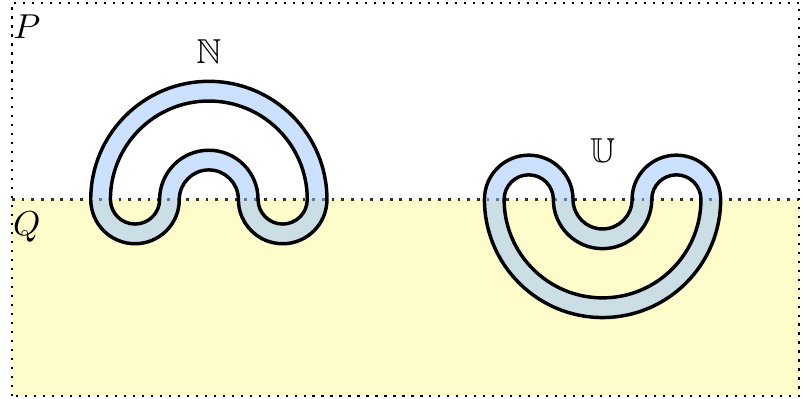}
	\caption{Subsystem choices for $\calC_{\mathbb{N}}$ and $\calC_{\mathbb{U}}$.}
	\label{fig:NNUU}
\end{figure}

We define the quantum dimensions of these sectors as follows:
\begin{equation}
\boxed{
\begin{aligned}
d_{\mathcal{N}} &=   \exp\left(\frac{S(\rho_{\mathbb{N}}^{\mathcal{N}}) - S(\rho^1_{\mathbb{N} })}{2}\right), \\	
d_{\mathcal{U} }   &=    \exp\left(\frac{S(\rho^{\mathcal{U}}_{\mathbb{U}}) - S(\rho^1_{\mathbb{U}})}{2}\right),
\end{aligned}  
}
\end{equation}
where $\mathcal{N}\in \mathcal{C}_{\mathbb{N}}$ and $\mathcal{U} \in \mathcal{C}_{\mathbb{U}}$ are the $\mathbb{N}$- and $\mathbb{U}$-type superselection sectors. We again use the superscript ``1'' to denote the extreme point obtained from the reference state $\sigma$.

There is a natural notion of embedding:
\begin{equation}
\begin{aligned}
    \eta_N&: \mathcal{C}_N \hookrightarrow \mathcal{C}_{\mathbb{N}} \\
    \eta_U&: \mathcal{C}_U \hookrightarrow \mathcal{C}_{\mathbb{U}},
\end{aligned}
\label{eq:embedding_N2NN}
\end{equation}
which is defined by tracing out the interior of the \emph{N}-shaped (or \emph{U}-shaped) subsystem; see Fig.~\ref{fig:eta}.
\begin{figure}[h]
    \centering
    \includegraphics[scale=1]{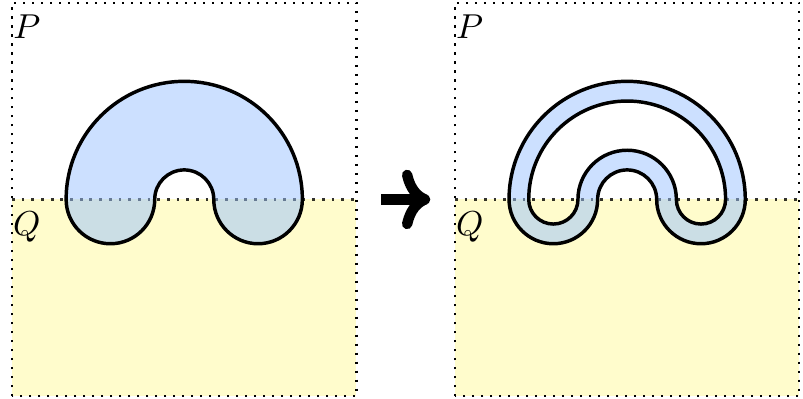}
    \caption{Given an extreme point of $\Sigma(N)$, we can obtain an extreme point of $\Sigma(\mathbb{N})$ by tracing out the middle part.}
    \label{fig:eta}
\end{figure}

In Eq.~\eqref{eq:embedding_N2NN}, we are implicitly asserting that an extreme point of $\Sigma(N)$, upon traced out the middle part, becomes an extreme point of an information convex set of a $\mathbb{N}$-shaped region. Below, we briefly sketch the underlying reason. 

Consider an $N$-shaped subsystem $N$, which is partitioned into $N'$ and $\mathbb{N} = \mathbb{N}_{\text{in}} \mathbb{N}_{\text{middle}} \mathbb{N}_{\text{out}}$, where $\mathbb{N}_{\text{in}}, \mathbb{N}_{\text{middle}},$ and $\mathbb{N}_{\text{out}}$ are non-overlapping $\mathbb{N}$-shaped regions. Specifically, we have the following sequence of $N$-shaped regions:
\begin{equation}
\begin{aligned}
    N' &\subset N'\mathbb{N}_{\text{in}} \\
    &\subset N'\mathbb{N}_{\text{in}}\mathbb{N}_{\text{middle}} \\
    &\subset N'\mathbb{N}_{\text{in}}\mathbb{N}_{\text{middle}}\mathbb{N}_{\text{out}},
\end{aligned}
\end{equation}
and the following sequence of $\mathbb{N}$-shaped regions:
\begin{equation}
    \begin{aligned}
        \mathbb{N}_{\text{in}}  &\subset \mathbb{N}_{\text{in}} \mathbb{N}_{\text{middle}} \\
        &\subset \mathbb{N}_{\text{in}} \mathbb{N}_{\text{middle}} \mathbb{N}_{\text{out}} \\
        &= \mathbb{N}.
    \end{aligned}
\end{equation}

Let $\rho^{\langle e\rangle}_N$ be an extreme point of $\Sigma(N)$. By the factorization property of the extreme point, we have
\begin{equation}
    I(N'\mathbb{N}_{\text{in}}: \mathbb{N}_{\text{out}})_{\rho^{\langle e\rangle}} =0.
\end{equation}
By the monotonicity of the mutual information, we get
\begin{equation}
I(\mathbb{N}_{\text{in}}: \mathbb{N}_{\text{out}})_{\rho^{\langle e\rangle}} =0.
\end{equation}
This is possible only if the reduced density matrix of $\rho^{\langle e\rangle}_N$ over $\mathbb{N}$ is an extreme point. 

Moreover, using the factorization property, we can derive:
\begin{equation}
\boxed{
\begin{aligned}
    d_{\eta_N(n)} &= d_n^2,\\
    d_{\eta_U(u)} &= d_u^2.
\end{aligned}
}\label{eq:eta_relation}
\end{equation}
To see why, without loss of generality, consider the subsystems described in Fig.~\ref{fig:eta_relation}. Here, both $N$ and $N'$ are $N$-shaped subsystems. Importantly, $N'\setminus N$ is a $\mathbb{N}$-shaped subsystem. Using the factorization property of the extreme points, we get:
\begin{equation}
\begin{aligned}
    (S_N + S_{N'} - S_{N'\setminus N})_{\rho_{N'}^n}&=0,\\
    (S_N + S_{N'} - S_{N'\setminus N})_{\rho_{N'}^1}&=0,
\end{aligned}
\end{equation}
where $\rho_{N'}^n \in \Sigma(N')$ is an extreme point associated with the sector $n\in \mathcal{C}_{N}$. From these equations and the definition of the quantum dimension, Eq.~\eqref{eq:eta_relation} follows. 
\begin{figure}[h]
  \centering
  \includegraphics[scale=1]{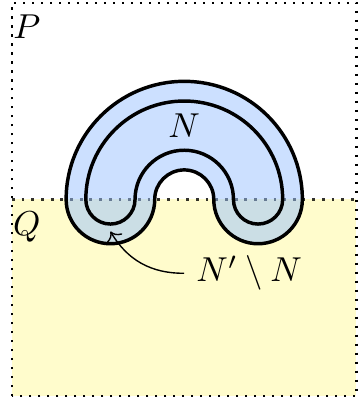}
    \caption{Subsystems involved in the proof of Eq.~\eqref{eq:eta_relation}}
    \label{fig:eta_relation}
\end{figure}

Later in Section~\ref{Sec:consistency_N_U}, we shall see that there is a one-to-one map between the set of $\mathbb{N}$-type sectors and the set of $\mathbb{U}$-type sectors. We denote this fact as follows:
\begin{equation}
\varphi:\calC_{\mathbb{N}}\to \calC_{\mathbb{U}},
\end{equation}
where $\varphi$ is a bijection. Later, we will show that this map preserves the quantum dimensions, namely
\begin{equation}
d_{\mathcal{N}} = d_{\varphi({\mathcal{N}}) }. \label{eq:1046_1}
\end{equation}
This would be certainly true if the domain wall is trivial since both subsystems can be smoothly deformed to an annulus. However, because $\mathbb{N}$ and $\mathbb{U}$ cannot be smoothly deformed into each other, Eq.~\eqref{eq:1046_1} is a nontrivial fact in general.

To summarize, different sets of superselection sectors are related to each other in the following way:
\begin{equation}
\begin{tikzcd}
\mathcal{C}_N \arrow[d, hook, "\eta_N"]  & \mathcal{C}_U \arrow[d, hook, "\eta_U"]
\\
\mathcal{C}_{\mathbb{N}} \arrow[leftrightarrow]{r} & \mathcal{C}_{\mathbb{U}}
\end{tikzcd}.
\end{equation}
\noindent
While the cardinality of $\mathcal{C}_N$ is generally different from that of $\mathcal{C}_U$, those two sets may be indirectly related to each other via $\mathcal{C}_{\mathbb{N}}$ and $\mathcal{C}_{\mathbb{U}}$.

\subsection{Generalities}
\label{sec:sectors_general}
In this section, we introduce general facts about superselection sectors. First, we explain an all-encompassing recipe to show that an information convex set of a subsystem is isomorphic to a simplex with orthogonal extreme points. The following discussion will assume the continuum limit, in which the familiar notion of topology is well-defined.

The following definition will be important.
\begin{definition}[Sectorizable Region]
A subsystem $R$ is \emph{sectorizable} if there is a region $\widehat{R}$ such that:
\begin{enumerate}
	\item $\widehat{R}$ contains disjoint regions $R$ and $R''$ and
	\item both $R$ and $R''$ can be connected to $\widehat{R}$ by a path, where the path is a sequence of extensions.
\end{enumerate}
\label{def:sectorizable}
\end{definition}

This definition is important because the information convex set of any sectorizable region is a simplex with orthogonal extreme points.
\begin{lemma} \label{lemma:sectorizable}
Let $R$ be a sectorizable region. Then
\begin{equation}
    \Sigma(R) = \left\{\bigoplus_I p_I \rho^I_R : \sum_I p_I=1, p_I\geq 0 \right\},
\end{equation}
where $\{ \rho^I_R \}$ is a set of density matrices that are mutually orthogonal to each other.
\end{lemma}

The proof of this lemma is straightforward, because it is a simple generalization of what we have been discussing so far. For completeness, we sketch the proof below. First, extend $R$ to $\widehat{R}$ using the isomorphism theorem, and then trace out $\widehat{R} \setminus (R \cup R'')$. We can obtain the following inequality:
\begin{equation}
    F(\rho_R, \rho'_{R}) \leq F(\rho_{RR''}, \rho'_{RR''}), \label{eq:fedility_le_general}
\end{equation}
where $\rho_R$ and $\rho'_R$ are two extreme points of $\Sigma(R)$, and $\rho_{RR''}$ and $\rho'_{RR''}$ are obtained from the former density matrices by an extension to $\widehat{R}$ and a partial trace over $\widehat{R} \setminus (R \cup R'')$. Note that Eq.~(\ref{eq:fedlity_le_X}) and Eq.~(\ref{eq:fidelity_le_N}) are special cases of Eq.~(\ref{eq:fedility_le_general}). 
By the factorization property,  we get $F(\rho_R, \rho'_R) \leq F(\rho_R, \rho'_R)^2$. Therefore, $F(\rho_R, \rho'_R)$ must be either $0$ or $1$. This proves Lemma~\ref{lemma:sectorizable}.

Second, for two sectorizable subsystems, the set of extreme points obeys the ``product rule.'' 

\begin{lemma}[Product rule]\label{lemma:product_rule}
Let $\Omega_1$ and $\Omega_2$ be sectorizable subsystems which are disjoint from each other. Then the subsystem $\Omega_1\Omega_2$ is another sectorizable subsystem and
\begin{equation}
\boxed{
\mathcal{C}_{\Omega_1  \Omega_2} \cong \mathcal{C}_{\Omega_1}  \times \mathcal{C}_{\Omega_2},
}
\label{eq:product_rule}
\end{equation}
where $\mathcal{C}_{\Omega_1}$, $\mathcal{C}_{\Omega_2}$ and $\mathcal{C}_{\Omega_1  \Omega_2}$ are the set of superselection sectors associated with sectorizable subsystems $\Omega_1$, $\Omega_1$ and $\Omega_1\Omega_2$ respectively.
Moreover, every extreme point of $\Sigma(\Omega_1  \Omega_2)$ is a tensor product of extreme points of $\Sigma(\Omega_1)$ and $\Sigma(\Omega_2)$.
\end{lemma}
\begin{proof}
First, the fact that $\Omega_1\Omega_2$ is again a sectorizable subsystem is easy to verify. The two conditions in	Definition~\ref{def:sectorizable} are verified by letting $\widehat{\Omega_1\Omega_2}= \widehat{\Omega}_1 \widehat{\Omega}_2$ and $(\Omega_1\Omega_2)''=\Omega''_1\Omega''_2$. 

Second, note that any extreme point of $\Sigma(\Omega_1 \Omega_2)$, once restricted to either $\Omega_1$ or $\Omega_2$, becomes an extreme point of $\Sigma(\Omega_1)$ and $\Sigma(\Omega_2)$ respectively. 
This is because, once we extend an extreme point in $\Sigma(\Omega_1 \Omega_2)$ to an element of $\Sigma(\widehat{\Omega_1\Omega_2})$ by using the isomorphism theorem and tracing out the appropriate subsystems, the mutual information of this state between $\Omega_1\Omega_2$ and $\Omega''_1\Omega''_2$ is zero. This fact follows from the factorization property of the extreme points. Therefore, the state must be factorized over $\Omega_1$ and $\Omega''_1$. The same factorization holds over $\Omega_2$ and $\Omega''_2$. Such factorization is possible only if the reduced density matrix of any extreme point of $\Sigma(\Omega_1 \Omega_2)$ over $\Omega_1$ and $\Omega_2$ are extreme points.

Now, we can use the factorization property of the extreme point of $\Sigma(\Omega_1)$ as follows. Note that the extreme points of $\Sigma(\Omega_1\Omega_2)$, restricted to $\Omega_1 \setminus \partial \Omega_1$, where $\partial \Omega_1$ is the thickened boundary of $\Omega_1$, must be factorized with anything that is outside of $\Omega_1$. Therefore, these extreme points must be factorized between $\Omega_1 \setminus \partial \Omega_1$ and $\Omega_2$. Using the isomorphism theorem, we conclude that the extreme points over $\Sigma(\Omega_1\Omega_2)$ must be factorized over $\Omega_1$ and $\Omega_2$.  
\end{proof}

\subsection{Summary}
We have so far studied the parton sectors and its (derivative) composite sectors. Below, we summarize our key results for the readers' convenience. First, we have summarized these superselection sectors in Fig.~\ref{fig:sectors}. Note that the set of composite sectors can be decomposed further into a disjoint union of sets, each of which is labeled by the parton sectors. For instance, $O$- and $S$-type sectors are labeled by an $N$- and a $U$-type sector. On the other hand, $\mathbb{N}$- and $\mathbb{U}$-type sectors are labeled by two $N$- and two $U$-type sectors.

\begin{figure}[h]
	\centering
\includegraphics[scale=1]{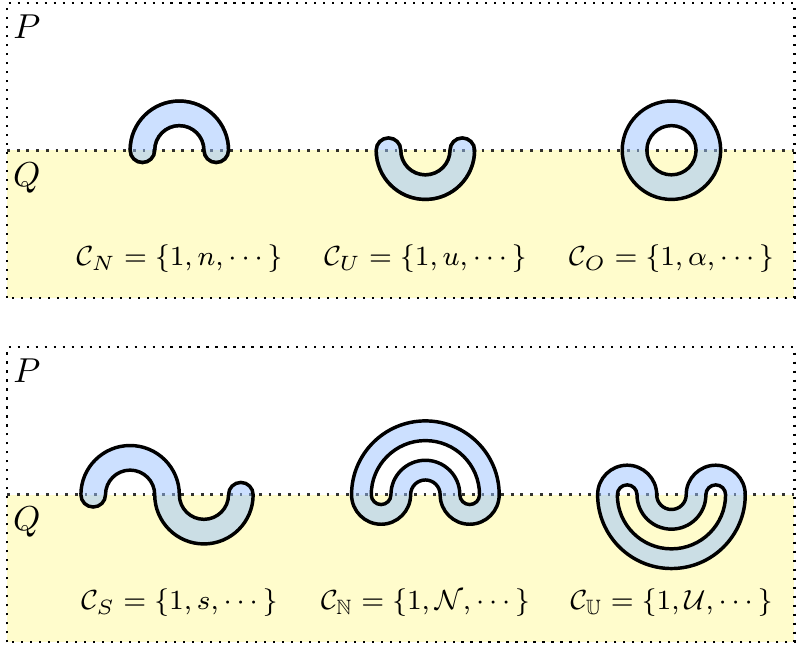}
	\caption{A list of subsystem topologies and the corresponding superselection sector labels.}\label{fig:sectors}
\end{figure}

The quantum dimensions of these sectors are all defined in the same way, in terms of the entanglement entropy of the extreme point associated with the superselection sector. 

We derived the following identities:
\begin{equation}
\begin{aligned}
    d_n^2d_u^2 &= \frac{\sum_{\alpha \in \mathcal{C}_O^{[n,u]}} d_{\alpha}^2}{\sum_{\alpha \in \mathcal{C}_O^{[1,1]}} d_{\alpha}^2} \\
    &= \sum_{s\in \mathcal{C}_S^{[n,u]}} d_s^2.
\end{aligned}
\end{equation}
Moreover, we studied the maps $\eta_N$ (as well as $\eta_U$) and $\varphi$ which has the following properties. The maps $\eta_N$ and $\eta_U$ are embeddings from $\mathcal{C}_N$ to $\mathcal{C}_{\mathbb{N}}$ and from $\mathcal{C}_U$ to $\mathcal{C}_{\mathbb{U}}$ respectively, such that
\begin{equation}
\begin{aligned}
    d_{\eta_N(n)} &= d_n^2\\
    d_{\eta_U(u)} &= d_u^2.
\end{aligned}
\end{equation}
$\varphi$ is a bijection between $\mathcal{C}_{\mathbb{N}}$ and $\mathcal{C}_{\mathbb{U}}$ such that
\begin{equation}
\begin{aligned}
    d_{\varphi(\mathcal{N})} &= d_{\mathcal{N}} \\
    d_{\varphi^{-1}(\mathcal{U})} &= d_{\mathcal{U}}.
\end{aligned}
\end{equation}

Finally, we mention that  anti-sectors are well defined for $\calC_N$, $\calC_U$ and $\calC_O$. The quantum dimension of every sector is equal to that of its anti-sector.
\begin{equation}
   \begin{aligned}
       d_{\bar{n}}&= d_n, \quad \bar{\bar{n}}=n, \quad \forall n\in \calC_N,\\
        d_{\bar{u}}&= d_u,\quad \bar{\bar{u}}=u, \quad \forall u\in \calC_U,\\
         d_{\bar{\alpha}}&= d_\alpha, \quad \bar{\bar{\alpha}}=\alpha, \quad \forall \alpha\in \calC_O,
   \end{aligned} 
\end{equation}
where we have used a ``bar'' over a sector label to denote its anti-sector.
We will prove these identities later.

\section{Fusion spaces}\label{sec:fusion}
In this section, we define and study the fusion spaces of the superselection sectors introduced in Section~\ref{sec:parton} and \ref{sec:composite_sectors}. To understand our definition of fusion space, it will be instructive to recall the definition of fusion space in the theory of anyon. In the anyon theory, a fusion space is a \emph{Hilbert space}. Specifically, when a sector $a$ and $b$ fuse into another sector $c$, there is a leftover degree of freedom, described by a state space of some Hilbert space. This underlying Hilbert space \emph{is} the fusion space. In this paper, we will adhere strictly to this rule and ascribe a fusion space to any space isomorphic to a state space of some Hilbert space. 

Without loss of generality, consider an information convex set $\Sigma(\Omega)$ associated with a subsystem $\Omega$. To characterize $\Sigma(\Omega)$, it will be helpful to study the information convex set of its thickened boundary $\partial \Omega$; see Fig.~\ref{fig:fusion_characterization}. Because $\partial \Omega$ is a sectorizable region, $\Sigma(\partial \Omega)$ is a simplex with orthogonal extreme points; see Definition~\ref{def:sectorizable} and Lemma~\ref{lemma:sectorizable}. Moreover, the factorization property of the extreme points implies that every extreme point of $\Sigma(\Omega)$ reduces to an extreme point of $\Sigma(\partial \Omega)$. Therefore, elements of $\Sigma(\Omega)$ can be divided further in terms of the extreme points of $\Sigma(\partial \Omega)$:
\begin{equation}
\boxed{
    \Sigma(\Omega) = \left\{ \bigoplus_{I \in \mathcal{C}_{\partial \Omega}} p_{I} \rho_{\Omega}^{I} : \sum_{I} p_{I} = 1, p_{I} \geq 0 \right\},
}\label{eq:fusion_space1}
\end{equation}
where $\mathcal{C}_{\partial \Omega}$ is a set of superselection sectors associated with the extreme points of $\Sigma(\partial \Omega)$ and $\rho_{\Omega}^{I}$ is an element of $\Sigma(\Omega)$ that, upon restricting to $\partial \Omega$, becomes an extreme point $\rho^I_{\partial \Omega} \in \Sigma(\partial \Omega)$. The set of $\rho_{\Omega}^{I}$ with a fixed $I$ forms a convex subset of $\Sigma(\Omega)$, which we shall denote as $\Sigma_I(\Omega)$. 

\begin{figure}[h]
    \centering
\includegraphics[scale=1]{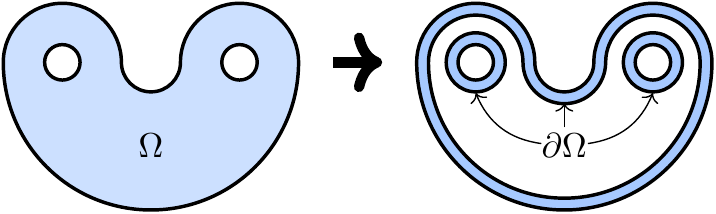}
    \caption{An example of thickened boundaries. Here, $\Omega$ is a sufficiently thick and smooth two-hole disk. The thickened boundary $\partial \Omega$ is a sectorizable subsystem. Furthermore, $\partial \Omega$ is a union of three disjoint annuli, each of which is a sectorizable subsystem.}
    \label{fig:fusion_characterization}
\end{figure}

It remains to characterize $\mathcal{C}_{\partial \Omega}$ and $\Sigma_{I}(\Omega)$. For $\mathcal{C}_{\partial \Omega}$, we can use the general strategy explained in Section~\ref{sec:composite_sectors}. For instance, if $\partial \Omega$ has multiple connected components, the set of superselection sectors $\mathcal{C}_{\partial \Omega}$ obeys the product rule (Lemma~\ref{lemma:product_rule}). For example, in Fig.~\ref{fig:fusion_characterization}, $\partial \Omega$ is the union of three disjoint annuli. In this case, $I\in \mathcal{C}_{\partial \Omega}$ is a triple of superselection sectors of the three annuli, i.e., $ \{I\} \cong \{(a,b,c)\}$ where $a, b,$ and $c$ belong to the set of superselection sectors associated with an annulus. 

For $\Sigma_I(\Omega)$, we can prove the following fact:
\begin{equation}
\boxed{
    \Sigma_{I}(\Omega) \cong \mathcal{S}(\mathbb{V}_I), \label{eq:fusion_space2}
}
\end{equation}
where $\mathcal{S}(\mathbb{V}_I)$ is the state space of a finite-dimensional Hilbert space $\mathbb{V}_I$, which generally depends on the choice of $I$. Combined with Eq.~\eqref{eq:fusion_space1}, this implies that one can assign a fusion space to any sufficiently smooth and thick subsystem. The dimension of the Hilbert space, $N_I =\dim \mathbb{V}_I$, is a non-negative integer known as the \emph{fusion multiplicity}.
 
As a sanity check, we can see that the fusion space defined in Eq.~\eqref{eq:fusion_space2} produces sensible results in known setups. In Fig.~\ref{fig:fusion_characterization}, $N_I$ is simply $N_{ab}^c$, the multiplicity for the fusion of anyons $a$ and $b$ into an anyon $c$. In that context, Eq.~(\ref{eq:fusion_space2}) was derived in Theorem~4.5 of Ref.~\cite{SKK2019}.

The proof of Eq.~\eqref{eq:fusion_space2} for general subsystems can be done similarly as the proof of Theorem~4.5 of Ref.~\cite{SKK2019}. Moreover, we provide an alternative proof which is simpler; see the \emph{Hilbert space theorem} (Theorem~\ref{thm:Hilbert}) in Appendix~\ref{appendix:fusion_space}. 

\subsection{Fusion on gapped domain walls}
\label{sec:fusion_examples}

In this section, we list a few examples of fusion spaces on gapped domain walls. A (partial) list of relevant subsystems is described in Fig.~\ref{fig:fusion_summary}. For example, we can consider a two-hole disk on the domain wall, both of the holes sitting on the domain wall; see the first figure in Fig.~\ref{fig:fusion_summary}. The thickened boundary of that region is a union of three disjoint annuli on the domain wall, with the extreme points labeled by $\mathcal{C}_O$.  Hence, the fusion space of the two-hole disks on the domain wall can be labeled by a triple $(\alpha, \beta, \gamma)$, where $\alpha, \beta, \gamma \in \mathcal{C}_O$. We may formally denote the fusion space as $\mathbb{V}_{\alpha\beta}^{\gamma}$ and the fusion multiplicity as $N_{\alpha\beta}^{\gamma}$.

The other examples listed in Fig.~\ref{fig:fusion_summary} can be understood in a similar way. While we have discussed our notation of superselection sectors in Section~\ref{sec:composite_sectors}, we restate it below for the readers' convenience:
\begin{equation}
\begin{aligned}
    \alpha,\beta,\gamma &\in \calC_O,\\
    a&\in \calC_P, \\
    x&\in \calC_Q, \\
    \mathcal{U}&\in \calC_{ \mathbb{U} }, \\
    \mathcal{N}&\in \calC_{ \mathbb{N} },
\end{aligned}
\end{equation}
where $\calC_P$ and $\calC_Q$ denote the set of anyon labels in phases $P$ and $Q$ respectively.

\begin{figure}[h]
	\centering
\includegraphics[scale=1]{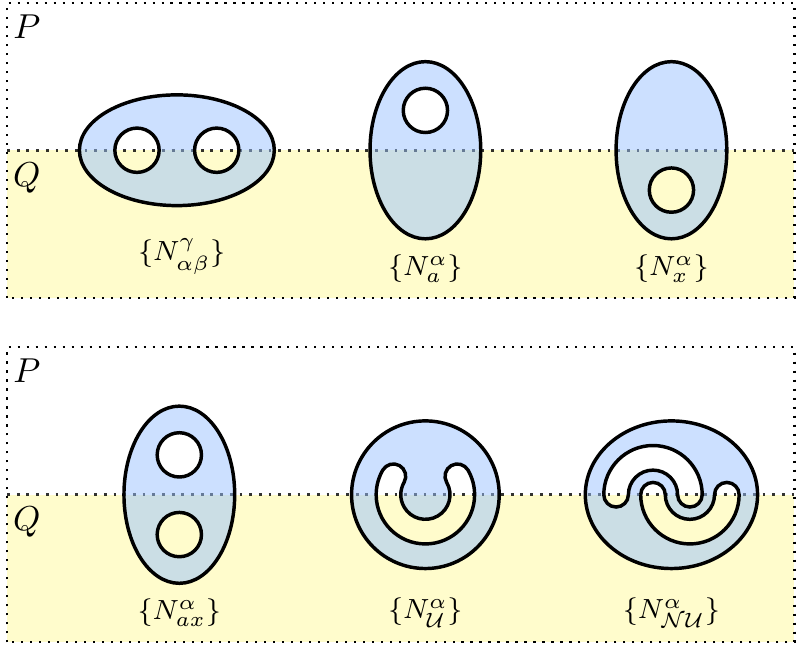}
	\caption{Subsystems that are relevant to the study of fusion spaces on gapped domain walls. Also shown are the labels for the fusion multiplicities.}\label{fig:fusion_summary}
\end{figure}

In Section~\ref{sec:fusion_rules}, we will in fact derive the fusion \emph{rules} that these fusion spaces must obey and derive intricate constraints on the fusion multiplicities. Let us briefly mention these results, deferring the details to Section~\ref{sec:fusion_rules}. We can formally express the following fusion processes:
\begin{equation}
\boxed{    \begin{aligned}
        \alpha \times \beta &=\sum_{\gamma} N_{\alpha \beta}^{\gamma} \gamma \\
        a &= \sum_{\alpha} N_a^{\alpha} \alpha \\
        x &= \sum_{\alpha} N_x^{\alpha} \alpha \\
        a \times x &= \sum_{\alpha} N_{ax}^{\alpha} \alpha .
    \end{aligned}
    } \label{eq:fusion_summary}
\end{equation}
Here, for any choice of sectors on the left-hand side there exists at least one fusion result on the right-hand side.

However, the same cannot be said about the fusion processes involving multiplicities $N^{\alpha}_{\mathcal{U}}$ and $N^{\alpha}_{\mathcal{N}\mathcal{U}}$. For example, for a particular choice of $\mathcal{N}$ and $\mathcal{U}$, there may not be an $\alpha$ such that $N_{\mathcal{N} \mathcal{U}}^{\alpha} \neq 0$.  We will revisit this issue in Section~\ref{sec:fusion_rules}.

\subsection{Quasi-fusion of parton sectors}\label{Sec:quasi_fusion_subsection}
We have seen examples of fusion spaces in Fig.~\ref{fig:fusion_summary}. They involve composite sectors on the domain wall as well as the superselection sectors of anyons in the 2D bulk $P$ and $Q$. One may wonder whether there is a similar generalization of fusion spaces to parton sectors. What happens if we ``fuse'' a pair of parton sectors (say $n,n'\in \calC_{N}$) together? Can they fuse into another parton sector $n''\in \calC_N$? Can we define a fusion space ($\mathbb{V}_{nn'}^{n''}$) associated with a triple of parton sectors?

Surprisingly, the answer to the last question is ``no.'' To see why, let us first formalize the problem. Consider a $M$-shaped subsystem ($M$) shown in Fig.~\ref{fig:M-shape_and_quasi_fusion}. There are three $N$-shaped subsystems, associated with superselection sectors $n, n',$ and $n''$ without loss of generality. The question is whether the state space with a fixed choice of $n, n',$ and $n''$ is isomorphic to a state space of some Hilbert space.

\begin{figure}[h]
  \centering
\includegraphics[scale=1]{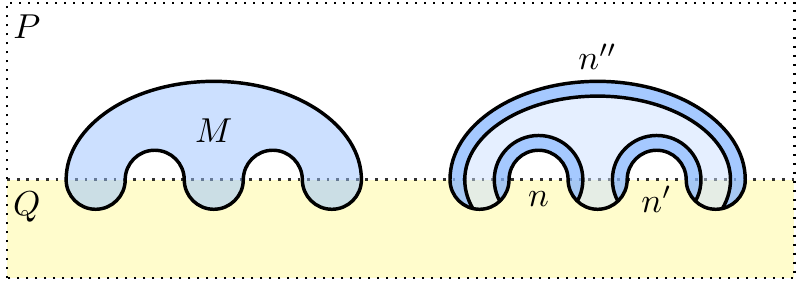}
	\caption{An $M$-shaped region. Three parton sectors $n, n', n'' \in \calC_N$ can be detected from the three $N$-shaped subsystems (darker color) of $M$.}
	\label{fig:M-shape_and_quasi_fusion}
\end{figure}

It turns out that this is not the case. What \emph{is} correct is the fact that $n, n',$ and $n''$ partially characterize the extreme points of $\Sigma(\partial M)$, where $\partial M$ is the thickened boundary of $M$. However, they do not characterize the extreme points of $\Sigma(\partial M)$ completely. This is because $\partial M$ is not a union of three $N-$shaped regions; the $N$-shaped regions associated with $n, n',$ and $n''$ are part of $\partial M$ but not all of it. Therefore, even after specifying $n, n',$ and $n'',$ one may have more than one fusion space, each labeled by an extreme point of $\Sigma(\partial M)$. We shall refer to this phenomenon as \emph{quasi-fusion} of parton sectors.

\begin{figure}[h]
	\centering
\includegraphics[scale=1]{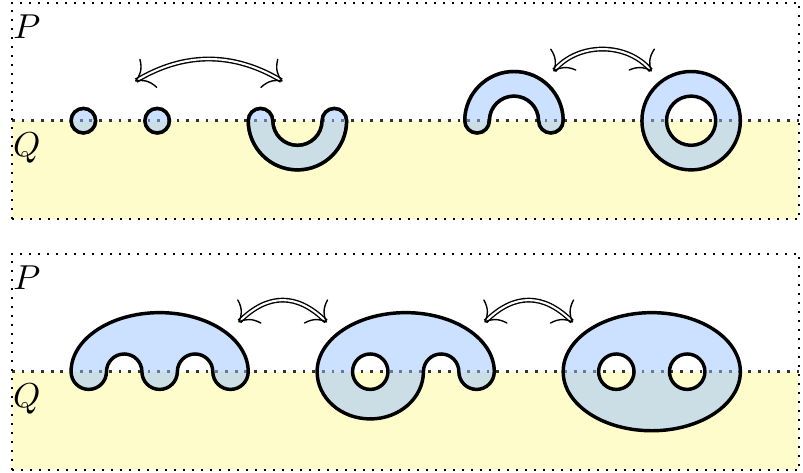}
	\caption{If $Q$ has trivial anyon content, we can apply topology-changing operation to the subsystems in $Q$ side without affecting the structure of the information convex sets. In this case, $\calC_U=\{1\}$ and $\calC_N \cong \calC_O$. Moreover, the quasi-fusion rule of $N$-type parton sectors is identical to the conventional fusion rule of point excitations on the domain wall.}
	\label{fig:boundary_connect_M-shape}
\end{figure}

However, when one side of the bulk phase, say $Q$, has a trivial anyon content, there is a unique fusion space (which can be labeled as $\mathbb{V}_{nn'}^{n''}$) for each choice of $n,n',n''\in \calC_N$. In this specific instance, the conventional fusion rule applies to the parton sectors as well.

Importantly, our statement applies even if the bulk phase with a trivial anyon content has a nonzero chiral central charge. A nontrivial example is the so-called $E_8$ state~\cite{Kitaev2006solo}. A proof of our claim is presented in Appendix~\ref{appendix:gapped boundary}. The key idea is that trivial anyon content implies a new type of entropic constraint. This new constraint allows us to prove a strengthening of the isomorphism theorem in which the underlying subsystems can undergo a topology change. We sketched this idea in Fig.~\ref{fig:boundary_connect_M-shape}, deferring the details to Appendix~\ref{appendix:gapped boundary}.

\section{Fusion rules}\label{sec:fusion_rules}

So far, we have defined a number of different superselection sectors and their fusion spaces. In this section, we will study their fusion rules.

To put our work into a context, let us recall the fusion rules in the bulk. Formally, we can write
\begin{equation}
    a\times b = \sum_{c} N_{ab}^c c,
\end{equation}
where $a, b,$ and $c$ are superselection sectors in the bulk; as we discussed in Section~\ref{sec:fusion_from_entanglement}, these are associated with the extreme points of the information convex sets of an annulus. $N_{ab}^c$ is the fusion multiplicity of $a$ and $b$ fusing into $c$.

In Ref.~\cite{SKK2019}, we were able to derive the following facts. 
\begin{equation}
\begin{aligned}
N_{ab}^c &= N_{ba}^c \\
N_{1a}^c &= N_{a1}^c = \delta_{a,c} \\
\forall a, \exists !\, \bar{a} \,\,\text{ s.t. }\,\, N_{ab}^1 &= \delta_{b,\bar{a}} \\
N_{ab}^c &= N_{\bar{b}\bar{a}}^{\bar{c}} \\
\sum_i N_{ab}^i N_{ic}^d &= \sum_j N_{aj}^d N_{bc}^j.
\end{aligned}
\label{eq:axioms_bulk_sec_fusion}
\end{equation} 
The first line says the fusion rule is commutative. The second line says the fusion with the vacuum is trivial. The third line implies that anti-sector is unique. The fourth line is a symmetry of the fusion multiplicity involving the replacement of sectors with their anti-sectors. The last line says the composition of fusion multiplicities is associative.

Furthermore, the quantum dimensions -- defined in terms of the entropy difference Eq.~\eqref{eq:d_a_definition} -- are constrained by the fusion multiplicities by the following equation:
\begin{equation}
d_a d_b = \sum_c N_{ab}^c d_c. \label{eq:quantum_dimension_bulk}
\end{equation}
In fact, this equation completely determines the set of (positive) quantum dimensions because the fusion multiplicities satisfy Eq.~\eqref{eq:axioms_bulk_sec_fusion}. It follows from this constraint that $d_1=1$ and $d_{a}= d_{\bar{a}} \ge 1$ for any $a$. Furthermore, $d_a$ is quantized in the sense that it cannot take an arbitrary value; for example, it cannot take any value in the interval $(1,\sqrt{2})$.

The primary purpose of this section is to derive identities on the fusion multiplicities analogous to these equations. We further derive the quantization of the quantum dimensions of parton sectors by relating them to these fusion multiplicities. We shall go through the fusion spaces described in Fig.~\ref{fig:fusion_summary} and derive their respective fusion rules.

\subsection{Fusion rules for $O$-type sectors} \label{Sec:fusion_alpha}

As a starter, let us first consider the fusion space formed by two sectors in $\mathcal{C}_O$ fusing into another sector in $\mathcal{C}_O$. We shall refer to these sectors as $\alpha, \beta,$ and $\gamma$. This fusion space can be defined over the information convex set over the blue subsystem described in Fig.~\ref{fig:alphabetagamma}, with the appropriately chosen superselection sectors. Formally, we can write this as
\begin{equation}
\alpha\times \beta = \sum_{\gamma} N_{\alpha\beta}^{\gamma} \,\gamma.
\end{equation}
\begin{figure}[h]
	\centering
\includegraphics[scale=1]{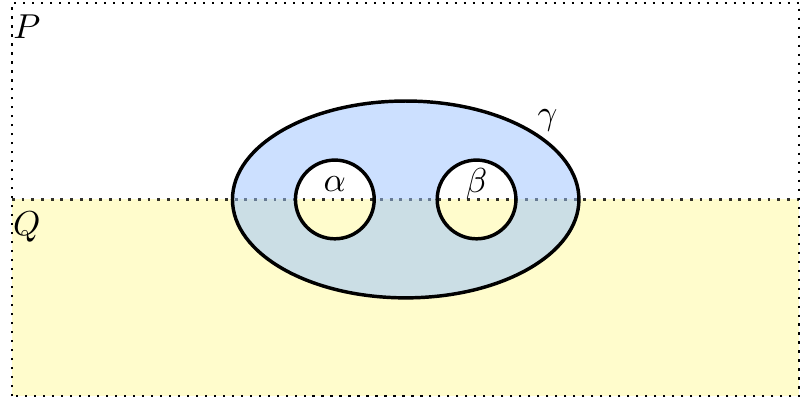}
	\caption{The subsystem choice and sector labels relevant to the fusion space $\mathbb{V}_{\alpha\beta}^{\gamma}$.}
	\label{fig:alphabetagamma}
\end{figure}

The fusion rules of the point-like superselection sectors on the domain wall are very similar to those of the bulk superselection sectors. We first summarize the results and provide some basic explanations. A discussion on the proof will then follow.

The following facts about the fusion multiplicities $\{ N_{\alpha\beta}^{\gamma} \}$ are derived from our assumptions.
\begin{equation}
\begin{aligned}
N_{1\alpha}^{\gamma} &= N_{\alpha 1}^{\gamma} = \delta_{\alpha,\gamma} \\
\forall \alpha, \exists !\, \bar{\alpha} \,\,\text{ s.t. }\,\, N_{\alpha\beta}^1 &= \delta_{\beta,\bar{\alpha}} = \delta_{\alpha,\bar{\beta}}  \\
N_{\alpha\beta}^{\gamma} &= N_{\bar{\beta}\bar{\alpha}}^{\bar{\gamma}} \\
\sum_{i\in \calC_O} N_{\alpha\beta}^i N_{i\gamma}^{\delta} &= \sum_{j\in \calC_O} N_{\alpha j}^{\delta} N_{\beta\gamma}^j.
\end{aligned}
\label{eq:wall_excitation_fusion}
\end{equation}
First, let us compare these identities with the bulk identities in Eq.~\eqref{eq:axioms_bulk_sec_fusion}. Every identity in Eq.~\eqref{eq:wall_excitation_fusion} has an analogous identity in the bulk. However, one bulk identity is generally violated in this context. Specifically, $N_{\alpha\beta}^{\gamma}\ne N_{\beta\alpha}^{\gamma}$ in general, in contrast to the identity $N_{ab}^c = N_{ba}^c$. Intuitively, this is because there is no room to permute two domain wall sectors.\footnote{This does not imply domain wall sectors are confined onto the domain wall. They are not. See Section~\ref{sec:string} for an explanation of this point.}

There is an identity which relates the quantum dimensions $\{d_{\alpha}\}$ to the fusion multiplicities $\{N_{\alpha\beta}^{\gamma}\}$:
\begin{equation}
d_{\alpha} d_{\beta} =\sum_{\gamma} N_{\alpha \beta}^{\gamma} d_{\gamma} . \label{eq:quantum_dimension_wall}
\end{equation}
This identity is analogous to Eq.~(\ref{eq:quantum_dimension_bulk}).
It completely determines the set of quantum dimensions $\{d_{\alpha} \}$ because the fusion multiplicities satisfy Eq.~\eqref{eq:wall_excitation_fusion}. Then it follows that $d_1=1$ and $d_{\alpha}= d_{\bar{\alpha}} \ge 1$ for $\forall \alpha\in\calC_O$. Furthermore, the quantum dimension $d_{\alpha}$ is quantized, just like its bulk counterpart. This completes the summary of the fusion properties of $O$-type superselection sectors.

In terms of proofs, Eq.~\eqref{eq:quantum_dimension_wall} follows from the same line of argument explained in Section~\ref{sec:fusion_from_entanglement}. Also, the proofs of the triviality of the vacuum and the associativity relation, [the first and fourth lines of Eq.~\eqref{eq:wall_excitation_fusion}], are identical to their bulk counterparts. We refer the readers to Ref.~\cite{SKK2019} for these proofs.

However, the proofs on the two properties involving the anti-sectors [the second and the third lines of Eq.~\eqref{eq:wall_excitation_fusion}] need to be modified a bit.

\subsubsection{Proofs related to anti-sectors}
Below, we derive the fact that, for each $\alpha \in \calC_O$, there is a unique anti-sector $\bar{\alpha} \in \calC_O$, such that 
\begin{equation}
\begin{aligned}
N_{\alpha\beta}^1 &= \delta_{\beta,\bar{\alpha}} = \delta_{\alpha,\bar{\beta}}  \\
N_{\alpha\beta}^{\gamma} &= N_{\bar{\beta}\bar{\alpha}}^{\bar{\gamma}} .
\end{aligned}\label{eq:to_proof}
\end{equation}

To prove these facts, it will be convenient to instead prove the following weaker statements. 
\begin{enumerate}
	\item [(i)] $\forall\alpha \in \mathcal{C}_O$, $\exists! \overrightarrow{\alpha}, \overleftarrow{\alpha} \in \mathcal{C}_O$ s.t.	
	$N_{\alpha \beta}^1 = \delta_{\beta,\overrightarrow{\alpha}} =  \delta_{\alpha,\overleftarrow{\beta}}$.
	\item [(ii)] $N_{\alpha \beta}^{\gamma} = N_{\overrightarrow{\beta} \overrightarrow{\alpha}}^{\overrightarrow{\gamma}}$.
\end{enumerate} 
Statement (i) means any $\alpha \in \mathcal{C}_O$ has a ``left anti-sector'' $\overleftarrow{\alpha}$ and a ``right anti-sector'' $\overrightarrow{\alpha}$.

These two statements as a whole is weaker than Eq.~\eqref{eq:to_proof}. Nevertheless, with the established triviality of the vacuum, $N_{1\alpha}^{\gamma} = N_{\alpha 1}^{\gamma} = \delta_{\alpha,\gamma}$, we can derive Eq.~\eqref{eq:to_proof} from these two weaker statements.
To see why, first note that $\overleftarrow{\overrightarrow{\alpha}} = \overrightarrow{\overleftarrow{\alpha}}= \alpha$, which follows from statement (i) alone; this is because statement (i) implies $N^1_{\alpha \overrightarrow{\alpha}}=N^1_{\overleftarrow{\overrightarrow{\alpha}} \overrightarrow{\alpha}}=1$ and $N^1_{\overleftarrow{\alpha} \alpha} =N^1_{\overleftarrow{\alpha} \overrightarrow{\overleftarrow{\alpha}} }=1$.
Moreover, statement (i) and the triviality of the vacuum imply that $1=\overrightarrow{1}= \overleftarrow{1}$; this is because statement (i) implies $N^1_{1 \overrightarrow{1}}= N^1_{\overleftarrow{1} 1} =1$.
Next, we choose $\gamma=1$ and $\alpha=\overleftarrow{\beta}$ for statement (ii). We see that
$N^1_{\overleftarrow{\beta} \beta}= N^1_{\overrightarrow{\beta} \beta} =1, \forall \beta$.
Thus, $\overleftarrow{\alpha} =\overrightarrow{\alpha}$, $\forall \alpha\in \calC_O$. In other words, the left anti-sector and the right anti-sector are identical.
Therefore, there is a unique anti-sector for every superselection sector  in $\mathcal{C}_O$. We denote the unique anti-sector of $\alpha$ as $\bar{\alpha}$. Plugging this result into the two statements, we arrive at Eq.~\eqref{eq:to_proof}.

We have seen that we only need to prove the two statements above in order to derive Eq.~\eqref{eq:to_proof}. Below, we provide these proofs.

\begin{figure}[h]
	\centering
\includegraphics[scale=1]{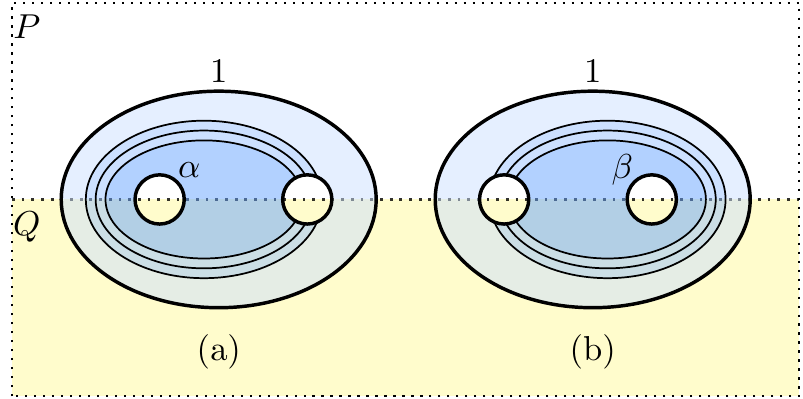}
	\caption{ (a) Merging extreme points associated with $1\in \calC_O$ and $\alpha\in \calC_O$. (b) Merging extreme points associated with $1\in \calC_O$ and $\beta \in \calC_O$. }
	\label{fig:merging_anti_alpha}
\end{figure}

Let us first focus on statement (i), the uniqueness of the anti-sector maps $\alpha \to \overrightarrow{\alpha}$ and $\alpha \to \overleftarrow{\alpha}$.
The idea is similar to the proof of Proposition~4.9 of Ref.~\cite{SKK2019}. Specifically, we can merge an extreme point carrying the sector $\alpha\in\calC_O$ with another extreme point carrying the sector $1\in \calC_O$. See Fig.~\ref{fig:merging_anti_alpha}(a). Here, the annulus that carries the sector $\alpha$ is inside the annulus that carries $1$. The existence of the merged state implies that $\forall \alpha, \,\exists \beta $ s.t. $N^1_{\alpha\beta}\ge 1$.
The entropy of the merged state can be obtained in two different ways, leading to the following equation:
\begin{equation}
2\ln d_{\alpha} = \ln d_{\alpha} + \ln (\sum_{\beta} N_{\alpha\beta}^{1} d_{\beta}), \label{eq:for_2'}
\end{equation}
where we have used the fact that $N_{1\beta}^1 =\delta_{ \beta,1 }$. Equation.~(\ref{eq:for_2'}) further simplifies into 
\begin{equation}
d_{\alpha} = \sum_{\beta} N_{\alpha\beta}^{1} d_{\beta}. \label{eq:L_2'}
\end{equation}
With the exact same approach, we can derive the following identity: 
\begin{equation}
d_{\beta} = \sum_{\alpha} N_{\alpha\beta}^{1} d_{\alpha} \label{eq:R_2'}
\end{equation}
by considering the merging process in Fig.~\ref{fig:merging_anti_alpha}(b).\footnote{Unlike the bulk version of the proof, we cannot ``rotate'' the 2-hole disk on the domain wall to switch the two holes. This is why we need the merging process Fig.~\ref{fig:merging_anti_alpha}(b).} 
For a chosen $\alpha$, we pick a $\beta$ such that $N^1_{\alpha\beta} \ge 1$. (As we have discussed, such a choice always exists.) For such $\beta$, $d_{\alpha} \geq d_{\beta}$.  Similarly, for a chosen $\beta$ there must be at least one $\alpha$ such that $N_{\alpha \beta}^1 \geq 1$ and $d_{\beta} \geq d_{\alpha}$. Therefore, $d_{\alpha}=d_{\beta}$ if $N_{\alpha\beta}^1\ge 1$. Obviously, these two quantum dimensions cannot be equal to each other if for the chosen $\alpha$, $N_{\alpha \beta}^1 \geq 1$ for more than one choice of $\beta$, nor can this happen if $N^1_{\alpha\beta}>1$. Therefore, we conclude that statement (i) is true. 

As a byproduct of this analysis, we have also found that
\begin{equation}
d_{\alpha}=d_{\overleftarrow{\alpha}}=d_{\overrightarrow{\alpha}}. \label{eq:apendix_B_dim_arrow}
\end{equation}

\begin{figure}[h]
	\centering
\includegraphics[scale=1]{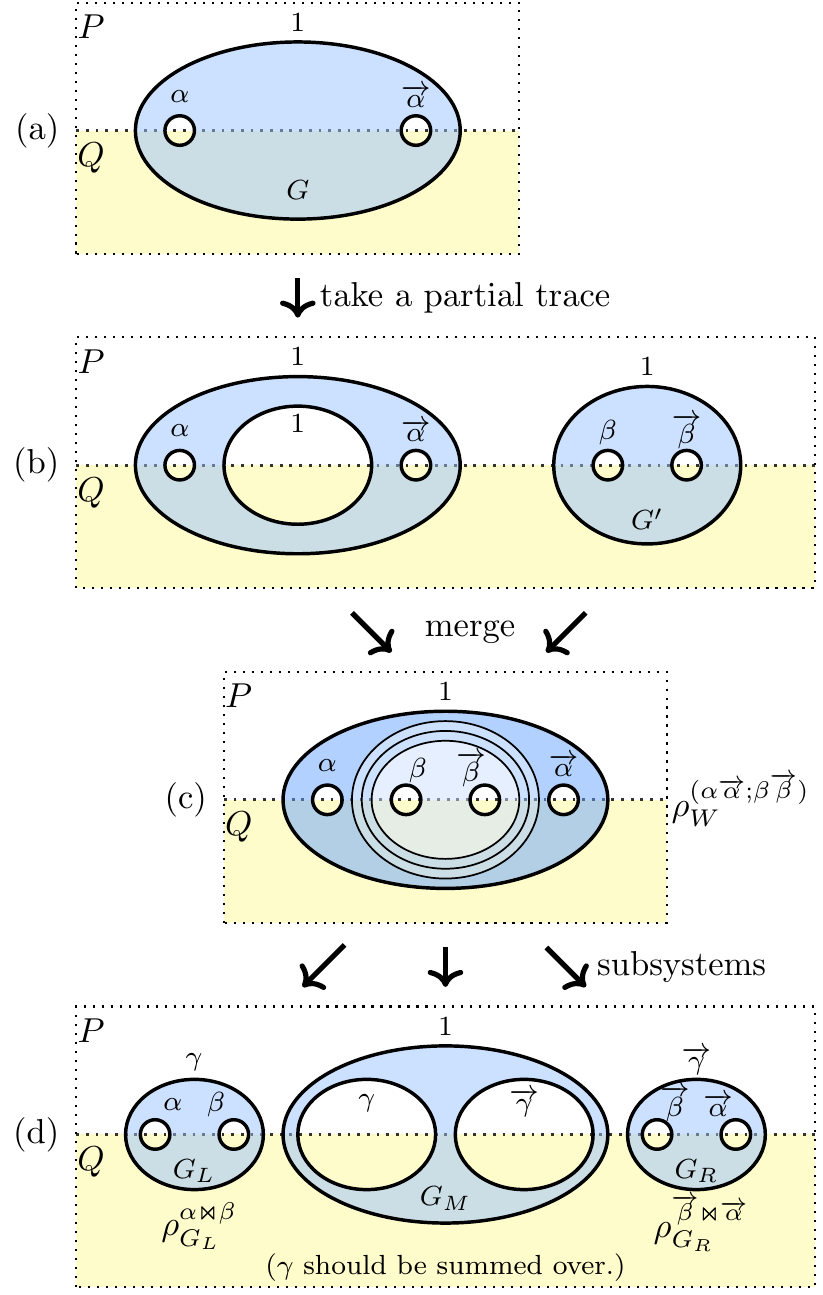}
	\caption{The overall picture of the proof of statement (ii).}
	\label{fig:4_hole_combined}
\end{figure}

Now, let us prove statement (ii), namely $N_{\alpha \beta}^{\gamma} = N_{\overrightarrow{\beta} \overrightarrow{\alpha}}^{\overrightarrow{\gamma}}$. This proof is similar to the proof of the bulk version (Proposition 4.10 of Ref.~\cite{SKK2019}), but with some modifications. 

The overall picture of the derivation is depicted in Fig.~\ref{fig:4_hole_combined}. The density matrix in Fig.~\ref{fig:4_hole_combined}(a) is the unique element of $\Sigma^1_{\alpha\overrightarrow{\alpha}}(G)$, where $G$ is the depicted subsystem. After taking a partial trace, we merge this density matrix with the unique element of $\Sigma^1_{\beta\overrightarrow{\beta}}(G')$, where $G'$ is the subsystem on the right side of Fig.~\ref{fig:4_hole_combined}(b).
The resulting 4-hole disk $W$ is depicted in Fig.~\ref{fig:4_hole_combined}(c).
The key object in the proof is the density matrix obtained from this merging process, which we shall refer to as $\rho_W^{(\alpha\overrightarrow{\alpha};\beta\overrightarrow{\beta})}$.

To see why the merged state $\rho_W^{(\alpha\overrightarrow{\alpha};\beta\overrightarrow{\beta})}$ helps in the proof of statement (ii), we consider its reduced density matrices on subsystems $G_L$, $G_M$ and $G_R$ depicted in Fig.~\ref{fig:4_hole_combined}(d). In general, while we have fixed the sectors $\alpha, \overrightarrow{\alpha},\beta,\overrightarrow{\beta}$, the sectors on the outer boundary of $G_L$ and $G_R$ are, in general, a  mixture. By inspecting the subsystem $G_M$, we see that the superselection sectors on the outer boundary of $G_L$ and $G_R$ must fuse to identity. Thus, we can denote the sectors as $\gamma$ and $\overrightarrow{\gamma}$ respectively. While there can be multiple possible choices of $\gamma$, we measure the sector $\overrightarrow{\gamma}$ on the outer boundary of $G_R$ whenever we measure the sector $\gamma$ on the outer boundary of $G_L$.  Therefore, the probability of finding the sector $\gamma$ on the outer boundary of $G_L$ equals to the probability of finding the sector $\overrightarrow{\gamma}$ on the outer boundary of $G_R$. Formally, we can write this fact as
\begin{equation}
P(\gamma \vert \rho_{G_L}^{(\alpha\overrightarrow{\alpha}; \beta \overrightarrow{\beta})})= P(\overrightarrow{\gamma} \vert \rho^{(\alpha\overrightarrow{\alpha}; \beta \overrightarrow{\beta})}_{G_R}). \label{eq:P_gamma}
\end{equation}

The next step is to calculate both sides of Eq.~\eqref{eq:P_gamma}. The key observation is the fact that the reduced density matrices of $\rho_W^{(\alpha \overrightarrow{\alpha};\beta \overrightarrow{\beta})}$ on $G_L$ and $G_R$ are the maximmum-entropy states with the respective superselection sector choices. 
In other words, 
\begin{equation}
\begin{aligned}
\Tr_{W\setminus G_L} \, \rho_W^{(\alpha \overrightarrow{\alpha} ;\beta \overrightarrow{\beta})}&= \rho_{G_L}^{\alpha \merge \beta} \\
\Tr_{W\setminus G_R} \, \rho_W^{(\alpha \overrightarrow{\alpha} ;\beta \overrightarrow{\beta})} &=  \rho_{G_R}^{\overrightarrow{\beta} \merge \overrightarrow{\alpha}},
\end{aligned}
\label{eq:merge_temp}
\end{equation}
where $\rho_{G_L}^{\alpha \merge \beta}$ and $\rho_{G_R}^{\overrightarrow{\beta} \merge \overrightarrow{\alpha}}$ can be obtained by merging two annuli associated with the specified superselection sectors.

Let us study the consequence of Eq.~\eqref{eq:merge_temp}, deferring the proof of Eq.~\eqref{eq:merge_temp} to Appendix~\ref{appendix:alpha_dual}. 
By using the fact that $\rho_{G_L}^{\alpha \protect\merge \beta}$ is the maximum-entropy state consistent with the chosen sectors $\alpha$ and $\beta$, we obtain
\begin{equation}
P(\gamma \vert \rho_{G_L}^{(\alpha\overrightarrow{\alpha}; \beta \overrightarrow{\beta})}) = \frac{N_{\alpha \beta}^{\gamma} d_{\gamma}}{d_{\alpha} d_{\beta}}.
\end{equation}
Similarly, 
\begin{equation}
P(\overrightarrow{\gamma} \vert \rho^{(\alpha\overrightarrow{\alpha}; \beta \overrightarrow{\beta})}_{G_R})  = \frac{N_{\overrightarrow{\beta}\overrightarrow{\alpha}}^{\overrightarrow{\gamma}} d_{\overrightarrow{\gamma}} }{d_{\overrightarrow{\beta}} d_{\overrightarrow{\alpha}}}.
\end{equation}
Because $d_{\alpha}= d_{\overrightarrow{\alpha}}$ (see Eq.~\eqref{eq:apendix_B_dim_arrow}), we conclude $N_{\alpha \beta}^{\gamma} =N_{\overrightarrow{\beta}\overrightarrow{\alpha}}^{\overrightarrow{\gamma}}$, as we claimed. This completes the proof of statement (ii).

In conclusion, we have justified Eq.~\eqref{eq:wall_excitation_fusion}.

\subsection{Fusion onto the domain wall}

In this section, we discuss the fusion rules of anyons, i.e., the bulk superselection sectors, onto the domain wall. We will use $a,b,c, \ldots \in \mathcal{C}_P$ to denote the anyons on the $P$ side and $x,y,z\ldots \in \mathcal{C}_Q$ to denote the anyons on the $Q$ side.

The fusion of anyons onto domain walls gives rise to fusion spaces that involve the superselection sectors of the bulk and the domain wall. One may consider moving an anyon $a \in \mathcal{C}_P$ onto the domain wall; moving an anyon $x \in \mathcal{C}_Q$ onto the domain wall; or alternatively, bringing a pair of anyons $a \in \mathcal{C}_P$ and $x \in \mathcal{C}_Q$ onto the domain wall. These processes can be formally written as
\begin{eqnarray}
a &=& \sum_{\alpha} N_{a}^{\alpha} \, \alpha,\\
x &=& \sum_{\alpha} N_{x}^{\alpha}\, \alpha, \\
a\times x &=& \sum_{\alpha} N_{ax}^{\alpha} \,\alpha,
\end{eqnarray}
where the fusion multiplicities $\{N_{a}^{\alpha} \}$, $\{N_{x}^{\alpha} \}$ and $\{ N_{ax}^{\alpha} \}$ are again non-negative integers. Here, $\alpha$ is in $\mathcal{C}_O$.

The fusion multiplicities have the following physical interpretation. $N_a^1$ is relevant to the process of condensing anyon $a$ onto the domain wall. What we mean by condensing is that if $N_a^1\ge 1$, it is possible to move an anyon $a$ onto the domain wall and annihilate it by a local process. This also means that if $N_a^1\ge 1$, we can create a single anyon $a$ in the bulk with a string operator attached to the domain wall. The physical interpretation of $N_x^1$ is similar. A pair $(a,x)$ with $N_{ax}^1\ge 1$ can be simultaneously annihilated (or created) in the vicinity of the gapped domain wall.  Similarly, $N_{a}^{\alpha}$ determines whether it is possible to fuse an anyon $a$ onto the domain wall and turn it into a domain wall sector $\alpha \in \mathcal{C}_O$.

\begin{figure}[h]
	\centering
\includegraphics[scale=1]{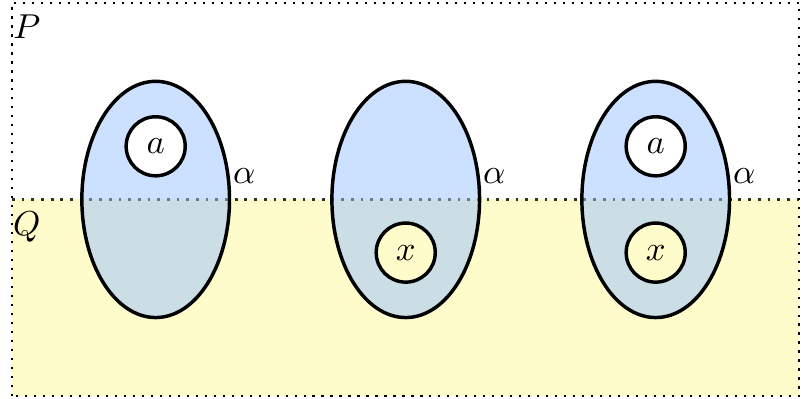}
	\caption{Three basic subsystem types that are relevant to the fusion of anyons onto the domain wall: (a) is relevant to $\{ N_{a}^{\alpha} \}$, (b) is relevant to $\{ N_{x}^{\alpha} \}$ and  (c) is relevant to  $\{ N_{ax}^{\alpha} \}$. }
	\label{anyon_to_wall_3_basic}
\end{figure}

The concrete rules that govern these processes can be deduced from three types of subsystems described in Fig.~\ref{anyon_to_wall_3_basic}. Repeating the analysis in Section~\ref{Sec:fusion_alpha}, we obtain the following results.

\begin{itemize}
	\item Fusion of the vacuum:
	\begin{eqnarray}
	\textrm{among }\{ N_a^{\alpha}\}:&& \,\,\,\, N_1^{\alpha}=\delta_{ \alpha,1 },\\
	\textrm{among }\{ N_x^{\alpha}\}:&& \,\,\,\, N_1^{\alpha}=\delta_{ \alpha,1 },\\
	\textrm{among }\{ N_{ax}^{\alpha}\}:&& \,\,\,\, N_{a 1}^{\alpha} = N_{a}^{\alpha},\,\,\,\, N_{1x}^{\alpha} = N_{x}^{\alpha}.
	\end{eqnarray}
\item Relations between quantum dimensions and fusion multiplicities
\begin{eqnarray}
d_a &=& \sum_{\alpha} N_a^{\alpha} d_{\alpha}\\
d_x &=& \sum_{\alpha} N_x^{\alpha} d_{\alpha}\\
 d_a d_x &=& \sum_{\alpha} N_{ax}^{\alpha} d_{\alpha}.  \label{eq:sum_alpha}
\end{eqnarray}
\item Relation between a fusion space and an fusion space formed by the anti-sectors:
\begin{eqnarray}
N_a^{\alpha}&=& N_{\bar{a}}^{\bar{\alpha}}\\
N_x^{\alpha}&=& N_{\bar{x}}^{\bar{\alpha}}\\
 N_{ax}^{\alpha}  &=& N_{\bar{a}\bar{x}}^{\bar{\alpha}}.
\end{eqnarray}
\item Associativity conditions:
\begin{eqnarray}
N_{ax}^{\alpha} &=& \sum_{\beta,\gamma} N_{a}^{\beta} N_{x}^{\gamma} N_{\beta\gamma}^{\alpha}\\
N_{ax}^{\alpha} &=& \sum_{\beta,\gamma} N_{x}^{\beta} N_{a}^{\gamma} N_{\beta\gamma}^{\alpha}\\
\sum_c N_{ab}^c N_c^{\gamma} &=& \sum_{\alpha,\beta} N_{a}^{\alpha} N_{b}^{\beta} N_{\alpha\beta}^{\gamma}\\
\sum_z N_{xy}^z N_z^{\gamma} &=& \sum_{\alpha,\beta} N_{x}^{\alpha} N_{y}^{\beta} N_{\alpha\beta}^{\gamma}\\
\sum_{\beta,\gamma} N_{\beta\gamma}^{\alpha} N_{ax}^{\beta} N_{by}^{\gamma} &=& \sum_{c,z}N_{cz}^{\alpha} N_{ab}^c N_{xy}^z  \label{eq:imply_stable_condition}
\end{eqnarray}
\end{itemize}

These identities are consistent with the proposal in the existing literature. For example, Eq.~(\ref{eq:imply_stable_condition}) implies the so called  \emph{stable condition} of Ref.~\cite{Lan2015}.

\subsection{ Fusions with $\mathbb{N}$-type and $\mathbb{U}$-type sectors}\label{Sec:consistency_N_U}

In this section, we discuss a few things related to the fusion of $\mathbb{N}$- and $\mathbb{U}$-type sectors.  Specifically, we consider the fusion spaces $\{ N_{\mathcal{N} \mathcal{U}}^{\alpha} \}$ and $\{ N_{\mathcal{U}}^{\alpha} \}$, defined in Figs.~\ref{fig:NbbandUbb_fusion_space}(a) and (b), respectively. 

\begin{figure}[h]
	\centering
\includegraphics[scale=1]{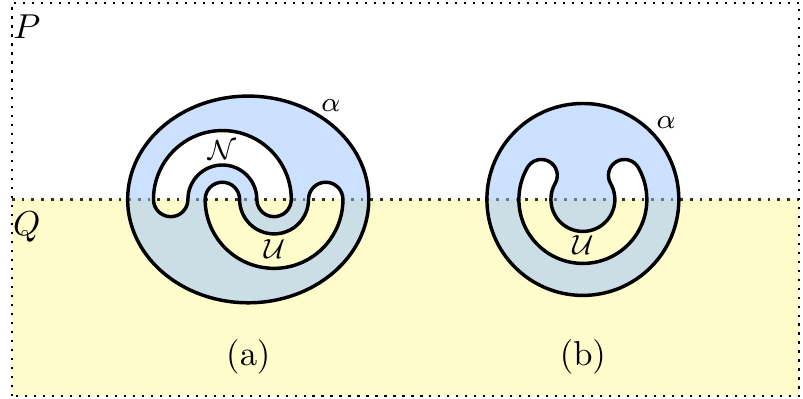}
	\caption{ (a) A subsystem relevant to the fusion involving the sectors $\mathcal{N} \in \mathcal{C}_{\mathbb{N}}$, $\mathcal{U} \in \mathcal{C}_{\mathbb{U}}$ and $\alpha \in \mathcal{C}_O$. (b) A subsystem relevant to the fusion involving the sector $\mathcal{U} \in \mathcal{C}_{\mathbb{U}}$ and $\alpha \in \mathcal{C}_O$.}
	\label{fig:NbbandUbb_fusion_space}
\end{figure}

These fusions are a bit abstract, so one may wonder why we consider them in the first place. A simple reason is that the constraints on these fusion spaces enable us to prove some fundamental properties of the simpler $N$-type and $U$-type parton sectors. These sectors do not obey the ordinary fusion rule, as we have briefly discussed in Section~\ref{Sec:quasi_fusion_subsection}. However, we can still derive nontrivial facts about their fusion by embedding those sectors into the $\mathbb{N}$- and $\mathbb{U}$-type sectors (see Section~\ref{sec:bbN_bbU}). 

%two integers
The most notable implication is Proposition~\ref{Prop:d_n_quantization}, which implies that the quantum dimensions of the partons, \emph{i.e.,} $\{d_n\}$, are uniquely determined by two sets of fusion multiplicities $\{N_{\alpha\beta}^{\gamma}\}$ and $\{N_{\mu(n)}^{\alpha}\}$. In that sense, these quantum dimensions are ``quantized.'' Moreover, unlike the $N$- and $U$-type sectors, the $\mathbb{N}$- and $\mathbb{U}$-sectors do allow a conventional definition of fusion space.

To study these fusion rules, let us begin by showing some simple properties when one of the sectors involved is the vacuum sector.
\begin{Proposition}
	Among the fusion multiplicities $\{ N_{\mathcal{N} \mathcal{U}}^{\alpha} \}$, we have 
	\begin{equation}
	N_{1\mathcal{U}}^{1} =\delta_{\mathcal{U},1} \label{eq:U11_fusion}
	\end{equation}
	and 
	\begin{equation}
	N_{\mathcal{N}1}^{1} =\delta_{\mathcal{N},1}. \label{eq:N11_fusion}
	\end{equation}
	Furthermore, among the fusion multiplicities $\{ N_{\mathcal{U}}^{\alpha} \}$, we have
	\begin{equation}
	N_{1}^{\alpha} = \delta_{ \alpha,1 }. \label{eq:1alpha_fusion}
	\end{equation}
\end{Proposition}

\begin{figure}[h]
	\centering
\includegraphics[scale=1]{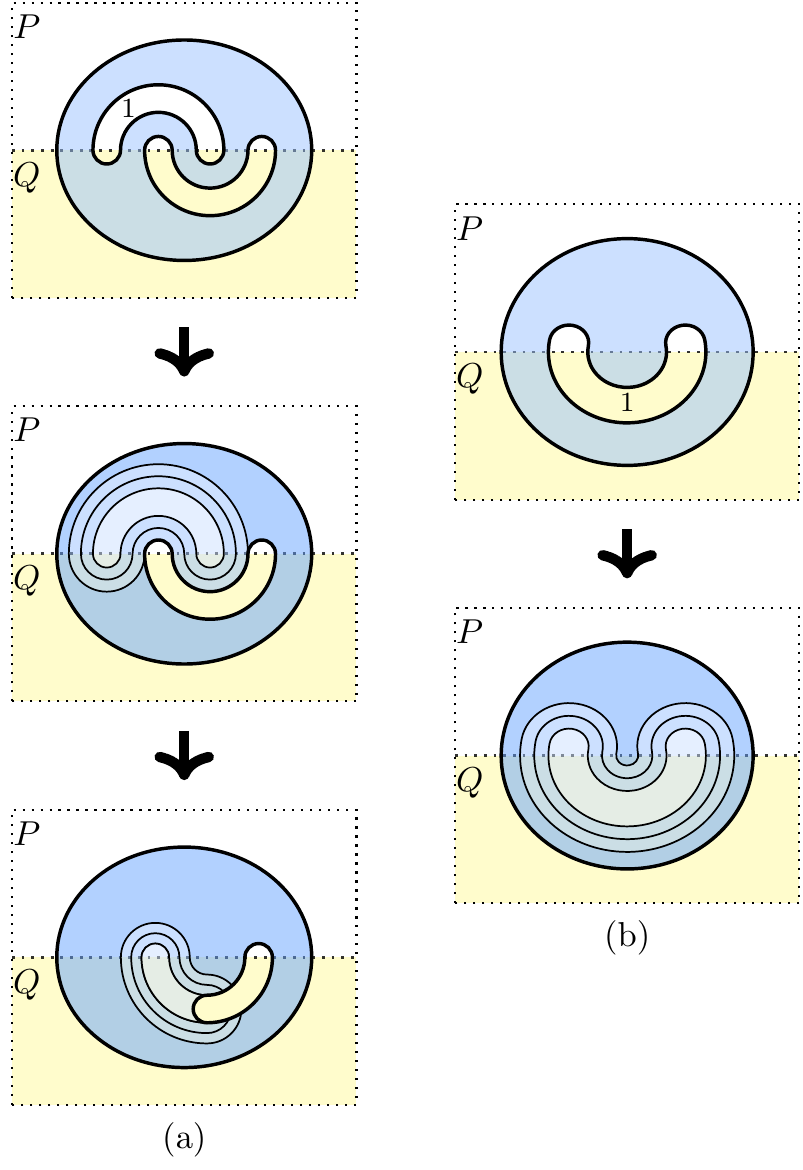}
	\caption{(a) A two-step merging process for $\mathcal{N}=1$. It is useful in the proof of $N^1_{1\mathcal{U}}= \delta_{\mathcal{U},1}$.  (b) A merging process for $\mathcal{U}=1$. It is useful in the proof of $N^{\alpha}_{\mathcal{U}=1}= \delta_{\alpha,1}$.}
	\label{NU_fill_the_vacuum}
\end{figure}
The key idea of the proof is to use the properties of the vacuum sector to design merging processes that can ``fill'' a hole. See Fig.~\ref{NU_fill_the_vacuum} for an illustration of the relevant merging processes.
\begin{proof}
	Let us first prove Eqs.~\eqref{eq:U11_fusion} Eq.~\eqref{eq:N11_fusion}. These two proofs are analogous to each other, so we only discuss the proof of Eq.~\eqref{eq:U11_fusion}. Recall that the fusion multiplicities $\{N_{\mathcal{N}\mathcal{U}}^{\alpha} \}$ are associated with the subsystem in Fig.~\ref{fig:NbbandUbb_fusion_space}(a), where the superselection sectors involved are $\alpha \in \calC_O$, $\mathcal{N}\in \calC_{ \mathbb{N}}$ and $\mathcal{U}\in \calC_{\mathbb{U}}$. To show $N_{1\mathcal{U}}^1 =\delta_{\mathcal{U},1}$, it suffices to prove the following statement. If $\alpha=1$ and $\mathcal{N}=1$, the density matrix of the blue region in Fig.~\ref{fig:NbbandUbb_fusion_space}(a) is equal to the reduced density matrix obtained from $\sigma$.
	
	To see why this is the case, we consider the two-step merging process shown in Fig.~\ref{NU_fill_the_vacuum}(a). This merging process is possible when $\mathcal{N}=1$ is the vacuum sector. The first step ``fills'' the $N$-shaped hole with the vacuum, by merging density matrix over the subsystem on the top of Fig.~\ref{NU_fill_the_vacuum} to the density matrix obtained from the reference state. This is possible because $\mathcal{N}=1$; the density matrix on the surrounding $\mathbb{N}$-shaped region is identical to that of the reference state, satisfying the requisite condition for the merging theorem (Theorem~\ref{thm:merging_info_convex_set}). The second step fills part of the $U$-shaped hole and turns it into a point-like area intersecting with the domain wall. This step is possible because $\mathcal{N}=1$ implies that one of the parton labels of the sector $\mathcal{U}$ must be the vacuum. (More precisely, $\mathcal{U}$ carries a pair of $N$-type parton sectors and a pair of $U$-type parton sectors. The $N$-type parton sector on the left is the vacuum sector due to $\mathcal{N}=1$.)
	
	After the two-step merging process, we obtain an element of the information convex set on an $O$-shaped region. Therefore the same $O$-type sector $\alpha\in \calC_O$ must appear on both boundaries of the $O$-shaped region, i.e., the annulus on the bottom of Fig.~\ref{NU_fill_the_vacuum}(a). If $\alpha=1$, the hole can be filled. The end result is the reduced density matrix of the reference state on a disk-like region on the domain wall. Therefore we must have $\mathcal{U}=1$. Furthermore, the density matrix labeled by $\alpha=1$, $\mathcal{N}=1$ and $\mathcal{U}=1$ in Fig.~\ref{fig:NbbandUbb_fusion_space}(a) is unique. This completes the proof of Eq.~\eqref{eq:U11_fusion}.
	
	The proof of Eq.~\eqref{eq:1alpha_fusion} follows from a  similar line of reasoning. The merging process in Fig.~\ref{NU_fill_the_vacuum}(b) is possible when $\mathcal{U}=1$. The end result is a disk on the domain wall. This implies that the $\alpha\in \calC_O$ in the original density matrix must be the vacuum sector. Moreover, this density matrix is unique. Therefore Eq.~\eqref{eq:1alpha_fusion} holds. This completes the proof.
\end{proof}

\begin{figure}[h]
	\centering
\includegraphics[scale=1]{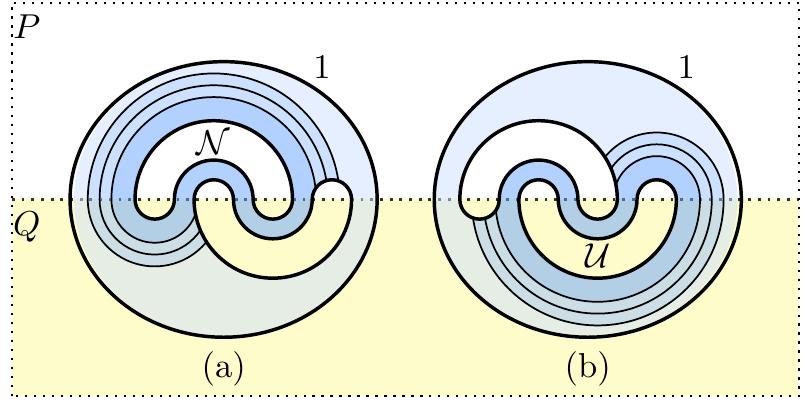}
	\caption{(a) A merging process involving a pair of regions, which carry $1\in \mathcal{C}_O$ and $\mathcal{N}\in \calC_{\mathbb{N}}$ respectively. (b) A merging process involving a pair of regions, which carry $1\in \mathcal{C}_O$ and $\mathcal{U}\in \calC_{\mathbb{U}}$ respectively.}
	\label{fig:Merging_N_1_new}
\end{figure}

Next, we show that there is a nontrivial isomorphism between $\mathcal{C}_{\mathbb{N}}$ and $\mathcal{C}_{\mathbb{U}}$. If the gapped domain wall is trivial, this result would be trivially true because the $\mathbb{N}$-shaped and $\mathbb{U}$-shaped subsystems can be smoothly deformed into each other. However, this fact is less obvious when the domain wall is nontrivial.

\begin{Proposition}
	There is an isomorphism\footnote{It is possible to define another isomorphism between $\calC_{\mathbb{N}}$ and $\calC_{\mathbb{U}}$ by considering the mirror image of Fig.~\ref{fig:Merging_N_1_new}. This isomorphism can be different from $\varphi$.} 
	\begin{equation}
	\varphi: \calC_{\mathbb{N}}\to \calC_{\mathbb{U}} \label{def:varphi}
	\end{equation}
	such that
	\begin{equation}
	N^1_{\mathcal{NU}}= \delta_{\mathcal{U},\varphi(\mathcal{N})}, \label{eq:1NU}
	\end{equation}
	and 
	\begin{equation}
	d_{\mathcal{N}} = d_{ \varphi(\mathcal{N}) }. \label{eq:quantum_dim_varphi}
	\end{equation}
\end{Proposition}
\begin{proof}
	We will use the merging processes described in Fig.~\ref{fig:Merging_N_1_new}, using the logic used in the proof of Proposition~4.9 of Ref.~\cite{SKK2019}.
	
	We consider the merging process in Fig.~\ref{fig:Merging_N_1_new}(a), which involves two sectors  $1\in \calC_O$ and $\mathcal{N}\in \calC_{ \mathbb{N} }$. Note that any $\mathcal{N}\in \calC_{ \mathbb{N} }$ is allowed in the merging process. This implies that $\sum_{\mathcal{U}} N^1_{\mathcal{N}\mathcal{U}}\ge 1, \,\forall \mathcal{N}\in\calC_{ \mathbb{N} }$. By calculating the entropy difference [between two sector choices (1, $\mathcal{N}$) and (1, 1)]in two different ways, we find 
	\begin{equation}
	2 \ln d_{\mathcal{N}} = \ln d_{\mathcal{N}} + \ln \sum_{\mathcal{U}} N_{\mathcal{N}\mathcal{U}}^1 d_{\mathcal{U}}. \label{eq:comparing_1N}
	\end{equation}
	The left-hand side is the entropy difference based on the entropy on the $\mathbb{N}$-shaped subsystem. The right-hand side is obtained by solving a maximization problem on the merged region.
	In the calculation of the right-hand side, we have applied Eq.~(\ref{eq:U11_fusion}), which implies $\sum_{\mathcal{U}} N_{1\mathcal{U}}^1 d_{\mathcal{U}} =1$. There is an analogous equation for the merging process in Fig.~\ref{fig:Merging_N_1_new}(b). 
	By simplifying these two equations, we find
	\begin{equation}
	\begin{aligned}
	d_{\mathcal{N}} &= \sum_{\mathcal{U}} N^1_{\mathcal{N}\mathcal{U}} d_{\mathcal{U}},\\
	d_{\mathcal{U}} &= \sum_{\mathcal{N}} N^1_{\mathcal{N}\mathcal{U}} d_{\mathcal{N}}. \label{eq:a_strong_constraint_NU}
	\end{aligned}
	\end{equation}
	Equation~\eqref{eq:a_strong_constraint_NU} is a strong constraint. For a chosen $\mathcal{N}$, pick a sector  $\mathcal{U}$ that satisfies $N^1_{\mathcal{N}\mathcal{U}} \ge 1$. It follows from Eq.~\eqref{eq:a_strong_constraint_NU} that $d_{\mathcal{N}}\ge d_{\mathcal{U}}$ and $d_{\mathcal{U}}\ge d_{\mathcal{N}}$.  We have used the fact that the multiplicities are nonnegative integers and that the quantum dimensions are positive. Therefore, for every choice of $\mathcal{N}$, there is a unique $\mathcal{U}$ for which the fusion multiplicity obeys $N^1_{\mathcal{N}\mathcal{U}} = 1$. Moreover, we have $d_{\mathcal{U}}=d_{\mathcal{N}}$ for this choice. For the same $\mathcal{N}$, a different choice of $\mathcal{U}$ gives $N^1_{\mathcal{N}\mathcal{U}} = 0$.
	
	Let $\varphi$ be the map from $\calC_{\mathbb{N}}$ to $\calC_{ \mathbb{U}}$, mapping a sector $\mathcal{N}\in \calC_{\mathbb{N}}$ to the unique sector $\mathcal{U}\in\calC_{ \mathbb{U}}$ satisfying $N^1_{\mathcal{N}\mathcal{U}}=1$. $\varphi$ is bijective because there is an inverse map obtained by the same argument, choosing $\mathcal{U}$ instead of $\mathcal{N}$ first. This completes the proof. 
\end{proof}

For later purpose, it will be convenient to consider an embedding $\mu: \calC_N \hookrightarrow \calC_{\mathbb{U}}$, defined as $\mu=\varphi\circ \eta_N$. Here, $\eta_N$ is the embedding defined in Eq.~(\ref{eq:embedding_N2NN}). From Eqs.~(\ref{eq:eta_relation}) and (\ref{eq:quantum_dim_varphi}), it follows that
\begin{equation}
d_{\mu(n)} = d_n^2,\quad \forall\, n\in \calC_N. \label{eq:quantum_dim_mu}
\end{equation}

There is an important subtlety about the fusion space $N_{\mathcal{U}}^{\alpha}$. (See Fig.~\ref{fig:NbbandUbb_fusion_space} for the relevant subsystems.) For some $\mathcal{U}\in \mathcal{C}_{\mathbb{U}}$, $\sum_{\alpha} N_{\mathcal{U}}^{\alpha}$ may vanish. This can happen when $\calC_U$ contains two or more elements. Recall that each $\mathcal{U}\in \calC_{ \mathbb{U}}$ can be labeled by four parton sectors, two of which are $U$-type parton sectors. If a sector $\mathcal{U}$ labels an extreme point of the subsystem in Fig.~\ref{fig:NbbandUbb_fusion_space}(b), the two $U$-type parton sectors must be in the vacuum because the disk that sits in the slot of the $\mathbb{U}$-shaped susbsystem is in the reference state. Importantly, this implies that generally
\begin{equation}
d_{\mathcal{U}} \ne \sum_{\alpha} N_{{\mathcal{U}}}^{\alpha} d_{\alpha}.
\end{equation}
However, the following proposition \emph{is} true.

\begin{Proposition}
	For every $\mathcal{U}$ that satisfies $\sum_{\alpha} N_{\mathcal{U}}^{\alpha}\ge 1$, the quantum dimension of $\mathcal{U}$ is given by
	\begin{equation}
	d_{\mathcal{U}}= \sum_{\alpha} N_{{\mathcal{U}}}^{\alpha} d_{\alpha}. \label{eq:f(U)}
	\end{equation}
\end{Proposition}

\begin{figure}[h]
	\centering
\includegraphics[scale=1]{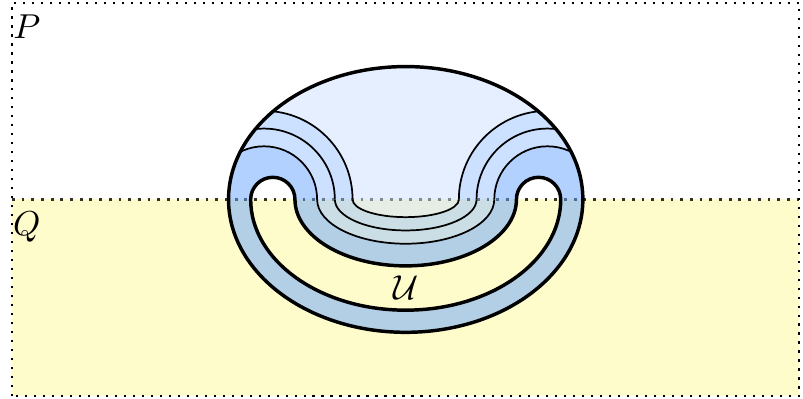}
	\caption{Merging of a disk with a $\mathbb{U}$-shaped region. This $\mathbb{U}$-shaped region carries a sector $\mathcal{U}\in \calC_{\mathbb{U}}$ that satisfies $\sum_{\alpha} N_{\mathcal{U}}^{\alpha} \geq 1$. }
	\label{fig:merge_U_and_disk}
\end{figure}

\begin{proof}
	We consider the merging process in Fig.~\ref{fig:merge_U_and_disk}. Note that any $\mathcal{U}$ such that $ \sum_{\alpha} N_{\mathcal{U}}^{\alpha} \ge 1$ is allowed. We can calculate the entropy difference between an arbitrary (allowed) choice of $\mathcal{U}$ and the vacuum ($1\in \calC_{\mathbb{U}})$. There are two different ways to do the same calculation. After comparing them, we can derive
	\begin{equation}
	2 \ln d_{\mathcal{U}} = \ln d_{\mathcal{U}} + \ln \sum_{\alpha} N_{\mathcal{U}}^{\alpha} d_{\alpha}. \label{eq:consistency_1U}
	\end{equation}
	The left-hand side of this equation is the entropy difference calculated using the entropy of the $\mathbb{U}$-shaped subsystem. The right-hand side is calculated by solving a maximization problem on the merged region. In the derivation of Eq.~(\ref{eq:consistency_1U}), we have applied Eq.~(\ref{eq:1alpha_fusion}) to conclude that $\sum_{\alpha} N_1^{\alpha} d_{\alpha}=1$.
	By simplifying Eq.~(\ref{eq:consistency_1U}), we obtain Eq.(\ref{eq:f(U)}).
\end{proof}

The following statement will be useful for proving the ``quantization'' of $d_n$.
\begin{Proposition} The multiplicities satisfy:
	\begin{equation}
	N^1_{\mathcal{U}}= \sum_{n\in\calC_N} \delta_{{\mathcal{U}},\mu(n)}. \label{eq:N_U^1}
	\end{equation}
\end{Proposition}
\begin{proof}
	If $N_{\mathcal{U}}^1 \ge 1$, then  $\exists ! n \in \calC_N$ such that $\mathcal{U} =\mu(n)$. This follows from Eq.~(\ref{eq:1NU}). Furthermore,  when $\mathcal{U} =\mu(n)$, we can ``fill'' the hole with the sector $\eta_N(n)$. The associated merging process can be inverted, which implies that $N^1_{\mu(n)}=1$. This completes the proof.
\end{proof}
From this proposition, we deduce the quantization of $d_n$.
\begin{Proposition}\label{Prop:d_n_quantization}
	The quantum dimensions of the $N$-type parton sectors, $\{d_n \}_{n\in \calC_N}$ are uniquely determined by two sets of integers $\{N_{\alpha\beta}^{\gamma}\}$ and $\{N_{\mu(n)}^{\alpha}\}$ according to
\begin{equation}
d_n^2=\sum_{\alpha} N_{\mu(n)}^{\alpha} d_{\alpha}.
\end{equation}
	Furthermore, $d_n\ge 1$ and $d_n$ cannot be in the interval $(1,\sqrt{2})$,  $\forall n\in \calC_N$.
\end{Proposition}
\begin{proof}
	Note that  $\sum_{\alpha} N_{\mu(n)}^{\alpha}\ge N_{\mu(n)}^1.$ Moreover, $N_{\mu(n)}^1 = 1$ because of Eq.~(\ref{eq:N_U^1}).
	It follows that
	\begin{equation}
	\begin{aligned}
	d_n^2 &= d_{\mu(n)} \\
	&= \sum_{\alpha} N_{\mu(n)}^{\alpha} d_{\alpha}.
	\end{aligned}\label{eq:f(n)}
	\end{equation}
	The first line follows from Eq.~(\ref{eq:quantum_dim_mu}). The second line follows from Eq.~(\ref{eq:f(U)}). Recall that $\{d_{\alpha}\}$ is uniquely determined by the set of fusion multiplicities $\{N_{\alpha\beta}^{\gamma}\}$ according to Eq.~(\ref{eq:quantum_dimension_wall}). Therefore,  $\{d_n \}_{n\in \calC_N}$ are uniquely determined by two sets of integers $\{N_{\alpha\beta}^{\gamma}\}$ and $\{N_{\mu(n)}^{\alpha}\}$.
	
	Note that, $d_n^2 = 1 + \sum_{\alpha\ne 1} N_{\mu(n)}^{\alpha} d_{\alpha}$, due to $N_{\mu(n)}^1 d_1=1$.
	Furthermore, $\{N_{\mu(n)}^{\alpha} \}$ are non-negative integers and $d_{\alpha}\ge 1$. Therefore, $d_n\ge 1$ and no value in the interval $(1,\sqrt{2})$ is allowed for $d_n$. This completes the proof. Obviously, a similar statement applies to $d_u$ as well.
\end{proof}

\section{Quasi-fusion rules}\label{Sectionquasi_fusion}

As we have briefly discussed already, parton sectors do not have the familiar notion of fusion space; see Section~\ref{Sec:quasi_fusion_subsection}. This phenomena, which we refer to as quasi-fusion, stems from the difference between Figs.~\ref{fig:M-shape_and_quasi_fusion} and \ref{fig:fusion_characterization}. The key point is that the three $N$-shaped subsystems in Fig.~\ref{fig:M-shape_and_quasi_fusion} do not constitute $\partial M$. On the other hand, in Fig.~\ref{fig:fusion_characterization}, the three annuli do constitute the entire $\partial \Omega$. Therefore, the three $N$-shaped subsystems do not generally fix the sector in $\partial M$. When there is more than one sector in $\partial M$, for the chosen parton labels, multiple fusion spaces are involved in the description of quasi-fusion. 

However, this does not mean that the quasi-fusion of parton sectors can be completely arbitrary. In this section, we will explain the basic rules that the parton sectors must obey when they are fused together, focusing on the similarities and differences with the ordinary rule of fusion. We will refer to these rules as \emph{quasi-fusion rules.}

Let us first say that there are some similarities between the fusion rule and the quasi-fusion rule. In particular, the notion of anti-sector is well-defined for the parton sectors as well. Given a sector $n\in \mathcal{C}_N$, one can show that there is a unique sector $\bar{n}\in\mathcal{C}_N$ such that
\begin{equation}
    d_n= d_{\bar{n}} \quad \textrm{and}\quad \bar{\bar{n}}= n.
\end{equation}
Similarly, we can define the anti-sectors for $u\in \mathcal{C}_U$ as well. We will prove these statements in Appendix~\ref{appendix:quasi-fusion}.

To further study the quasi-fusion rules, it is convenient to introduce the following sets.
Let us define $\Sigma_{nn'}^{n''}(M)$ with $n,n',n'' \in \calC_N$ to be the subset of $\Sigma(M)$ consisting of elements that reduce to the extreme points of these three sectors on the three $N$-shaped subsystems.\footnote{Note that $\Sigma_{nn'}^{n''}(M)$ may be an empty set for some choices of $n,n',n'' \in \calC_N$.} With this definition, one can verify the following proposition. We leave the proof in Appendix~\ref{ap:Subsystem_M_related}.

\begin{Proposition}
	The convex set  $\Sigma_{nn'}^{n''}(M)$ with $n,n',n'' \in \calC_N$ satisfies the following properties:
	\begin{enumerate}
		\item Every extreme point of $\Sigma(M)$ is contained in some $\Sigma_{nn'}^{n''}(M)$.
		\item $\cup_{n''} \Sigma_{nn'}^{n''}(M)$ is nonempty for $\forall \, n,n'\in\calC_N$.
		\item $\Sigma_{n1}^{n''}(M)$ is the empty set for $n''\ne n$. For $n''=n$, it has a unique element.
		\item $\Sigma_{1n'}^{n''}(M)$ is the empty set for $n''\ne n'$. For $n''=n'$, it has a unique element.
		\item $\Sigma_{nn'}^{1}(M)$ is the empty set for $n'\ne \bar{n}$. For $n'=\bar{n}$, it has a unique element.
	\end{enumerate}
\label{Prop:Sigma_nn'^n''}
\end{Proposition}
Physically, the content of Proposition~\ref{Prop:Sigma_nn'^n''} should be viewed as a relaxation of the fusion rule. For instance, one can see that two parton sectors can always fuse into some parton sector; see the second statement of Proposition~\ref{Prop:Sigma_nn'^n''}. Moreover, the triviality of the vacuum sector is stated in the third and the fourth result. The fifth result states that a parton sector and its anti-sector can fuse to the vacuum.

Note that $\Sigma_{nn'}^{n''}(M)$ is similar to $\Sigma_{ab}^c(Y)$ in that it may store quantum information. However, unlike $\Sigma_{ab}^c(Y)$, $\Sigma_{nn'}^{n''}(M)$ is generally \emph{not} isomorphic to a state space of some Hilbert space. Moreover, the entropy difference between two extreme points $\rho\in \Sigma_{nn'}^{n''}(M)$ and $\sigma\in  \Sigma_{11}^{1}(M)$ can be
\begin{equation}
	S(\rho)-S(\sigma) \ne \ln d_{n} + \ln d_{n'} + \ln d_{n''}, \label{eq:ee_exotic}
\end{equation}
unlike the extreme points in $\Sigma_{ab}^c(Y)$.

\section{Domain wall topological entanglement entropy}\label{sec:tee}
In this section, we introduce and study the domain wall topological entanglement entropies. These are order parameters that can detect the presence of gapped domain walls. Moreover, for the class of states that we considered, these order parameters are invariant under a small deformation of the subsystem. The following is a list of objects that one can obtain directly from the ground state entanglement entropy. 
\begin{itemize}
\item $\calD_N=\sqrt{\sum_{n\in \calC_N} d_n^2}$.
\item $\calD_U=\sqrt{\sum_{u\in \calC_U} d_u^2}$.
\item $\calD_O=\sqrt{\sum_{\alpha\in \calC_O} d_{\alpha}^2}$.
\end{itemize}
During this analysis, we also prove that
\begin{equation}
    \calD_O^2=\sqrt{\sum_{a\in\calC_P} d_a^2} \cdot \sqrt{\sum_{x\in\calC_Q} d_x^2}
\end{equation}
Furthermore, the total quantum dimension of the snake sectors in $\calC_S$,  defined as $\calD_S=\sqrt{\sum_{s\in \calC_S} d_s^2}$,  and the total quantum dimension of the sectors in $\calC_O^{[1,1]}$, defined as $\calD_{O^{[1,1]}}=\sqrt{\sum_{\alpha\in \calC_O^{[1,1]}} d_{\alpha}^2 } $, are encoded in the ground state as well. This is because they can be expressed as
\begin{eqnarray}
    \calD_S &=& \calD_N \calD_U, \label{eq:D_S}\\
    \calD_{O^{[1,1]}} &=& \frac{\calD_O}{ \calD_N \calD_U}, \label{eq:D_O11}
\end{eqnarray}
where Eq.~\eqref{eq:D_S} follows from Eqs.~\eqref{eq:d_s_consistency} and \eqref{eq:D_O11} follows from Eq.~\eqref{eq:qd_parton_composite}.

These derivations are quite similar to each other. To start with, let us consider 
\begin{equation}
 S_{\text{topo}, N} :=  (S_{BC}+S_{CD}-S_B-S_D)_{\sigma},  \label{eq:def_S_P_to_E}
\end{equation}
where the subsystems $B,$ $C$, abd $D$ are described in Fig.~\ref{fig:S_P_to_E}(a).

We obtain a number of (equivalent) expressions for this quantity; see Proposition~\ref{Prop:S_P_to_E}. As a byproduct of this analysis, we also obtain a nontrivial identity:
\begin{equation}
    \sum_{n\in \mathcal{C}_N} d_n^2 = \sum_{a\in \mathcal{C}_P} N_a^1 d_a.
\end{equation}
\begin{Proposition}\label{Prop:S_P_to_E}
\begin{equation}
\exp\left( S_{\text{topo}, N} \right)= \calD_N^2 = \sum_{a\in \mathcal{C}_P} N_a^1 d_a. \label{eq:S_P_to_E}
\end{equation}
Moreover,
\begin{equation}
 S_{\text{topo}, N} = I(A:C\vert B)_{\sigma} \label{eq:TEE_2nd_subsystem}
\end{equation}
for the subsystem $A, B,$ and $C$ shown in Fig.~\ref{fig:S_P_to_E}(b).
\end{Proposition}

\begin{figure}[h]
	\centering
\includegraphics[scale=1]{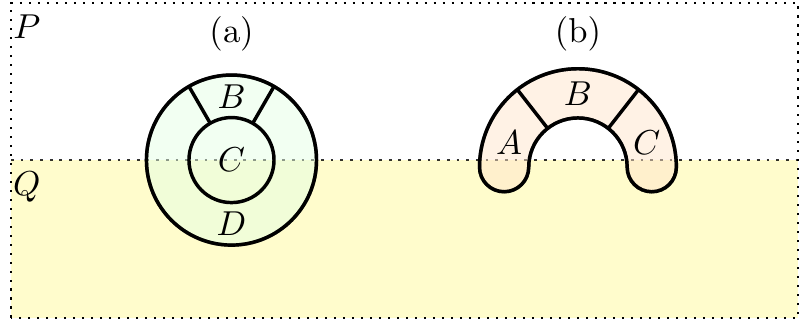}
\caption{The total quantum dimension of $N$-type parton sectors $\calD_N$ shows up in the ground state entanglement entropy for both of these partitions. (a) $S_{\text{topo}, N}=(S_{BC} + S_{CD} - S_C - S_D)_{\sigma}$. (b) $S_{\text{topo}, N} = I(A:C|B)_{\sigma}$. This is a domain wall version of the Levin-Wen partition~\cite{Levin2006}.}
\label{fig:S_P_to_E}
\end{figure}

\begin{figure}[h]
  \centering
\includegraphics[scale=1]{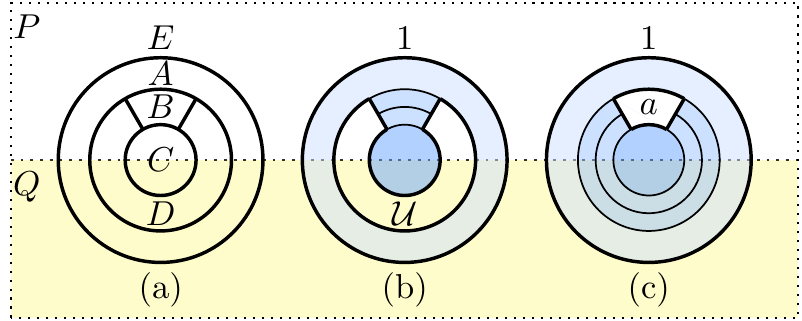}
	\caption{(a) Subsystems $A$, $B$, $C$, $D$ and $E$. (b) The merging process that generates $\tau_{ABC}= \sigma_{AB}\merge \sigma_{BC}$, a state involving different choices of $\mathcal{U}$. (c) The merging process that generates $\lambda_{ACD}= \sigma_{AD}\merge \sigma_{CD}$, a state involving different choices of $a$. }
	\label{fig:Merging_for_TEE}
\end{figure}

\begin{proof}
First, we observe that the entropy combination in Eq.~(\ref{eq:def_S_P_to_E}) can be rewritten as a conditional mutual information. For the  partitions in Fig.~\ref{fig:Merging_for_TEE}(a),
\begin{eqnarray}
 S_{\text{topo}, N} &=& I(A:C\vert B)_{\sigma}, \label{eq:TEE_mutual_1}\\ 
 S_{\text{topo}, N} &=& I(A:C\vert D)_{\sigma}, \label{eq:TEE_mutual_2}
\end{eqnarray}
To derive Eq.~(\ref{eq:TEE_mutual_1}), we use the fact that $S_{AB}= S_{CD}+S_E$ and $S_{ABC}= S_D + S_E$ for the reference state.\footnote{If the reference state is pure, $E$ refers to the complement of $ABCD$. If the reference state is a mixed state, we purify the reference state and let $E$ include the purifying system.} These relations follow from the domain wall version of {\bf A0} and the fact that we could deform the regions with other axioms and SSA. A similar derivation applies to Eq.~(\ref{eq:TEE_mutual_2}).

Let $\tau_{ABC} = \sigma_{AB} \merge \sigma_{BC}$.  It follows that
\begin{equation}
\begin{aligned}
 S_{\text{topo}, N} &= I(A:C\vert B)_{\sigma}     \\   
           &= S(\tau_{ABC}) - S(\sigma_{ABC})\\
           &= \ln (\sum_{\mathcal{U}\in \calC_{\mathbb{U}}} N_{\mathcal{U}}^1 d_{\mathcal{U}})\\
           &= \ln (\sum_{n\in\calC_N} d_{\mu(n)})\\
           &= \ln (\sum_{n \in \calC_N} d_n^2)\\
           &= 2\ln \calD_N. 
\end{aligned}\label{TEE_calculation_ABC}
\end{equation}
The first line is Eq.~(\ref{eq:TEE_mutual_1}). The second line follows from the fact that $\tau_{ABC}$ and $\sigma_{ABC}$ have identical reduced density matrices over $AB$ and $BC$ and that $I(A:C\vert B)_{\tau}=0$. In the third line, we have computed the maximum entropy over the set of density matrices with a sector $1\in \mathcal{C}_O$ and subtracted it from the entanglement entropy of the reference state. [See Fig.~\ref{fig:Merging_for_TEE}(b).] The fourth and the fifth line follows from Eqs.~(\ref{eq:N_U^1}) and (\ref{eq:quantum_dim_mu}), respectively. The last line follows from the definition of $\calD_N$.

For the second half of the main claim, let $\lambda_{ACD} =\sigma_{AD} \merge \sigma_{CD}$.  It follows that
\begin{equation}
\begin{aligned}
 S_{\text{topo}, N} &= I(A:C\vert D)_{\sigma} \\   
&= S(\lambda_{ACD}) - S(\sigma_{ACD})\\
&= \ln (\sum_{a} N_a^1 d_a)
\end{aligned} \label{TEE_calculation_ACD}
\end{equation}
The second line follows from the fact that $\lambda_{ACD}$ and $\sigma_{ACD}$ have identical reduced density matrices over $AD$ and $CD$ and that $I(A:C\vert D)_{\lambda}=0$. In the third line, we computed the maximum entropy over the set of density matrices with a sector $1\in \mathcal{C}_O$ and subtracted it from the entanglement entropy of the reference state. [See Fig.~\ref{fig:Merging_for_TEE}(c).]

From Eqs.~(\ref{TEE_calculation_ABC}) and (\ref{TEE_calculation_ACD}) we conclude Eq.~(\ref{eq:S_P_to_E}). 
For the subsystem choice in Fig.~\ref{fig:S_P_to_E}(b), the mutual information is $I(A:C\vert B)_{\sigma}= 2\ln \calD_N$, which justifies Eq.~(\ref{eq:TEE_2nd_subsystem}).
\end{proof}

Similarly, we can define an analogous quantity $S_{\text{topo}, U}$ by considering a set of subsystems that are mirror images of the aforementioned subsystems along the domain wall. 

Let us emphasize that the pair of domain wall topological entanglement entropies ($S_{\text{topo}, N}$ and $S_{\text{topo}, U}$) contains genuine data about the gapped domain wall. These quantities are not completely determined by the bulk data that defines the phase $P$ and $Q$. Physically, there are two ways to interpret $S_{\text{topo}, N}$. An interpretation of $S_{\text{topo}, N}= \ln\left(\sum_{a\in \mathcal{C}_P} N_a^1 d_a \right)$ is that $S_{\text{topo}, N}$ measures the total amount of anyon condensation from $P$ to the gapped domain wall. Alternatively, one can look at $S_{\text{topo}, N} = \ln \left( \sum_{n\in \mathcal{C}_N} d_n^2\right)$ and say that it measures the total quantum dimension of the $N$-type parton sectors. Both are valid interpretation of the same result.

Another interesting point is that, by definition, $S_{\text{topo}, N}$ is the amount by which the axiom {\bf A1} breaks down in the presence of a gapped domain wall. What is interesting is that the axiom {\bf A1} cannot be broken in an arbitrary way. Instead, there has to be a minimal ``gap,'' which is bounded from below by $\ln 2$.

In fact, we can consider another quantity $S_{\text{topo}, O}$. Unlike the previous ones, $S_{\text{topo},O}$ characterizes the total amount of condensation of anyon pair $a\in \mathcal{C}_P$ and $x \in \mathcal{C}_Q$ to the gapped domain wall.  Let us define $S_{\text{topo}, O}$ as the entropy combination on the reference state
\begin{equation}
S_{\text{topo}, O}:= (S_{BC} +S_{CD}- S_B -S_D)_{\sigma},    \label{eq:S_PQ_to_E}
\end{equation}
where $B$, $C,$ and $D$ are shown in Fig.~\ref{S_PQ_to_E}(a). This quantity is completely determined by the bulk data, unlike $S_{\text{topo}, N}$ and $S_{\text{topo}, U}$; see Proposition~\ref{prop:232}.

\begin{Proposition}
\label{prop:232}
	\begin{equation}
	\exp \left(S_{\text{topo}, O} \right) = \calD_O^2= \sum_{a\in\calC_P,x\in \calC_Q} N_{ax}^1 d_a d_x. \label{eq:S_PQ_identities}
	\end{equation}
	Moreover,
	\begin{equation}
	\exp \left(S_{\text{topo}, O} \right) = \sqrt{\sum_{a\in \calC_P} d_a^2}\cdot \sqrt{\sum_{x\in \calC_Q} d_x^2} \label{eq:S_PQ_bulk_expression}
	\end{equation}
	and
	\begin{equation}
	S_{\text{topo}, O}= I(A:C\vert B)_{\sigma} \label{eq:S_PQ_CMI}
	\end{equation}
	for the $A$, $B$, and $C$ in Fig.~\ref{S_PQ_to_E}(b).
\end{Proposition} 

\begin{figure}[h]
	\centering
\includegraphics[scale=1]{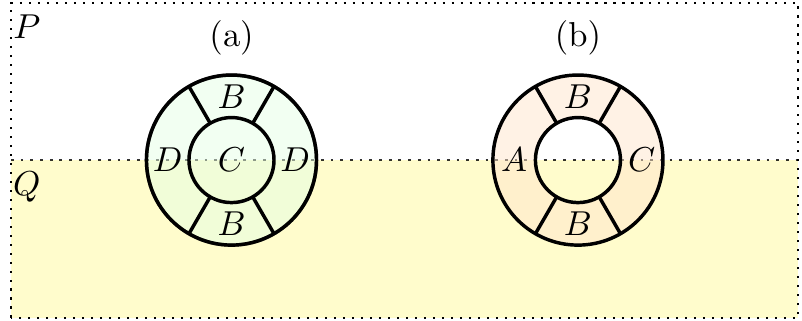}
	\caption{ (a) Subsystems appearing in the definition of $S_{\text{topo}, O}$.
	(b) A Levin-Wen type partition on the domain wall.}
	\label{S_PQ_to_E}
\end{figure}

\begin{figure}[h]
	\centering
\includegraphics[scale=1]{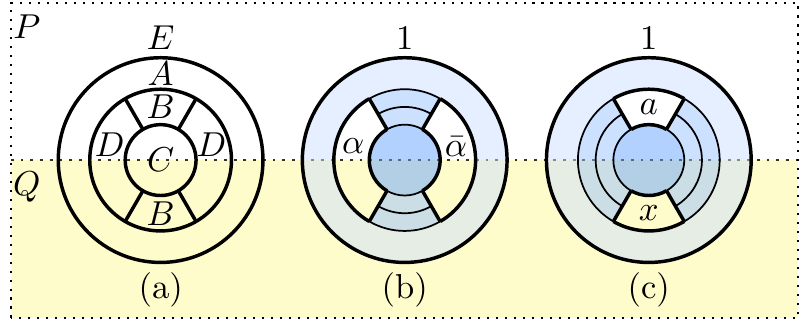}
	\caption{(a) Subsystems $A$, $B$, $C$, $D$ and $E$. (b) The merging process that generates the state $\sigma_{AB}\merge \sigma_{BC}$. Different choices of $\alpha\in \calC_O$ exist in the state. (c) The merging process that generates the state $\sigma_{AD}\merge \sigma_{CD}$. Different choices of $a$ and $x$ exist in the state.}
	\label{fig:Merging_for_TEE2}
\end{figure}

\begin{proof}
	At a technical level, the proof of Eqs.~(\ref{eq:S_PQ_identities}) and (\ref{eq:S_PQ_CMI}) are similar to the proof of Proposition~\ref{Prop:S_P_to_E}. The only difference is the choice of subsystems, shown in Fig.~\ref{fig:Merging_for_TEE2}.

    The derivation of Eq.~(\ref{eq:S_PQ_bulk_expression}) involves the comparison of two calculations of the entropy of the maximum-entropy state in  $\Sigma(F)$, where $F=ACD$ for the partition in Fig.~\ref{fig:Merging_for_TEE2}(a).  We denote the maximum-entropy state as $\widetilde{\rho}_F$. 
    
    One way to calculate the entropy difference $S(\widetilde{\rho}_F) - S(\sigma_F)$ is to solve a maximization problem. This involves writing $\widetilde{\rho}_F$ as a convex combination involving different choices of $\alpha$, $a$, $x$ sectors. By applying the entropy-maximization procedure, which has been repeatedly used in Section~\ref{sec:composite_sectors}, we find
    \begin{equation}
    \begin{aligned}
        S( \widetilde{\rho}_F )- S(\sigma_F) &= \ln (\sum_{a\in \calC_P,x\in\calC_Q} \sum_{\alpha\in \calC_O} N_{ax}^{\alpha} d_a d_x d_{\alpha})\\
                                            &=\ln (\sum_{a\in \mathcal{C}_P} d_a^2) + \ln (\sum_{x\in \mathcal{C}_Q} d_x^2). \label{eq:max-min_1}
    \end{aligned}
    \end{equation}
    In the second line, we have applied Eq.~\eqref{eq:sum_alpha} to simplify the sum.
    
    The second way to calculate the entropy difference $S(\widetilde{\rho}_F) - S(\sigma_F)$ is to calculate the entropy difference on subsystems. For this purpose, we use the merging technique. Consider a third state $\lambda_F := \sigma_{AD}\merge \sigma_{CD}$  obtained by the merging process in Fig.~\ref{fig:Merging_for_TEE2}(c).  We further observe that the state $\widetilde{\rho}_F$ is a merged state, namely $\widetilde{\rho}_F= \widetilde{\rho}_{AD}\merge \sigma_{CD}$. Here, $\widetilde{\rho}_{AD}$ is the maximum-entropy state of $\Sigma(AD)$. Note that the subsystem involved in this merging process is again Fig.~\ref{fig:Merging_for_TEE2}(c); the only difference lies in the difference of the state on $\Sigma(AD)$.
    Instead of computing $S(\widetilde{\rho}_F) - S(\sigma_F)$ directly, we can compute the following two quantities:
    \begin{equation}
        S(\widetilde{\rho}_F) - S(\lambda_F) \quad \text{ and }\quad S(\lambda_F)-S(\sigma_F).\label{eq:709}
    \end{equation}
    Because the merging with the disk-like region $CD$ does not change the entropy difference, the first quantity in Eq.~\eqref{eq:709} equals to $\ln \left(\sum_{\alpha \in \mathcal{C}_O} d_{\alpha}^2 \right)$. The value of the second quantity in Eq.~\eqref{eq:709} can be computed by the entropy-maximization procedure. The end result is $\ln (\sum_{a\in \mathcal{C}_P,x\in\mathcal{C}_Q} N_{ax}^1 d_a d_x)$, leading to the following result:
    \begin{equation}
    S(\widetilde{\rho}_F)- S(\sigma_F)= \ln (\sum_{\alpha \in \mathcal{C}_O} d_{\alpha}^2) + \ln (\sum_{a\in \mathcal{C}_P,x\in\mathcal{C}_Q} N_{ax}^1 d_a d_x). \label{eq:max-min_2}
    \end{equation}
    By comparing Eqs.~(\ref{eq:max-min_1}) and (\ref{eq:max-min_2}) and then using Eq.~\eqref{eq:S_PQ_identities}, one can verify Eq.~(\ref{eq:S_PQ_bulk_expression}).
\end{proof}

\section{String operators}\label{sec:string}

In the bulk, an anyon and its antiparticle can be created by a string-like unitary operator; see the operator $U^{(a,\bar{a})}$ in Fig.~\ref{fig:bulk_string}. In fact, such a unitary operator can be deformed freely. The deformability of these operators follow from axiom \textbf{A0} and \textbf{A1} \cite{SKK2019}, which were briefly discussed in Section~\ref{sec:fusion_from_entanglement}. 

\begin{figure}[h]
	\centering
\includegraphics[scale=1]{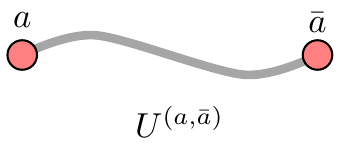}
	\caption{A  unitary string operator $U^{(a,\bar{a})}$ in the bulk. It creates an anyon $a$ and its antiparticle $\bar{a}$.}\label{fig:bulk_string}
\end{figure}

Similarly, we can establish the existence of string-like operators in the vicinity of gapped domain walls. The underlying logic is similar to the discussion in Appendix H of Ref.~\cite{SKK2019}. We will explain how that analysis can be applied to our setup, focusing on the physical meaning. A pair of domain wall excitations $\alpha\in \calC_O$ and $\bar{\alpha} \in \calC_O$ can be created by a string-like unitary operator $U^{(\alpha,\bar{\alpha})}$; see Fig.~\ref{fig:string_wall}(a). Note that, in general, the support of the string stretches into both sides of the bulk. However, the support of this operator can be restricted further to one side if $\alpha$ is either in $\calC_O^{[n,1]}$ or $\calC_O^{[1,u]}$. In the former case, the string can be restricted to $P$; in the latter case, the string can be restricted to $Q$. See Figs.~\ref{fig:string_wall}(b) and \ref{fig:string_wall}(c).

\begin{figure}[h]
	\centering
\includegraphics[scale=1]{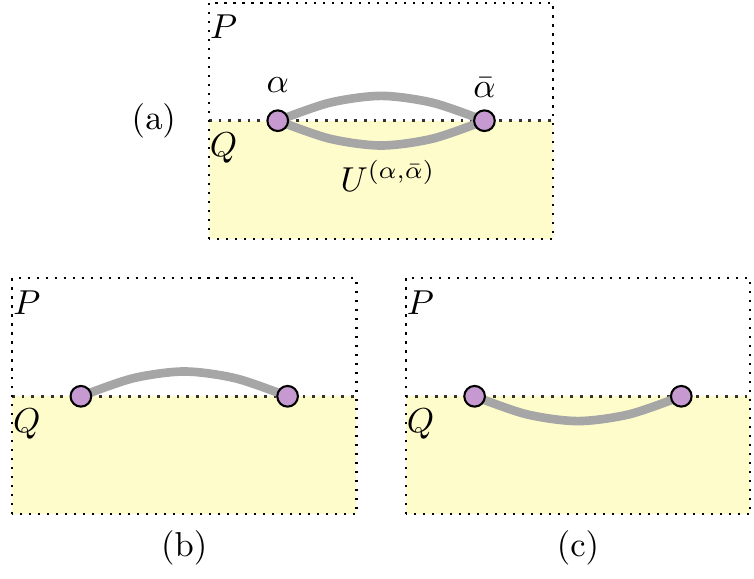}
	\caption{ (a) A unitary string operator $U^{(\alpha,\bar{\alpha})}$ near the gapped domain wall. It can create a general domain wall excitation pair $\alpha$ and $\bar{\alpha}$, where $\alpha\in \calC_O$. (b) The support of $U^{(\alpha,\bar{\alpha})}$ can be restricted to $P$, when $\alpha \in \cup_{n\in \calC_N} \calC_O^{[n,1]}$.  (c) The support of $U^{(\alpha,\bar{\alpha})}$ can be restricted to $Q$, when $\alpha \in \cup_{u\in \calC_U} \calC_O^{[1,u]}$. }
	\label{fig:string_wall}
\end{figure}

\begin{figure}
	\centering
\includegraphics[scale=1]{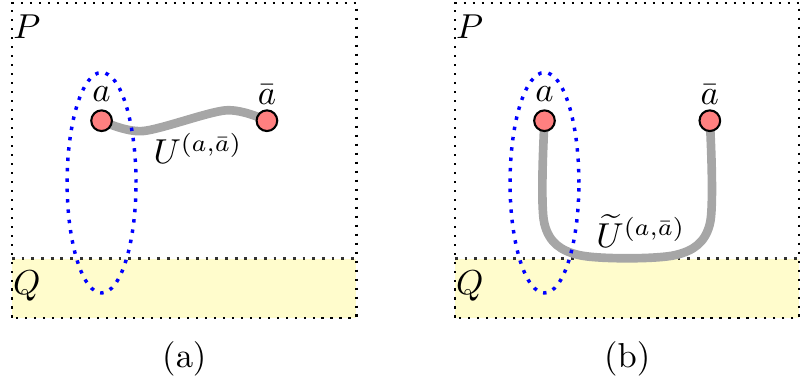}
	\caption{ Detecting the domain wall sector carried by an anyon $a$ with a measurement on the dotted circle.
		In the two figures, the pair of anyons $a$ and $\bar{a}$ are created in two ways, which are generically inequivalent. (a) The string $U^{(a,\bar{a})}$ is in the bulk. (b) The string $\widetilde{U}^{(a,\bar{a})}$ touches the domain wall.} \label{fig:deconfined}
\end{figure}

The domain wall superselection sectors in $\calC_O$ are defined in the vicinity of the domain wall. However, we should not interpret them as excitations confined to the domain wall. As illustrated in Fig.~\ref{fig:deconfined}, anyons in the bulk can carry a domain wall sector as well. That sector makes sense if we perform a measurement on a region that is (i) anchored on the domain wall and (ii) surrounding the anyon. 

Once we create an anyon with a superselection sector $a$ [using a string operator $U^{(a,\bar{a})}$ shown in Fig.~\ref{fig:deconfined}(a)], the domain wall sector for that anyon will be generally indefinite. The probability of finding a sector $\alpha \in \calC_O^{[n,1]}$ on the annulus around the dotted circle in Fig.~\ref{fig:deconfined}(a) can be computed. The result is:
\begin{equation}
P_{(a\to \alpha)} = \frac{N_{a}^{\alpha} d_{\alpha}}{d_a}.
\end{equation}

On the other hand, if we apply a string operator that can touch the domain wall, e.g., the string operator $\widetilde{U}^{(a,\bar{a})}$ shown in Fig.~\ref{fig:deconfined}(b), the anyon with a superselection sector $a$ can carry any sector $\alpha \in \calC_O^{[n,1]}$ that has $N_{a}^{\alpha}\ge 1$. The specific sector $\alpha$ depends on the choice of the string operator. This sector can be detected on the annulus around the dotted circle in Fig.~\ref{fig:deconfined}(b).

\begin{figure}
	\centering
\includegraphics[scale=1]{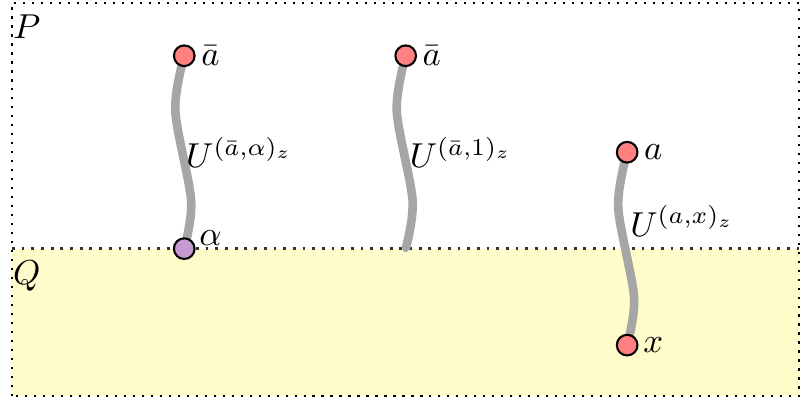}
	\caption{Three sets of string operators $\{U^{(\bar{a},\alpha)_z} \}$, $\{U^{(\bar{a}, 1)_z} \}$ and $U^{(a,x)_z} $. The label $z$ parametrize the corresponding fusion space.
		A string operator $U^{(\bar{a},\alpha)_z}$ that creates $\bar{a}$ and $\alpha$ exists when $N_a^{\alpha}\ge 1$.  A string operator $U^{(\bar{a}, 1)_z} $  that creates $\bar{a}$ exists when $N_a^{1}\ge 1$. A string operator $U^{(a,x)_z} $ that creates $a$ and $x$ exists when $N_{ax}^1\ge 1$.}
	\label{fig:string_fusion_to_wall}
\end{figure}

We can also identify the set of string operators that connect a bulk excitation and a domain wall excitation; see Fig.~\ref{fig:string_fusion_to_wall}. Specifically, there exists a set of string operators $\{ U^{(\bar{a},\alpha)_z} \}$ that creates $\bar{a} \in \mathcal{C}_P$ and $\alpha \in \mathcal{C}_O$ if $N_a^{\alpha}\ge 1$. The label $z$ is introduced to parametrize the states in a $N_a^{\alpha}$ dimensional Hilbert space. For the special case $N_{a}^{\alpha=1}\ge 1$, it is possible to choose the string operators such that no excitation is created on the domain wall. For the set of string operators $\{ U^{(\bar{a},1 )_z} \}$ in Fig.~\ref{fig:string_fusion_to_wall}, the strings can be deformed freely, including the endpoint on the domain wall; note that the string cannot detach from the domain wall. 

Similarly, there exists a set of string operators $\{ U^{(a, x)_z} \}$ for any pair of anyons $a\in \calC_P$ and $x\in\calC_Q$ that satisfies $N_{ax}^1\ge 1$; see Fig.~\ref{fig:string_fusion_to_wall}. Here, $z$ parametrize the state in a $N_{ax}^1$ dimensional Hilbert space. Such a string can simultaneously create an anyon $a$ on the $P$ side and an anyon $x$ on the $Q$ side. This type of string is related to the phenomena of tunneling an anyon through the wall. In general, although all anyons can fuse onto the domain wall, only a subset can condense on (or tunnel through) the domain wall.

\section{Discussion}
\label{sec:discussion}

In this paper, we have observed a remarkably rich structure of many-body quantum entanglement that arises from our simple assumptions (Fig.~\ref{fig:axioms_all}). The notion of superselection sectors, fusion spaces, and fusion multiplicities in the vicinity of a gapped domain wall were all deduced from these assumptions. Moreover, we derived a set of nontrivial identities relating these objects. These results extend the bulk version of the entanglement bootstrap method~\cite{SKK2019} to a broader physical context.

While some of these results are known, others are new. In particular, we have identified a new type of arguably fundamental superselection sector called the parton sector. While prior studies did not necessarily exclude the possibility of such sectors, parton sectors provide more fine-grained information about the superselection sectors on the gapped domain wall. This is because they subdivide the known superselection sectors for the point-like excitations on the domain wall~\cite{KitaevKong2012,Kong2014}.

We could also derive an expression for the domain wall topological entanglement entropy, which may prove useful for detecting the presence of nontrivial domain walls numerically. A more in-depth discussion of this will appear in our companion paper~\cite{EntanglementBootstrap_short}.

The main philosophy behind these derivations was simple. We assume that the subleading contribution to the entanglement entropy obeys a set of sensible constraints but do not assume anything more than that. What is surprising is that merely specifying these rules constrain the subleading terms so strongly that we can derive a large number of nontrivial constraints. Moreover, these constraints are strong enough to imply that fundamental objects such as the quantum dimensions cannot have an arbitrary value; for example, the quantum dimensions of the parton sectors cannot possess value in the range of $(1, \sqrt{2})$. In other words, these values are ``quantized.'' 

For deriving the fusion rules, the essential observation seems to be the following. Given any sufficiently thick subsystem $\Omega$, an extreme point of $\Sigma(\Omega)$  always carries a well-defined set of sectors on the thickened boundary of $\Omega$. Furthermore, once the sectors in the thickened boundary of $\Omega$ are fixed, the remaining degrees of freedom form a convex subset of $\Sigma(\Omega)$. This subset is isomorphic to the state space of some finite-dimensional Hilbert space. As such, it makes sense to refer to them as \emph{fusion spaces}. One can relate the superselection sectors defined over different subsystems by either smoothly deforming one subsystem to another, merging the subsystems, or reducing to a smaller subsystem.

The key advantage of this approach is that the basic emergent laws that govern the low-energy excitations of topologically ordered systems are derived from our principle that has been elucidated in Fig.~\ref{fig:axioms_all}. This led us to the discovery of parton sectors, which would have been difficult to envision otherwise.

It will be good to understand the mathematical framework that can accurately describe our findings. This framework should naturally explain, among other things, the parton sectors, all kinds of composite sectors as well as the quasi-fusion rules of the parton sectors. The number of snake sectors is unbounded. It is not clear to us how to determine them from a finite set of data in general. In ordinary fusion rule, the state space in which two sectors fuse to another sector is isomorphic to the state space of some finite-dimensional Hilbert space. For the quasi-fusion of parton sectors, the state space in which two sectors ``fuse'' to another parton can be a \emph{convex hull} of more than two state spaces, each of which is isomorphic to the state space of some finite-dimensional Hilbert space. In other words, there is a piece of information about the ``composite charge'' that remains unspecified. Would category theory continue to be the right framework to describe these properties? We do not know the answer to this question. 

We deduced our results from the axioms in Fig.~\ref{fig:axioms_all}. Because the axioms are expected to hold on very general classes of physical systems, including the gapped domain wall between 2D chiral phases, our results are expected to hold with the same generality. Nonetheless, readers may wonder whether nontrivial parton sectors exist in known solvable models for gapped domain walls~\cite{Beigi2010,KitaevKong2012}. Our theory predicts that they do because, one can verify our axioms and the fact that $S_{\textrm{topo}, N}$ or $S_{\textrm{topo}, U}$ is nonzero for some of these models.\footnote{In particular, we can consider the models described in Section~6 of Ref.~\cite{Beigi2010}, which depends on a finite group $U\subseteq G\times G'$ and a $2$-cocycle. If we let $U=\{ (k,k) : k\in K\}$, where $K$ is a subgroup of both $G$ and $G'$ and let the $2$-cocycle be that in the trivial class, then one can verify $S_{\textrm{topo}, N}= \ln \frac{\vert G\vert }{\vert K\vert}$ and $S_{\textrm{topo}, U}= \ln \frac{\vert G'\vert }{\vert K\vert}$.}

% Closely related physical setup.
For future work, closely related physical setups can be studied. A codimension-2 defect can separate two different gapped domain walls.  As a special case, if the two topologically ordered systems are identical 2D phases, the gapped domain wall may have endpoints. If, in addition, the domain walls are transparent,  the endpoints are isolated point-like regions that break condition {\bf A1}; these endpoints are topological defects~\cite{Bombin2010}. Our approach can be generalized to accommodate these physical situations. We also expect our approach to be generalized to the gapped domain walls between higher-dimensional topologically ordered systems.

% Mention future direction about braiding
Another direction to pursue is the ``braiding properties'' of these sectors. For example, what kind of anyons can disappear on a gapped domain wall? What are the consistency conditions the domain wall excitations have to satisfy that are beyond the fusion rules? If the two phases that lie on each side of the domain wall are both nontrivial, can we ``factorize'' the boundary version of the Verlinde formula~\cite{2019arXiv190108285S} further into simpler ones?
What are the necessary conditions the two phases have to satisfy for the gapped domain wall to exist? Can we show the chiral central charge of the two phases must match from merely the axioms in Fig.~\ref{fig:axioms_all}? Answers to these questions can be nontrivial. Progress in this direction may be made from a generalization of the method in Ref.~\cite{2020PhRvR...2b3132S}, which derives the mutual braiding statistics of anyons from the same set of axioms.

\section*{Acknowledgment}
We thank Peter Huston and especially Corey Jones for describing a few mathematical structures that the parton sectors cannot correspond to, in existing categorical frameworks.
IK’s work was supported by the Simons Foundation It from Qubit Collaboration and by the Australian Research Council via the Centre of Excellence in Engineered Quantum Systems (EQUS) project number CE170100009.
BS is supported by the National Science Foundation under Grant No. NSF DMR-1653769, University of California Laboratory Fees Research Program, grant LFR-20-653926, as well as the Simons Collaboration on Ultra-Quantum Matter, grant 651440 from the Simons Foundation.

\appendix

\section{Extensions of axioms}\label{appendix:extensions_of_axioms}
In this section, we explain why the axioms described in Fig.~\ref{fig:axioms_all} implies that the same set of conditions hold at an arbitrarily large scale. The proof of this statement is essentially identical to the proof of the analogous statement in Ref.~\cite{SKK2019}. These are straightforward consequences of the strong subadditivity of entropy~\cite{Lieb1973}.

In this appendix, we shall refer to the entropy conditions in Fig.~\ref{fig:axioms_all} in red color as bulk or domain wall version of {\bf A0}, and we refer to the entropy conditions in Fig.~\ref{fig:axioms_all} in green color as bulk and domain wall version of {\bf A1}. The fact that condition {\bf A0} and {\bf A1} hold on arbitrarily large length scales in the bulk is derived in Proposition 3.3 of Ref.~\cite{SKK2019}. Therefore, we shall focus on the extension of the domain wall version of condition {\bf A0} and {\bf A1}. The ideas behind these proofs are similar.

 \begin{figure}[h]
	\centering
\includegraphics[scale=1]{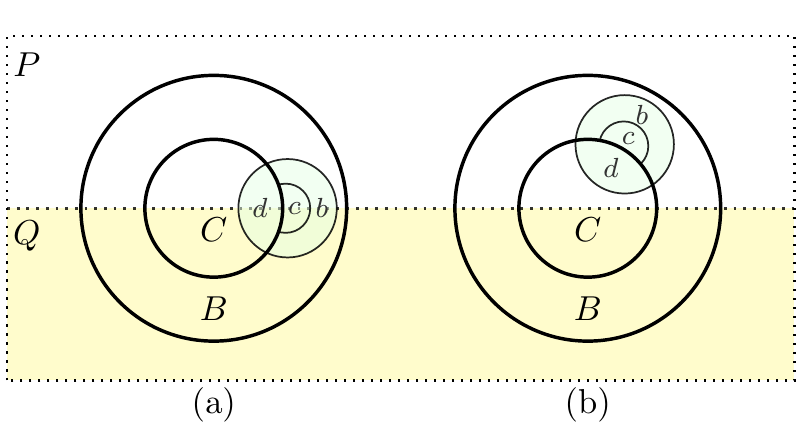}
	\caption{The extension of the domain wall version of {\bf A0}. In both figures, $d \subset C$ and $bc \subset B$. For (a), the green disk $bcd$ is a partition for the domain wall version of axiom {\bf A1}. For (b), the green disk $bcd$ is a partition for the bulk version of axiom {\bf A1}. }
	\label{fig_axiom_A0_enlarge}
\end{figure}

First, let us study how the domain wall version of axiom {\bf A0} can be extended. We want to show
\begin{equation}
    (S_{C} + S_{BC} -S_B)_{\sigma} = 0, \label{eq:temp111}
\end{equation}
for large $B$ and $C$ that is topologically equivalent to the red disk $BC$ in Fig.~\ref{fig:axioms_all}.

It suffices to consider two ways of deforming the subsystem. First, consider enlarging $B$ to $BB'\supset B$ while keeping $C$ fixed. We will see that 
\begin{equation}
    (S_{C} + S_{CBB'} - S_{BB'})_{\sigma} =0. \label{eq:enlarge_outer_B}
\end{equation}
To see why this condition holds, we note that
\begin{equation}
\begin{aligned}
   (S_{C} + S_{CBB'} - S_{BB'})_{\sigma} & \le (S_{C} + S_{CB} - S_{B})_{\sigma}\\
   &=0.
\end{aligned}
\end{equation}
The first line follows from SSA. In the second line, we applied the domain wall version of condition {\bf A0} on $BC$. On the other hand, SSA implies that $(S_{C} + S_{CBB'} - S_{BB'})_{\sigma} \ge 0 $. Thus, Eq.~(\ref{eq:enlarge_outer_B}) holds.

Secondly, consider deforming the boundary between $B$ and $C$ so that $C$ is enlarged while $BC$ as a whole remains unchanged; see Fig.~\ref{fig_axiom_A0_enlarge} for an illustration.  More precisely, we consider a disk in green color, divided into $b$, $c$, and $d$. Here, $c\subset B$, attached to the boundary between $B$ and $C$, is sufficiently small such that there are subsystems $d\subset C$ and $b \subset B$ that surrounds $c$. We have
\begin{equation}
\begin{aligned}
(S_{Cc} + S_{BC} - S_{B\setminus c})_{\sigma} & \le (S_{Cc} - S_{B\setminus c} + S_B - S_C)_{\sigma}\\
& \le (S_{dc} + S_{bc} - S_d - S_b)_{\sigma}\\
&=0.
\end{aligned}
\end{equation}
In the first line, we applied the domain wall version of the condition {\bf A0} for $BC$. In the second line, we used SSA. In the third line, we applied axiom {\bf A1} to the partition $bcd$ in Fig.~\ref{fig_axiom_A0_enlarge}(a), and applied the domain wall version of axiom {\bf A1} to the partition $bcd$ in Fig.~\ref{fig_axiom_A0_enlarge}(b).

This argument implies that we can expand $C$ into the bulk and also along the domain wall so long as we can choose an appropriate green disk to apply our axioms (versions of {\bf A1}).

Here are a few remarks. The green disk formed by the union of $b$, $c$ and $d$ has a radius at most $r$ so that our axioms can apply.  $B$ must be thick enough so that $b$ and $c$ together is a subset of $B$. Moreover, the boundary between $b$ and $d$ should not cross the domain wall, as we did not make any assumptions about the entanglement entropies of such subsystems. We emphasize that we do not attempt to exhaust all possible ways of enlarging $C$. Nevertheless, the deformations shown in Fig.~\ref{fig_axiom_A0_enlarge} provide at least one way of enlarging $C$ from a small-sized subsystem to an arbitrarily large one. Therefore, the deformations we have explained are enough to accomplish our proof.

Now, let us move onto the enlargement of the domain wall version of axiom {\bf A1}. Using SSA, one can again see that the axiom continues to hold when we expand $B$ or $D$ whilst fixing $C$. It remains to consider two cases: deforming the boundary between $B$ and $C$ while keeping both $BC$ and $D$ fixed and deforming the boundary between $C$ and $D$ while keeping both $CD$ and $B$ fixed. The underlying arguments are practically identical, so we just consider the first case and omit the argument for the second case. 

The idea is again to find green disks (of radius $r$) on which we can apply a version of {\bf A1}. As illustrated in Fig.~\ref{fig_axiom_A1_enlarge}, there are three (inequivalent) possibilities. By using SSA, we can show that it is possible to enlarge $C$ by adding the central area of the green disk while preserving the boundary version of entropy condition {\bf A1}. These types of deformations are all we need to deform a small $C$ to an arbitrarily large one, with a position of our choice.

One may wonder why we did not discuss the deformation of the boundary between $B$ and $D$. While we can include more partitions to achieve these, that is unnecessary. When we enlarge the subsystems, we can enlarge them in such a way that the boundary between $B$ and $D$ is the correct one. This is possible because, when we expand $B$ and $D$ while keeping $C$ fixed, the same axiom holds independent of how we expand them. This completes the proof of the extensions of axioms.

 \begin{figure}[h]
	\centering
\includegraphics[scale=1]{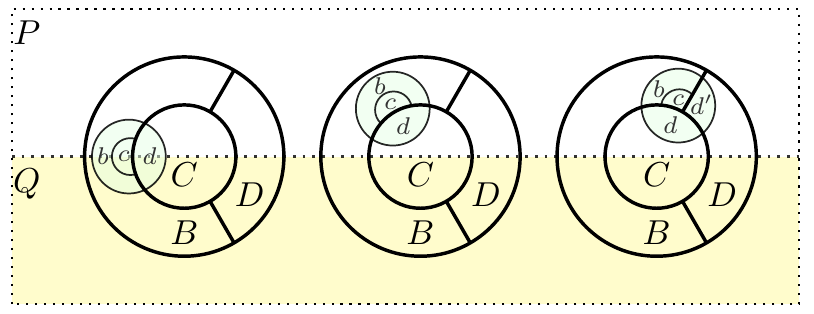}
	\caption{The extension of the domain wall version of axiom {\bf A1}. On each green disk, a version of axiom {\bf A1} is used to show that we can enlarge $C$ to include the central area ($c$) of the green disk.}
	\label{fig_axiom_A1_enlarge}
\end{figure}

\section{Isomorphism theorem}\label{appendix:isomorphism}

The proof of the isomorphism theorem (Theorem~\ref{thm:isomorphism}) is also very similar to its bulk analog in Ref.~\cite{SKK2019}. We briefly sketch the main idea, focusing on how the arguments of Ref.~\cite{SKK2019} can be applied to our setup. 

The main idea is to use versions of axiom {\bf A1} to smoothly deform the region by a sequence of ``small'' deformations called \emph{elementary step} of deformation. A finite sequence of elementary steps is a \emph{path} that connects a pair of regions that can be separated far apart. This way, we establish a well-defined notion of smooth deformation. The isomorphism theorem states that the information convex sets associated with a pair of regions are isomorphic if the pair can be connected by a path.

Let us consider the elementary steps illustrated in Fig.~\ref{fig:isomorphism_proof_sketch}. For illustration purpose, an annulus topology is shown. However, the argument applies more generally. Consider the information convex sets $\Sigma(AB)$ and $\Sigma(ABC)$ for the subsystem choice shown in either Fig.~\ref{fig:isomorphism_proof_sketch}(a) or Fig.~\ref{fig:isomorphism_proof_sketch}(b). Here, the deformation of region is $AB\leftrightharpoons ABC$.

 \begin{figure}[h]
	\centering
\includegraphics[scale=1]{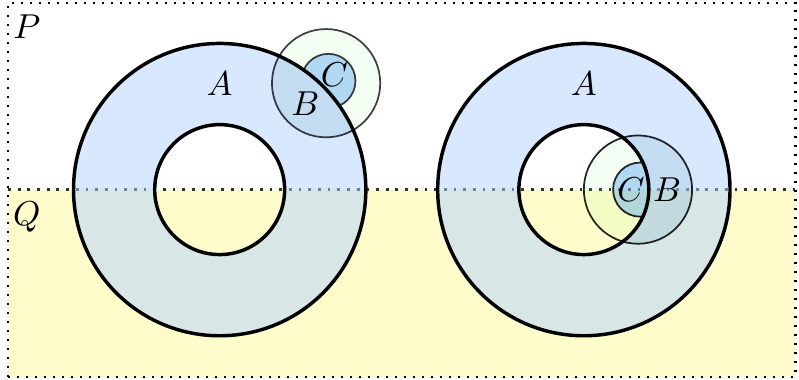}
	\caption{For the proof of the isomorphism theorem. For the partitions in this figure, $I(A:C\vert B)=0$ for any $\rho_{ABC}\in \Sigma(ABC)$. For (a), this fact follows from the bulk version of axiom {\bf A1}. For (b), this fact follows from the domain wall version of axiom {\bf A1}. }
	\label{fig:isomorphism_proof_sketch}
\end{figure}

Note that the usage of axiom {\bf A1} in this figure is similar to that in the proof of the extension of axioms (Appendix~\ref{appendix:extensions_of_axioms}). However, the proof of the isomorphism theorem turns out to be much more subtle and intricate. The subtlety is due to the fact that we consider the information convex set rather than the reference state. We will discuss this subtlety as we go.

We wish to show $\Sigma(AB)$ and $\Sigma(ABC)$ are isomorphic, denoted as $\Sigma(AB)\cong \Sigma(ABC)$, by which we mean the following.
\begin{enumerate}
	\item The partial trace $\Tr_C$ and the Petz map\footnote{
	The Petz map~\cite{Petz1987} is a quantum channel, which has an explicit expression $\calE_{B\to BC}^{\sigma}(X_{AB})= \sigma_{BC}^{\frac{1}{2}} \sigma_{B}^{-\frac{1}{2}} X_{AB} \sigma_{B}^{-\frac{1}{2}} \sigma_{BC}^{\frac{1}{2}}$ on the support of $\sigma_{BC}$.} $\calE_{B\to BC}^{\sigma}$ are maps between $\Sigma(AB)$ and $\Sigma(ABC)$:
	\begin{eqnarray}
	\Tr_C \rho_{ABC} \in \Sigma(AB),&&\quad \forall \rho_{ABC}\in \Sigma(ABC), \label{eq:iso_simple_1}\\
	\calE_{B\to BC}^{\sigma}(\rho_{AB}) \in \Sigma(ABC), &&\quad \forall \rho_{AB} \in \Sigma(AB).	\label{eq:iso_hard_1}
	\end{eqnarray}
	\item The following two operations are identity maps on the respective information convex set:
	\begin{eqnarray}
		 \calE_{B\to BC}^{\sigma}\circ \Tr_C: \,\,\,\,&&\Sigma(ABC)\to \Sigma(ABC), \label{eq:iso_simple_2}\\ 
	\Tr_C \circ \calE_{B\to BC}^{\sigma}:\,\,\,\, &&\Sigma(AB)\to \Sigma(AB). \label{eq:iso_hard_2}
	\end{eqnarray}
	This establishes the fact that $\Tr_D$ and $\calE_{B\to BC}^{\sigma}$ are bijections between $\Sigma(AB)$ and $\Sigma(ABC)$.
	\item The entropy difference and distance measures are preserved under the isomorphism.
\end{enumerate}
 Equation~(\ref{eq:iso_simple_1}) and (\ref{eq:iso_simple_2}) are simple to derive. Equation~(\ref{eq:iso_simple_1})  follows from the definition of the information convex set. Equation~(\ref{eq:iso_simple_2}) follows from the fact that $I(A:C\vert B)_{\rho}=0$ for any $\rho_{ABC}\in \Sigma(ABC)$. Under this condition, the Petz map ($\calE_{B\to BC}^{\sigma}$) is a quantum channel that recovers the state $\rho_{ABC}$ from $\rho_{AB}$.
 
  The preservation of entropy difference and distance measure is easy to establish once the isomorphism is established. The preservation of entropy difference follows from the conditional independence of $I(A:C\vert B)_{\rho}=0$ for any $\rho_{ABC}\in \Sigma(ABC)$ and that $\Tr_A \rho_{ABC}=\sigma_{BC}$ for the partitions in Fig.~\ref{fig:isomorphism_proof_sketch}. The preservation of distance measure is a consequence of the fact that distance measures are monotonic under the action of quantum channels; both $\Tr_D$ and $\calE_{B\to BC}^{\sigma}$ are quantum channels, and they reverse each other on the information convex sets.
  
 The subtle part of the proof is the justification of Eqs.~(\ref{eq:iso_hard_1}) and (\ref{eq:iso_hard_2}). If every $\rho_{AB}\in \Sigma(AB)$ has an ``extension'' in  $\Sigma(ABC)$, which has $\rho_{AB}$ as its reduced density matrix, then both Eqs.~(\ref{eq:iso_hard_1}) and (\ref{eq:iso_hard_2}) follow. However, the existence of this extension is part of what we need to prove. (Recall that we cannot use the isomorphism theorem at this point because we are trying to prove it.)

 The rest of this section is devoted to a sketch of the existence of this extension. How to show that for any $\rho_{AB}\in \Sigma(AB)$, there is a density matrix on a larger region $ABC$ that matches it? Furthermore, how do we show this extension $\rho_{ABC}$ belongs to $\Sigma(ABC)$?
 The key technique is the merging lemma (Lemma~\ref{lemma:merging_lemma}) and the merging theorem (Theorem~\ref{thm:merging_info_convex_set}). Specifically, we can prove that for every $\rho_{AB}\in \Sigma(AB)$ there exists an element  $\rho_{ABC}\in \Sigma(ABC)$ such that $\Tr_C \rho_{ABC}=\rho_{AB}$ so long as $AB$ is thick enough. The requirement that $AB$ is thick enough is for the purpose of avoiding potential pathological counterexamples.\footnote{ $AB$ should be thicker than $2r$. If $AB$ is thinner than that, there can be pathological counterexamples for which there is no room to achieve the deformation of regions required in the merging theorem.} 
 
 The merging lemma provides a way to generate a density matrix on the larger region $ABC$ that is consistent with $\rho_{AB}$; the merging theorem guarantees that the resulting density matrix is an element of $\Sigma(ABC)$. This is why the extension exists.
 This completes the sketch of the proof of the isomorphism theorem.

Finally, for completeness, we also provide a very brief sketch on why the merging theorem (Theorem~\ref{thm:merging_info_convex_set}) is true. The full proof is technical, and the interested reader is encouraged to read  Appendix C of Ref.~\cite{SKK2019} for the details.
The key idea behind the proof of the merging theorem is to introduce a convex set of density matrices, which we denote as $\hat{\Sigma}(\Omega)$. The definition of $\hat{\Sigma}(\Omega)$ does not make use of an extra layer as $\Sigma(\Omega)$ does. Instead, it requires some additional \emph{internal} conditional independence condition on its elements; these involve partitions near the boundary of $\Omega$. [These additional conditions mimic the conditional independence induced by the extra layer in the definition of $\Sigma(\Omega)$.]
A version of the merging theorem can be proved for $\hat{\Sigma}(\Omega)$. This merging theorem on $\hat{\Sigma}(\Omega)$ implies that every element in $\hat{\Sigma}(\Omega)$ can be consistently extended to a larger region containing $\Omega$, which subsequently implies that $\hat{\Sigma}(\Omega)= {\Sigma}(\Omega)$. Therefore, the merging theorem applies to the information convex set $\Sigma(\Omega)$ as well.

\section{Factorization of extreme points}\label{appendix:extreme_points_factorization}

In this appendix, we provide a streamlined proof of the factorization of extreme points of information convex sets. The main idea is to make use of (an enlarged version of) axiom {\bf A0}. It is amusing to contrast this usage of {\bf A0} with previous usage of {\bf A1} in the proof of two other important properties (Appendixes~\ref{appendix:extensions_of_axioms} and \ref{appendix:isomorphism}). Note that the condition {\bf A0} on the domain wall has no difference with that in the bulk. Therefore, this proof is essentially a recap of that in Ref.~\cite{SKK2019}. 

\begin{figure}[h]
	\centering
\includegraphics[scale=1]{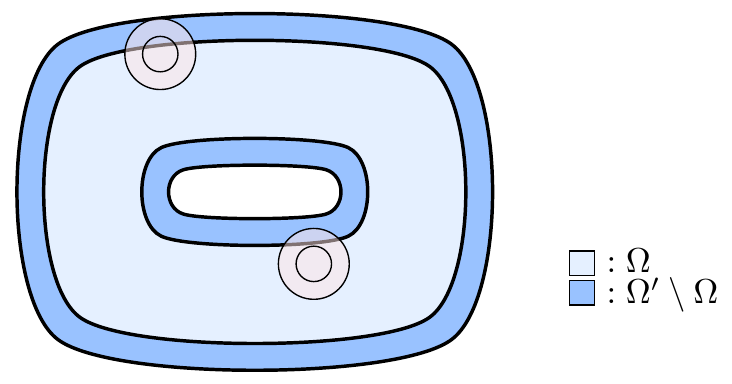}
	\caption{A thick enough subsystem $\Omega$ and a shell ($\Omega'\setminus\Omega$) around it. They form a region $\Omega'$. Note that $\Omega'$ can be smoothly deformed into $\Omega$ and $\Omega'\setminus \Omega$ is a thickened boundary of $\Omega'$. While we only depicted an annulus topology in this figure, the same factorization property applies to any sufficiently smooth subsystems. The red disks are regions on which an enlarged version of axiom {\bf A0} is considered; they will be called as $b'$ in the proof.}
	\label{fig:factorization_appendix}
\end{figure}

Let $\Omega$ be a thick enough but otherwise arbitrary subsystem. Let ${\Omega}'\supset \Omega$ be a subsystem that can be smoothly deformed into $\Omega$ for which $\Omega'\setminus\Omega$ is a thickened boundary of $\Omega'$. See Fig.~\ref{fig:factorization_appendix} for an illustration. For any extreme point $\rho_{{\Omega'}}^{\langle e\rangle}\in \Sigma({\Omega'})$, we show that
\begin{equation}
(S_{{\Omega'}} + S_{\Omega} - S_{{\Omega'} \setminus \Omega})_{\rho^{\langle e\rangle}}  =0. \label{eq:appendix_factorization}
\end{equation}
Because Eq.~\eqref{eq:appendix_factorization} is a straightforward consequence of the following equation:
\begin{equation}
\Tr_{{\Omega'}\setminus \Omega} \,|i\rangle_{{\Omega'}}\langle j| =\delta_{i,j} \rho_{\Omega}^{\langle e\rangle}, \label{eq:ltqo}
\end{equation}
where $\{|i\rangle_{{\Omega'}}\}$ is the set of eigenvectors of $\rho_{{\Omega'}}^{\langle e\rangle}$ with positive eigenvalues, we will focus on proving Eq.~\eqref{eq:ltqo}.

For this purpose, we first show that the states in the span of $\{|i\rangle_{{\Omega'}}\}$, reduced to $\Omega$, are in $\Sigma(\Omega)$. It suffices to show that
\begin{equation}
\Tr_{{\Omega'}\setminus b} |i\rangle_{{\Omega'}}\langle j| = \delta_{i,j} \,\sigma_b \label{eq:inclusion_information_convex}
\end{equation}
for any disk $b$ of radius $r$ that can be enlarged into $b'\subset \Omega'$. Here, $b'\setminus b$ is a thickened boundary of $b'$. (As an illustration, a red disk in Fig.~\ref{fig:factorization_appendix} is a $b'$, which contains a smaller disk $b$ in the middle.) By the extension of our axioms (Appendix~\ref{appendix:extensions_of_axioms}), we have
\begin{equation}
(S_b + S_{b'} - S_{b'\setminus b})_{\rho^{\langle e\rangle}} =0.
\end{equation}
This subsequently implies that the purification of $\rho^{\langle e\rangle}_{\Omega'}$ has a certain ``factorization property.'' Specifically, let 
\begin{equation}
|\varphi\rangle_{{\Omega'}P} = \sum_i \sqrt{p_i} \vert i\rangle_{\Omega'} \otimes \vert i\rangle_P 
\end{equation}
be the purification of $\rho^{\langle e\rangle}_{\Omega'}$ with a purifying space $P$, where $p_i >0$ for all $i$. By SSA, we can conclude that $(S_b + S_P - S_{bP})_{\vert \varphi\rangle}=0$. Because any bipartite state with a vanishing mutual information $I(A:B) := S_{A} + S_B - S_{AB}$ must be a factorized state, we can explicitly write down the following identity:
\begin{equation}
\sum_{i,j} \sqrt{p_i p_j} \left( \Tr_{{\Omega'} \setminus b}|i\rangle_{{\Omega'}}\langle j| \right) \otimes |i\rangle_P\langle j| = \sigma_b \otimes \sum_j p_j |j\rangle_P\langle j|, \label{eq:temp}
\end{equation}
where the summation is taken over $i$ and $j$ such that $p_i, p_j>0$. Eq.~\eqref{eq:inclusion_information_convex} follows straightforwardly from Eq.~\eqref{eq:temp}. Therefore, the reduced density matrix of any state in the span of $\{|i\rangle_{{\Omega'}} \}$ to $\Omega$ must belong to $\Sigma(\Omega)$.

Now, we are in a position to prove Eq.~\eqref{eq:ltqo}. We present a proof by contradiction. 
Suppose there is a state $|\phi\rangle_{{\Omega'}}$ in the span of $\{|i\rangle_{{\Omega'}}\}$ whose reduced density matrix is different from $\rho_{\Omega}^{\langle e\rangle}$. Note that $\rho_{{\Omega'}} = p|\phi\rangle_{{\Omega'}}\langle \phi| + (1-p) \rho'_{\Omega'}$ for some $p>0$ and $\rho'_{\Omega'}$ living in the state space of the Hilbert space spanned by $\{ |i\rangle_{{\Omega'}}\}$. This means that 
\begin{equation}
\rho_{\Omega}^{\langle e\rangle } = p \,\Tr_{{\Omega'} \setminus \Omega}\,|\phi\rangle_{{\Omega'}}\langle \phi| + (1-p) \Tr_{{\Omega'} \setminus \Omega}\,\rho'_{\Omega'},
\end{equation}
where both density matrices on the right-hand-side belong to $\Sigma(\Omega)$. However, this is a contradiction because the left-hand-side must be an extreme point by the isomorphism theorem; $\rho^{\langle e\rangle}_{\Omega}$ was obtained from an extreme point of another information convex set via an isomorphism. Therefore, both terms on the right-hand-side must be equal to $\rho^{\langle e\rangle}_{\Omega}$. We thus conclude that any state in the span of $\{
|i\rangle_{{\Omega'}}\}$, restricted to $\Omega$, must be equal to $\rho_{\Omega}^{\langle e \rangle}$. By inspecting the matrix elements, we conclude Eq.~\eqref{eq:ltqo}.

Finally, our main claim (Eq.~\eqref{eq:appendix_factorization}) follows because Eq.~\eqref{eq:ltqo} implies that any purification of $\rho^{\langle e\rangle}_{\Omega'}$ must have vanishing mutual information between the purifying space and $\Omega$. This completes the derivation of Eq.~\eqref{eq:appendix_factorization}.

\section{Fusion space}\label{appendix:fusion_space}
In this appendix, we provide a proof of Theorem~\ref{thm:Hilbert}. We shall refer to this result as the \emph{Hilbert space theorem}. Specifically, consider a sufficiently thick but otherwise arbitrary subsystem $\Omega$. We claimed that, once we fix the extreme point associated with $\partial \Omega$, i.e., the thickened boundary of $\Omega$, the remaining degrees of freedom is isomorphic to the state space of some finite-dimensional Hilbert space.

The proof of Theorem~\ref{thm:Hilbert} presented below is an improvement of that in Appendix E of Ref.~\cite{SKK2019}. While we will depict subsystems in the bulk for concreteness, the underlying logic applies more generally, for instance, to the subsystems intersecting with the domain wall. This is because every argument is based on the extensions of axioms, isomorphism theorem, and the factorization of extreme points, which we have generalized in Appendixes~\ref{appendix:extensions_of_axioms},~\ref{appendix:isomorphism}, and~\ref{appendix:extreme_points_factorization}.

We will use a well-known structure theorem of quantum Markov state \cite{2004CMaPh.246..359H}. The precise statement is presented below as a lemma.
\begin{lemma}[Structure of quantum Markov states~\cite{2004CMaPh.246..359H}] \label{lemma:HJPW}
	If $\rho_{ABC}$ satisfies $I(A:C\vert B)=0$, there exists a decomposition $\calH_B= \bigoplus_j \calH_{B^L_j} \otimes \calH_{B^R_j}$, such that
	\begin{equation}
	\rho_{ABC} = \bigoplus_j p_j \rho_{AB^L_j} \otimes \rho_{B^R_j C},
	\end{equation}
	where $\{ p_j\}$ is a probability distribution, $\rho_{AB^L_j} $ is a density matrix on $\calH_A \otimes \calH_{B^L_j}$ and $\rho_{B^R_j C}$ is a density matrix on $ \calH_{B^R_j} \otimes \calH_C$.
\end{lemma}
\begin{remark}
	There is no known generalization of this lemma for approximate quantum Markov states, states with small but nonzero $I(A:C\vert B)$; see Ref.~\cite{Ibinson2008} for a related discussion. Therefore, while we expect the conclusion of this paper to be extended to the setup in which the assumptions in Fig.~\ref{fig:axioms_all} hold approximately, the proofs in this appendix do not. Additional techniques need to be developed for approximate cases. 
\end{remark}

\begin{figure}[h]
	\centering
\includegraphics[scale=0.97]{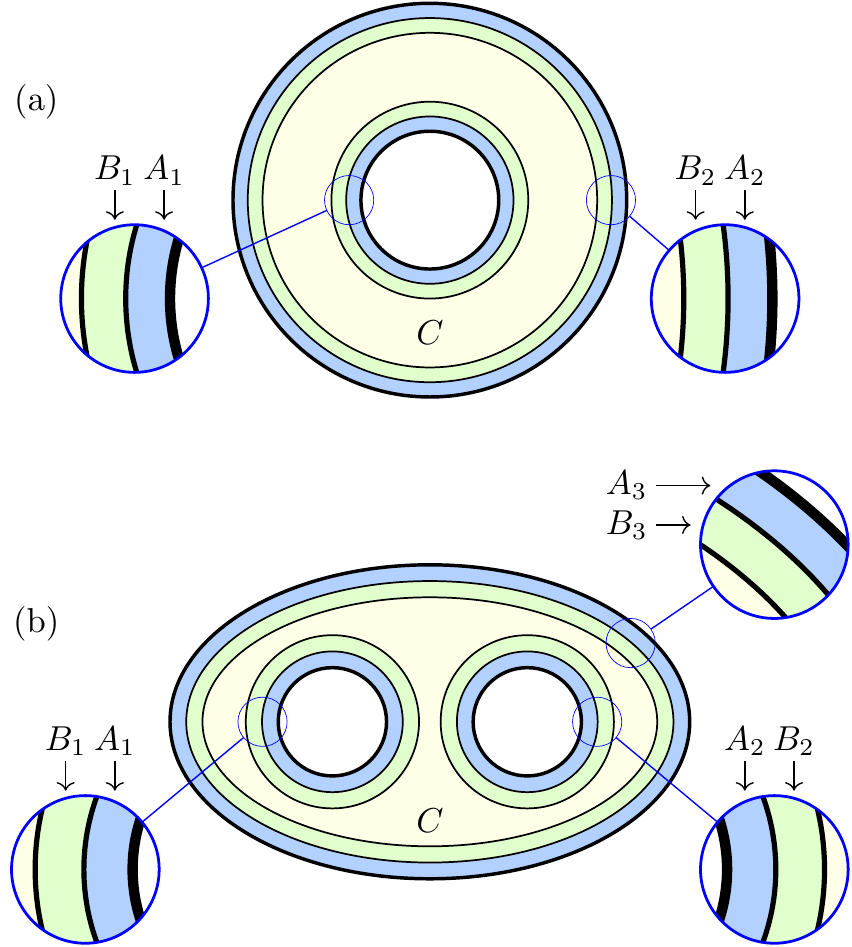}
	\caption{The partition $\Omega= (\cup_{i=1}^K A_i B_i) \cup C$ for: (a) an annulus, which has $K=2$, and (b) a 2-hole disk, which has $K=3$.}
	\label{fig:A_iB_iC_partition_of_Omega}
\end{figure}

The following is another useful result. 
\begin{lemma}\label{Lemma:Trace}
Consider a subsystem $\Omega' \supset \Omega$ that can be smoothly deformed into $\Omega$, where $\Omega'\setminus \Omega$ is the thickened boundary of $\Omega'$. Suppose $\rho_{\Omega'} \in \Sigma(\Omega')$ can be written as $\rho_{\Omega'} = \sum_i q_i \lambda^i_{\Omega'}$, where $\{q_i\}$ is a probability distribution with $q_i >0$, $\forall i$ and $\{ \lambda^i_{\Omega'}\}$ is a set of density matrices. Then
\begin{equation}
\Tr_{\Omega'\setminus \Omega} \lambda^i_{\Omega'} \in \Sigma(\Omega).
\end{equation}
\end{lemma}
The proof of this statement is nearly identical to that of Lemma D.1 in Ref.~\cite{SKK2019}. The only difference is that we also need the domain wall version of the extension of condition {\bf A0} discussed in Appendix~\ref{appendix:extensions_of_axioms}, which generalizes an analogous statement in Ref.~\cite{SKK2019}.

We will frequently consider the partition of $\Omega$  shown in Fig.~\ref{fig:A_iB_iC_partition_of_Omega}. 
Explicitly, we have $\Omega= (\cup_{i=1}^K A_i B_i) \cup C$, where $A_i$, $B_i$ and $A_i B_i$ are thickenings of the $i$-th boundary of $\Omega$ with different thicknesses; $A_i$ is the outer layer and $B_i$ is the inner layer. For example, as  illustrated in Fig.~\ref{fig:A_iB_iC_partition_of_Omega}, an annulus has $K=2$ and a 2-hole disk has $K=3$. Note that these types of partitions are very general, and they can be applied to regions intersecting with the gapped domain wall as well.

First, we observe a useful conditional independence property of this partition.
\begin{lemma}\label{prop:CMI_general}
	Let $\Omega$ be a subsystem with $K$ disjoint boundaries. Let $\Omega= (\cup_{i=1}^K A_i B_i) \cup C$, where $A_i$, $B_i$ and $A_i B_i$ are thickenings of the $i$-th boundary of $\Omega$ with different thickness; $A_i$ is the outer layer and $B_i$ is the inner layer. (See Fig.~\ref{fig:A_iB_iC_partition_of_Omega}.) We have
	\begin{equation}
	I(A_i : \Omega\setminus A_i B_i \vert B_i)_{\rho} =0,\quad \forall \, i \textrm{ and } \forall \rho_{\Omega} \in \Sigma(\Omega).
	\end{equation}
\end{lemma}

The proof of this proposition is the same as the bulk version in Ref.~\cite{SKK2019}; see Lemma~D.2 therein. The idea is that we can smoothly deform $\Omega\setminus A_i$ to $\Omega$. Using versions of axiom {\bf A1}, we can derive the claimed conditional independence relation.

The following proposition characterizes the universal structure of elements in $\Sigma_I(\Omega)$.
\begin{Proposition}\label{prop:docomposition_general}
		Consider $\rho^I_{\Omega}\in \Sigma_I(\Omega)$, (see Section~\ref{sec:fusion}). For the partition of $\Omega= (\cup_{i=1}^K A_i B_i) \cup C$ described above, there exists a decomposition
	\begin{equation}
	\calH_{B_i} = (\calH_{B_i^L} \otimes \calH_{B_i^R}) \oplus \mathcal{H}' \label{eq:B_i_decomposition_1}
	\end{equation}
	for some Hilbert space $\mathcal{H}'$ such that
	\begin{equation}
	\rho^I_{\Omega} = \left(\otimes_{i=1}^K \,\rho^I_{A_i B_i^L} \right) \otimes \rho_{(\cup_i B_i^R) \cup C}, \label{eq:decomposition_for_state_in_Sigma_I}
	\end{equation}
	where the density matrix $\rho^I_{A_iB_i^L}$, supported on $\calH_{A_i}\otimes \calH_{B_i^L}$, is independent of the specific choice of element in $\Sigma_I(\Omega)$ once $I$ is fixed. $\rho_{(\cup_i B_i^R) \cup C}$ is a density matrix supported on $(\otimes_{i=1}^K \calH_{ B_i^R} )\otimes \calH_C$.
\end{Proposition}
\begin{remark}
    The decomposition Eq.~\eqref{eq:B_i_decomposition_1} does not necessarily imply that $B^L_i$ is a subsystem of $B$. In general it is not. Therefore, Eq.~\eqref{eq:decomposition_for_state_in_Sigma_I} does \emph{not} imply that the state $\rho^I_{\Omega}$ is a tensor product over a subsystem of $A_iB_i$ and its remainder within $\Omega$.
\end{remark}
\begin{proof}	
	First, because $\rho^I_{\Omega} \in \Sigma(\Omega)$, according to Lemma~\ref{prop:CMI_general}, we have $I(A_1 : \Omega\setminus A_1 B_1 \vert B_1)_{\rho^I} =0$. Then, by Lemma~\ref{lemma:HJPW}, there exists a decomposition $\calH_{B_1}= \bigoplus_j \calH_{B_{1j}^L} \otimes \calH_{B_{1j}^R}$ such that 
	\begin{equation}
	\rho^I_{\Omega} = \sum_j p_j \rho^I_{A_1B^L_{1j}} \otimes \rho_{B_{1j}^R C (\cup_{i\ne 1} A_i B_i)}, \label{eq:general_composition}
	\end{equation}
	where $\{ p_j\}$ is a probability distribution.
	Furthermore, because $\rho^I_{\Omega} \in \Sigma_I(\Omega)$, it carries a fixed sector label $I \in \calC_{\partial\Omega}$. Therefore, its reduced density matrix on sectorizable subsystem $B_1$ must be an extreme point. In fact, a stronger condition holds. According to Lemma~\ref{Lemma:Trace}, $\rho^I_{A_1B^L_{1j}} \otimes \rho_{B_{1j}^R C(\cup_{i\ne 1} A_i B_i)}$, for any $j$ with $p_j>0$, reduces to the same extreme point of $\Sigma(B_1)$. However, the Hilbert spaces $\calH_{B^L_{1j}}\otimes \calH_{B^R_{1j}}$ for different chocies of $j$ are orthogonal subspaces. The only consistent choice is that $ p_j$ is nonzero for only one choice of $j$. Therefore, Eq.~\eqref{eq:general_composition} can be simplified into
	\begin{equation}
	\rho^I_{\Omega} = \rho^I_{A_1 B_1^L} \otimes \rho_{B_1^R C (\cup_{i\ne 1} A_i B_i)}. \label{eq:simplified_composition}
	\end{equation}
	We can repeat the same logic for any $i$. The end result of this analysis is Eq.~\eqref{eq:decomposition_for_state_in_Sigma_I}. 
	
	Finally, we explain the fact that for any $\rho^I_{\Omega} \in \Sigma_I(\Omega)$, the decomposition of $\calH_{B_i}$ in Eq.~\eqref{eq:B_i_decomposition_1} and the density matrices $\{ \rho^I_{A_i B_i^L}\}$ in Eq.~\eqref{eq:decomposition_for_state_in_Sigma_I} can be chosen to be the same. This follows from the fact that different elements of $\Sigma_I(\Omega)$, for a fixed $I$, can be converted into each other by a quantum channel on $\Omega\setminus (A_iB_i)$, for any $i$.
	Specifically, without loss of generality, consider two density matrices $\rho_{\Omega}^{I,1}$ and $\rho_{\Omega}^{I,2}$. We can consider an additional layer $D_i \subset \Omega \setminus (A_iB_i)$ that surrounds $B_i$. We have $I(A_iB_i : \Omega\setminus (A_iB_iD_i) \vert D_i)=0$. Therefore, one can map $\rho_{\Omega}^{I,1}$ to $\rho_{\Omega}^{I,2}$ and vice versa by taking a partial trace on $\Omega \setminus (A_iB_iD_i)$ and then applying the Petz map from $D_i$ to $\Omega \setminus (A_iB_i)$. This two-step process only involved quantum channels acting on $\Omega\setminus (A_iB_i)$, thus completing the proof.
\end{proof}

Below, we provide a proof of the \emph{Hilbert space theorem}. This proof is in many sense simpler than the original one in Ref.~\cite{SKK2019}. Furthermore, it manifests the fact that the fusion space is physically accessible on a deformable region $(\cup_i B_i) \cup C$ within $\Omega$. 
\begin{theorem}[Hilbert space theorem]\label{thm:Hilbert}
	\begin{equation}
	\Sigma_I(\Omega) \cong \mathcal{S}(\mathbb{V}_I), \label{eq:Hilbert_1}
	\end{equation}
	where $\mathcal{S}(\mathbb{V}_I)$ is the state space of a finite dimensional Hilbert space $\mathbb{V}_I$. 
	Moreover, under the partition $\Omega= (\cup_{i=1}^K A_i B_i) \cup C$ described above, an arbitrary extreme point of $\Sigma_I(\Omega)$ has the following explicit expression 
	\begin{equation}
	\rho^{I \langle e\rangle}_{\Omega} = \left(\otimes_{i=1}^K \,\rho^I_{A_i B_i^L} \right) \otimes \vert \varphi \rangle \langle \varphi \vert, \label{eq:Hilbert_2}
	\end{equation}
	where the set of possible states $\{ \vert \varphi \rangle \}$ is the set of normalized pure states of a $\dim \mathbb{V}_I$ dimensional subspace of $(\otimes_{i=1}^K\calH_{ B_i^R}) \otimes \calH_C$.
\end{theorem}

\begin{proof}
	Let us enlarge $\Omega$ into $\Omega'$ by letting $\Omega' = (\cup_i A'_i B_i)\cup C$, where $A'_i \supset A_i$ and $A'_i\setminus A_i$ is a thickened $i$-th connected piece of the boundary of $\Omega'$.  According to Proposition~\ref{prop:docomposition_general}, any element of $\Sigma_I(\Omega)$ can be written as
	\begin{equation}
	\rho^{I}_{\Omega} = \left(\otimes_{i=1}^K \,\rho^I_{A_i B_i^L} \right) \otimes \rho_{(\cup_i B_i^R) \cup C}, \label{eq:ext_decomposition_1}
	\end{equation}
	for some density matrix $\rho_{(\cup_i B_i^R) \cup C}$, where the set of density matrices $\{ \rho^I_{A_i B_i^L} \}$ are fixed by the choice of $I$.
	This implies that $\Sigma_I(\Omega) \cong \{ \rho_{(\cup_i B_i^R) \cup C} \}$, where the isomorphism ``$\cong$'' preserves the entropy difference and any distance measure. 
	
	Below, we determine the set of density matrices $\{ \rho_{(\cup_i B_i^R) \cup C} \}$. We will show that $\{ \rho_{(\cup_i B_i^R) \cup C} \}$ forms the state space of a finite dimensional subspace of $(\otimes_{i=1}^K\calH_{ B_i^R}) \otimes \calH_C$.
	By the isomorphism theorem, we can obtain an element of $\Sigma_I(\Omega')$ by an extension of $\rho^I_{\Omega}$. This element can be written as
	\begin{equation}
	\rho^{I }_{\Omega'} = \left(\otimes_{i=1}^K \,\rho^I_{A'_i B_i^L} \right) \otimes \rho_{(\cup_i B_i^R) \cup C}, \label{eq:ext_decomposition_enlarged}
	\end{equation}
	where $ \Tr_{A'_i \setminus A_i} \rho_{A'_i B_i^L} = \rho_{A_i B_i^L}$. Equation.~\eqref{eq:ext_decomposition_enlarged} holds because the extension from $\Omega$ to $\Omega'$ can be done by applying a sequence of nonoverlapping quantum channels on each $A_i$.
	
	Let $\vert \varphi \rangle$ be a (normalized) state in the span of the eigenstates of $\rho_{(\cup_i B_i^R) \cup C}$ with positive eigenvalues. Then it follows that \begin{equation}
	\rho_{(\cup_i B_i^R) \cup C} = p \vert \varphi \rangle \langle \varphi\vert +(1-p) \widetilde{\rho}_{(\cup_i B_i^R) \cup C}
	\end{equation}
	for some $p\in (0,1)$ and density matrix $\widetilde{\rho}_{(\cup_i B_i^R) \cup C}$. Therefore,
	\begin{equation}
	\rho^{I}_{\Omega'} = p \rho^{I;\varphi}_{\Omega'} + (1-p) \widetilde{\rho}^{I}_{\Omega'}
	\end{equation}
	where 
	\begin{eqnarray}
	\rho^{I;\varphi}_{\Omega'} &=&  \left(\otimes_{i=1}^K \,\rho^I_{A'_i B_i^L} \right) \otimes \vert \varphi \rangle \langle \varphi \vert,\\
	\widetilde{\rho}^{I}_{\Omega'} &=&  \left(\otimes_{i=1}^K \,\rho^I_{A'_i B_i^L} \right) \otimes \widetilde{\rho}_{(\cup_i B_i^R) \cup C},
	\end{eqnarray}
	Because of Lemma~\ref{Lemma:Trace}, $\Tr_{\Omega'\setminus \Omega} \, \rho^{I;\varphi}_{\Omega'}$ must belong to $\Sigma_I(\Omega)$. Therefore,
	\begin{equation}
	\left(\otimes_{i=1}^K \,\rho^I_{A_i B_i^L} \right) \otimes \vert \varphi \rangle \langle \varphi \vert \in \Sigma_I(\Omega). \label{eq:extreme_factorized_varphi}
	\end{equation}
	Moreover, the state on the left-hand side of Eq.~\eqref{eq:extreme_factorized_varphi} must be an extreme point of $\Sigma_I(\Omega)$. This is because every element of $\Sigma_I(\Omega)$ is of the form Eq.~\eqref{eq:ext_decomposition_1}. Therefore,  $\{ \rho_{(\cup_i B_i^R) \cup C} \}$ forms the state space of a finite dimensional subspace of $(\otimes_{i=1}^K\calH_{ B_i^R}) \otimes \calH_C$.

	In particular, $\{ \rho_{(\cup_i B_i^R) \cup C} \}\cong \mathcal{S}(\mathbb{V}_I)$ for some finite dimensional Hilbert space $\mathbb{V}_I$. This justifies Eq.~\eqref{eq:Hilbert_1}. Equation~\eqref{eq:Hilbert_2} holds because every extreme point of $\Sigma_I(\Omega)$ is of the form shown on the left-hand side of Eq.~\eqref{eq:extreme_factorized_varphi} for some $\vert \varphi\rangle$.
	This completes the proof.
\end{proof}

\section{Aspects of quasi-fusion}\label{appendix:quasi-fusion}
In this section, we initiate a yet-to-be-completed theory of quasi-fusion. We begin by showing that the notion of anti-sector is well-defined.

\begin{figure}[h]
	\centering
\includegraphics[scale=1]{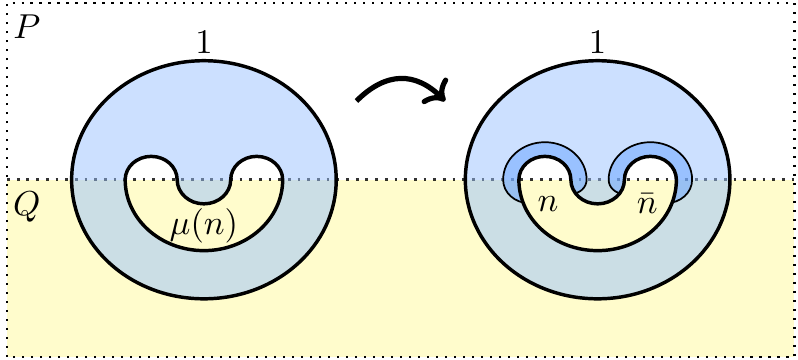}
	\caption{The definition of anti-sector $\bar{n}$.}

	\label{fig:n_bar}
\end{figure}

We define the anti-sector map ($n\to \bar{n}$) as the automorphism of $\calC_N$ illustrated in Fig.~\ref{fig:n_bar}. Here, the two disconnected boundaries of the subsystem carry sector labels $1\in \calC_O$ and $\mu(n)\in\calC_{\mathbb{U}}$ respectively. Then we look at the two $N$-shaped subsystems in darker blue. On the left side, we have $n\in \calC_N$. The unique sector on the right side, which we denote as $\bar{n}\in \calC_N$, is \emph{defined as} the anti-sector of $n$.

To see why the anti-sector map is an automorphism of $\calC_N$, we show the map is a bijection. This fact follows from two observations. First, the map $n \to \mu(n)$ is a bijection. This follows from the definition $\mu= \varphi \circ \eta_N$, where $\eta_N: \calC_N \to \calC_{ \mathbb{N} }$ is an embedding and $\varphi: \calC_{\mathbb{N}} \to \calC_{ \mathbb{U}}$ is an isomorphism; see Eqs.~\eqref{eq:embedding_N2NN} and \eqref{def:varphi}. Second, the map $\bar{n} \to \mu(n)$ is a bijection. The details are similar to the analysis shown above; we simply need to consider the ``mirror image'' of the isomorphism $\varphi$. 
Thus, the anti-sector map $n\to \bar{n}$ is an automorphism of $\calC_N$.

\begin{figure}[h]
	\centering
\includegraphics[scale=1]{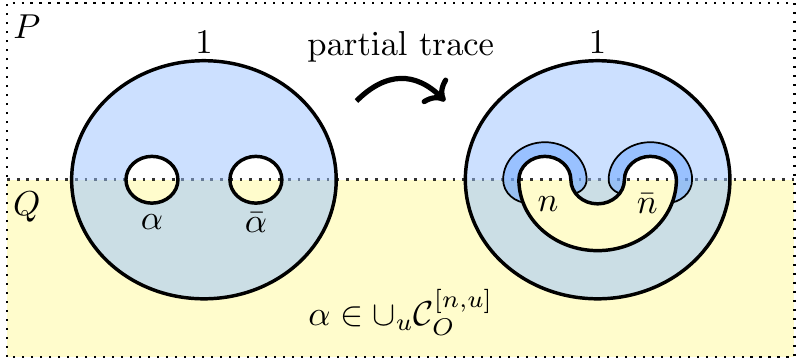}
	\caption{An alternative definition of $\bar{n}$.}
	\label{fig:n_bar_alternative}
\end{figure}

Furthermore, from Fig.~\ref{fig:n_bar_alternative}, it is easy to see that $\alpha \in \calC_O^{[n,u]}$ if and only if $\bar{\alpha} \in \calC_O^{[\bar{n}, \bar{u}]}$.  This is because, when $\alpha \in \calC_O^{[n,u]}$, after taking a partial trace, we obtain the same density matrix shown in Fig.~\ref{fig:n_bar}. This fact can serve as an alternative definition of the anti-sector. From this alternative (equivalent) definition, we have
\begin{equation}
\begin{aligned}
   \bar{\bar{n}} &= n,\\
   d_n &= d_{\bar{n}}.
\end{aligned}
\end{equation}
The first line follows from the fact that $\bar{\bar{\alpha}}=\alpha$. The second line follows from $d_{\alpha}=d_{\bar{\alpha}}$ and Eq.~\eqref{eq:qd_parton_composite}.

\subsection{Quasi-fusion rule} \label{ap:Subsystem_M_related}
In this section, we provide some details of the quasi-fusion rule of the $N$-type parton sectors. In particular, we present the proof of Proposition~\ref{Prop:Sigma_nn'^n''}. We restate the content below for the readers' convenience. These results concern the information convex set of a $M$-shaped subsystem, denoted as $M$. This $M$-shaped subsystem (see Fig.~\ref{fig:M-shape_and_quasi_fusion}) contains three $N$-shaped subsystems. For this reason, we can consider the following convex subsets of $\Sigma(M)$:
\begin{equation}
    \Sigma_{nn'}^{n''}(M),
\end{equation}
where $n, n', n'' \in \calC_N$ are the sector labels for the three $N$-shaped subsystems. These sets satisfy the following statements:
	\begin{enumerate}
		\item Every extreme point of $\Sigma(M)$ is contained in some $\Sigma_{nn'}^{n''}(M)$.
		\item $\cup_{n''} \Sigma_{nn'}^{n''}(M)$ is nonempty for $\forall \, n,n'\in\calC_N$.
		\item $\Sigma_{n1}^{n''}(M)$ is the empty set for $n''\ne n$. For $n''=n$, it has a unique element.
		\item $\Sigma_{1n'}^{n''}(M)$ is the empty set for $n''\ne n'$. For $n''=n'$, it has a unique element.
		\item $\Sigma_{nn'}^{1}(M)$ is the empty set for $n'\ne \bar{n}$. For $n'=\bar{n}$, it has a unique element.
	\end{enumerate}

\begin{figure}[h]
	\centering
\includegraphics[scale=0.97]{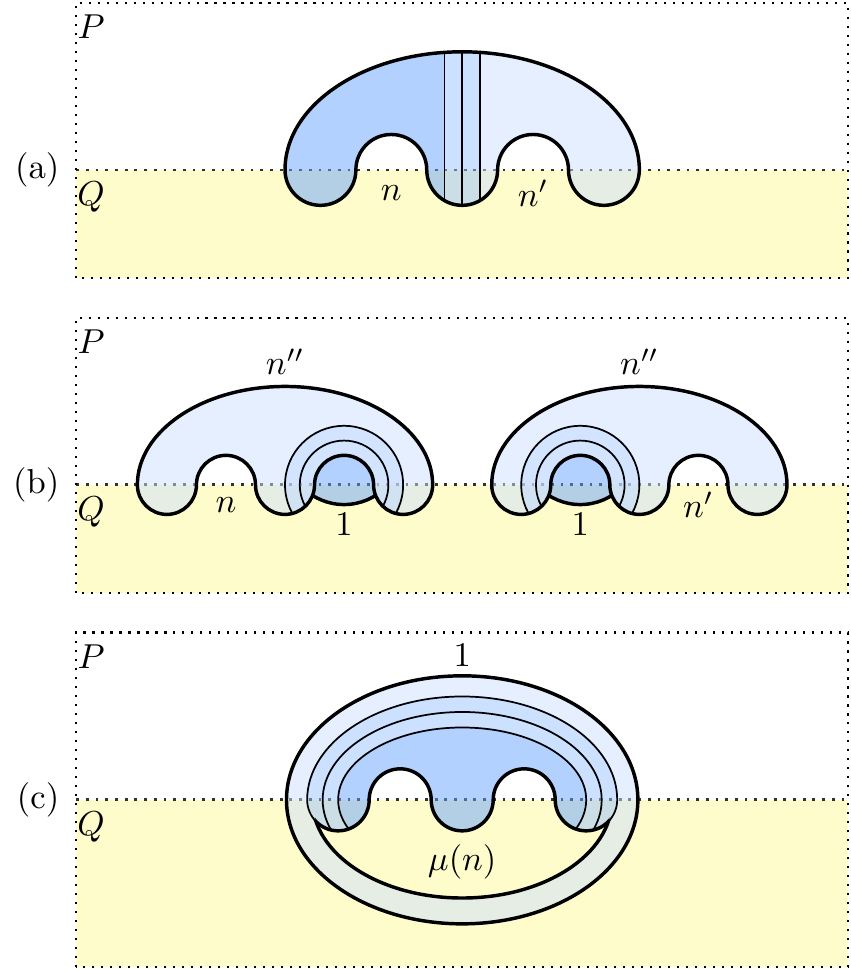}
	\caption{(a) The merging of two $N$-shaped subsystems. (b) Patching a slot for the sector choice $n'=1\in \calC_N$ (left) and $n=1\in \calC_N$ (right). (c) The merging of $M$ with an annulus for $n''=1\in\calC_N$, where the annulus is in the vacuum sector $1\in \calC_O$. }
	\label{fig:M_proofs}
\end{figure}

Let us prove these statements one by one.
For the proof of the first statement, let $\rho_M^{\langle e\rangle }$ be an extreme point of $\Sigma(M)$. Let $\partial M$ be the thickened boundary of $M$. According to the general discussion of the fusion space at the beginning of Section~\ref{sec:fusion}, $\partial M$ is a sectorizable subsystem; moreover, an extreme point $ \rho_M^{\langle e\rangle }$ must carry a fixed sector $I\in \calC_{\partial M}$. Furthermore, we see from the shape of $\partial M$ that the sector $I$ is a composite sector (see Section~\ref{sec:composite_sectors}). Because the three $N$-shaped subsystems in Fig.~\ref{fig:M-shape_and_quasi_fusion} can be identified as subsystems of $\partial M$, we always find a set of parton sectors $n,n', n''\in \calC_N$ on them. This implies that each extreme point of $\Sigma(M)$ must carry a definite set of sectors $n, n', n'' \in \calC_N$. This establishes the first statement.
	
The second statement follows from the merging process described in Fig.~\ref{fig:M_proofs}(a). For any choice of $n,n'\in \calC_N$, the merged state exists, and it belongs to $\conv\big(\cup_{n''} \Sigma_{nn'}^{n''}(M)\big)$. Thus $\cup_{n''} \Sigma_{nn'}^{n''}(M)$ is nonempty.

The third and the fourth statements follow from the same idea. The relevant merging processes are illustrated in Fig.~\ref{fig:M_proofs}(b). Due to the similarity, we only provide the proof for the third statement. Suppose $\Sigma_{n1}^{n''}(M)$ contains at least one element. Then we can apply the merging process in the left figure of Fig.~\ref{fig:M_proofs}(b). Let the resulting subsystem be $N$. Then the original density matrix must be the reduced density matrix of the extreme point $\rho_N^n$ of $\Sigma(N)$. This implies that $\Sigma_{n1}^{n''}(M)$ is an empty set for $n''\ne n$. For the same reason, $\Sigma_{n1}^{n}(M)$ must have a unique element for $\forall n\in \calC_N$. 
	
For the fifth statement, suppose $\Sigma_{nn'}^1(M)$ is nonempty. For an element of $\Sigma_{nn'}^1(M)$, we can apply the merging process shown in Fig.~\ref{fig:M_proofs}(c). Here, the annulus carries $1\in \calC_O$. This merging process is possible because $n''=1$. The merged state must carry the sector $\mu(n)\in \calC_{\mathbb{U}}$ in the newly formed boundary. Therefore, the merged state is unique; it is the state depicted in Fig.~\ref{fig:n_bar}. Thus, $n'=\bar{n}$ is a necessary condition for $\Sigma_{n n'}^1(M)$ to be nonempty.
Furthermore, there is a unique element in $\Sigma_{n \bar{n}}^1(M)$. The existence follows from the fact that for each $n$, a unique state depicted in Fig.~\ref{fig:n_bar}  exists, which is labeled by $1\in \calC_O$ and $\mu(n)\in \calC_{ \mathbb{U}}$. Every element in $\Sigma_{n \bar{n}}^1(M)$ must be the reduced density matrix of this unique density matrix, and therefore the choice is unique.
This completes the proof.

\subsection{When quasi-fusion becomes a fusion}\label{appendix:gapped boundary}

Sometimes, the quasi-fusion rule reduces to the ordinary rule of fusion. This happens when one side of the domain wall, say $Q$, has a trivial anyon content. In this section, we explain this reasoning.

Let us emphasize that our argument applies even if $Q$ is \emph{not} adiabatically connected to the (trivial) product state. There is one such nontrivial example, namely the ${E}_8$ state~\cite{Kitaev2006solo}. While such phases support a chiral edge mode, one may be able to gap out this mode by placing a topological phase on the $P$ side that matches the chiral central charge and turning on some perturbations along the domain wall. If this is possible, our argument would still apply.

Here are the key results.
\begin{enumerate}
	\item $\calC_U =\{ 1\}$.
	\item $\calC_N $ and $ \calC_O$ are isomorphic. Furthermore, under the isomorphism $\mathbf{C}: \calC_N \to \calC_O$, we have $d_n = d_{\mathbf{C}(n)}$.
	\item The quasi-fusion rule of parton sectors in Fig.~\ref{fig:M-shape_and_quasi_fusion} coincides with the conventional fusion rule. Namely, when we specify $n,n',n''\in \calC_N$, there is a unique fusion space, which can be labeled as $\mathbb{V}_{nn'}^{n''}$ and the associated fusion multiplicity $N_{nn'}^{n''}$ satisfies 
	\begin{equation}
	N_{nn'}^{n''} = N_{\mathbf{C}(n)\mathbf{C}(n')}^{\mathbf{C}(n'')}.
	\end{equation}
\end{enumerate}

\begin{figure}[h]
	\centering
\includegraphics[scale=1]{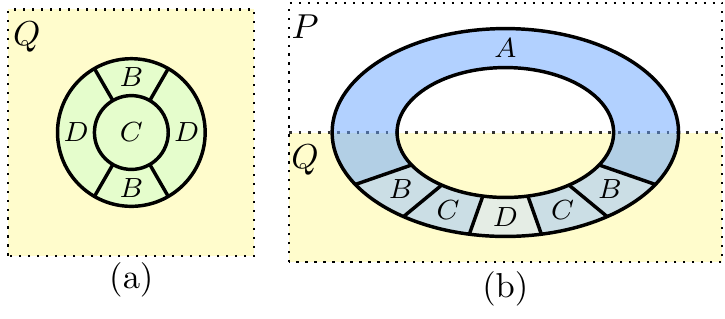}
	\caption{(a) If $Q$ has a trivial anyon content, the entropy combination $(S_{BC} + S_{CD} - S_B - S_D)_{\sigma}$ vanishes for the partition shown in this diagram. (b) The merging process in this figure is a ``connection process'', which turns an $N$-shaped subsystem into an $O$-shaped subsystem.}\label{fig:C_D_trivial_Q}
\end{figure}

The proofs of all three statements are similar.
The idea is to strengthen the isomorphism theorem when $Q$ has a trivial anyon content. We shall derive the first fact in detail. The rest follows straightforwardly. (The key idea has been illustrated in Fig.~\ref{fig:boundary_connect_M-shape}.)

Under the assumption that $Q$ has trivial anyon content, we have an additional identity. Namely, for the set of subsystems of $Q$ described in Fig.\ref{fig:C_D_trivial_Q}(a), we have
\begin{equation}
(S_{BC} + S_{CD} - S_B - S_D)_{\sigma} = 0. \label{eq:extra_condition}
\end{equation}
This additional identity implies that one can establish an isomorphism theorem between two subsystems with different topologies. Specifically, we can imagine a topology-changing connection/disconnection of a subsystem on the $Q$ side; see Figs.~\ref{fig:boundary_connect_M-shape} and \ref{fig:C_D_trivial_Q}(b). This connection/disconnection preserves the structure of the information convex sets.

Specifically, let $N=ABC$  ($O=ABCD$) be the $N$-shaped ($O$-shaped) subsystem shown in Fig.~\ref{fig:C_D_trivial_Q}(b). A connection process is a pair of operations acting on the region $N$ and its information convex set $\Sigma(N)$. The connection process turns $N$ into $O$. This connection process is associated with a map $\mathbf{C}: \Sigma(N)\to \Sigma(O)$, which is defined by the merging process in Fig.~\ref{fig:C_D_trivial_Q}. Conversely, the disconnection process $\mathbf{D}$ turns $O$ into $N$. The action of this map in $\Sigma(O)$ is simple; simply take a partial trace on $O\setminus N$.

While both of these processes can be applied to any choice of $P$ and $Q$, $\mathbf{C} \circ \mathbf{D} : \Sigma(O)\to \Sigma(O)$ is irreversible in general. However, when $Q$ has a trivial anyon content, due to the extra condition Eq.~(\ref{eq:extra_condition}), an arbitrary element $\rho_O\in \Sigma(O)$ satisfies
\begin{equation}
I(A:CD\vert B)_{\rho} = I(AB:D \vert C)_{\rho}=0.
\end{equation}In this case, $\mathbf{C} \circ \mathbf{D}$ is the identity operation on $\Sigma(O)$. This implies that $\Sigma(N)$ and $\Sigma(O)$ are isomorphic, with isomorphisms given by $\mathbf{C}$ and $\mathbf{D}$.

\section{Proof of maximal entropy}\label{appendix:alpha_dual}

In this appendix, we prove Eq.~\eqref{eq:merge_temp}, which implies that the reduced density matrices of $\rho_{W}^{(\alpha \overrightarrow{\alpha}; \beta \overrightarrow{\beta})}$ on $G_L$ and $G_R$ are certain maximum-entropy states. Because the derivation of the two identities are similar, we will only present the derivation of
\begin{equation}
\Tr_{W\setminus G_L} \, \rho_W^{(\alpha \overrightarrow{\alpha} ;\beta \overrightarrow{\beta})}= \rho_{G_L}^{\alpha \merge \beta} \label{eq:_rho_independence_G_L}
\end{equation}
in details. The subsystems and the merging process relevant to this proof are illustrated in Fig.~\ref{fig:4_hole_combined_appendix}.
\begin{figure}[h]
	\centering
\includegraphics[width=0.85\columnwidth]{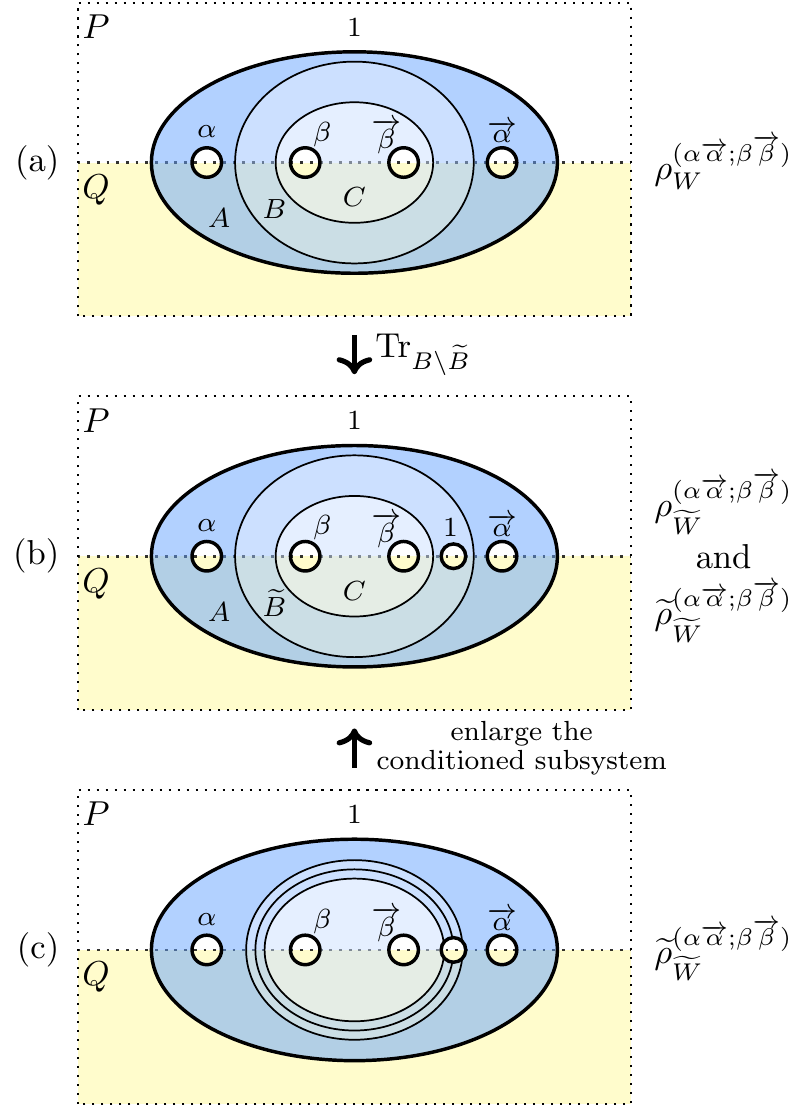}
	\caption{(a) The partition $W=ABC$.
	(b) Tracing out a hole from $B$ and get $\widetilde{B}$. We denote $A\widetilde{B}C$ as $\widetilde{W}$. For the density matrices we consider, the fifth hole is in the vacuum sector. (c) The merging process that defines $\widetilde{\rho}_{\widetilde{W}}^{(\alpha\protect\overrightarrow{\alpha};\beta\protect\overrightarrow{\beta})}$.}
	\label{fig:4_hole_combined_appendix}
\end{figure}

\begin{proof}
The main strategy is to construct another merged state $\widetilde{\rho}^{(\alpha\overrightarrow{\alpha}; \beta\overrightarrow{\beta})}$. We show that this is identical to ${\rho}^{(\alpha\overrightarrow{\alpha}; \beta\overrightarrow{\beta})}$ on a subsystem containing $G_L$. We further show that the reduced density matrix of $\widetilde{\rho}^{(\alpha\overrightarrow{\alpha}; \beta\overrightarrow{\beta})}$ on $G_L$ is $\rho^{\alpha\merge \beta}_{G_L}$. Below are the details.

First, if we divide the 4-hole disk $W$ into $W=ABC$ as shown in Fig.~\ref{fig:4_hole_combined_appendix}(a), we obtain a conditional independence condition
\begin{equation}
I(A:C\vert B)_{\rho^{(\alpha\overrightarrow{\alpha}; \beta\overrightarrow{\beta})}}=0, \label{eq:4-hole_ABC}
\end{equation}
where the annulus $B$ can be either the overlapping region of the merging process depicted in Fig.~\ref{fig:4_hole_combined}(b$\to$c) or a region enlarged from it. If $B$ is the overlapping region, the conditional independence follows from that of the merged state; if $B$ is enlarged from the overlapping region, we use the following consequence of SSA, $I(AA':CC' \vert B)\le I(A:C\vert A'C'B)$, to establish the conditional independence relation.

Second, we cut a hole from $B$ and reduce it to $\widetilde{B}$; see Figs.~\ref{fig:4_hole_combined_appendix}(a) and \ref{fig:4_hole_combined_appendix}(b). We shall refer to this hole as the fifth hole from now on. Also, we shall denote $A\widetilde{B} C$ as $\widetilde{W}$. Obviously, for the state ${\rho}^{(\alpha\overrightarrow{\alpha}; \beta\overrightarrow{\beta})}$, the fifth hole is in the vacuum sector. One can show
\begin{equation}
I(A:C\vert \widetilde{B})_{\rho^{(\alpha\overrightarrow{\alpha}; \beta\overrightarrow{\beta})}}=0 \label{eq:4-hole_ABC_reduced}
\end{equation}
using the following argument. Note that the state ${\rho}_W^{(\alpha\overrightarrow{\alpha}; \beta\overrightarrow{\beta})}$ satisfies an \emph{extended} domain wall version of condition \textbf{A0} on the disk covering the fifth hole and an extra layer surrounding that hole, leading to
\begin{equation}
\begin{aligned}
(S_{B\setminus\widetilde{B}})_{\rho^{(\alpha\overrightarrow{\alpha}; \beta\overrightarrow{\beta})}} &= (S_{\widetilde{B}}-S_{B})_{\rho^{(\alpha\overrightarrow{\alpha}; \beta\overrightarrow{\beta})}}\\
&=(S_{A\widetilde{B}}-S_{AB})_{\rho^{(\alpha\overrightarrow{\alpha}; \beta\overrightarrow{\beta})}} \\
    &= (S_{\widetilde{B}C}-S_{BC})_{\rho^{(\alpha\overrightarrow{\alpha}; \beta\overrightarrow{\beta})}} \\
    &= (S_{A\widetilde{B}C}-S_{ABC})_{\rho^{(\alpha\overrightarrow{\alpha}; \beta\overrightarrow{\beta})}}.
\end{aligned}
\end{equation}
Plugging in these identities to Eq.~\eqref{eq:4-hole_ABC}, we obtain Eq.~\eqref{eq:4-hole_ABC_reduced}.

Third, we consider a different merging process depicted in Fig.~\ref{fig:4_hole_combined_appendix}(c). The density matrices involved are identical to that involved in the merging process in Fig.~\ref{fig:4_hole_combined}(a,b$\to $c), but the subsystem choices are different; in the case of Fig.~\ref{fig:4_hole_combined_appendix}(c), we reduce the density matrices to the ones on smaller regions before merging them. The resulting region is $\widetilde{W}$ instead of $W$. Let us denote the merged state as $\widetilde{\rho}_{\widetilde{W}}^{(\alpha\overrightarrow{\alpha}; \beta\overrightarrow{\beta})}$, for which the fifth hole of $\widetilde{W}$ carries the vacuum sector because $N_{1\gamma}^1=\delta_{\gamma,1}$.
This density matrix has the following properties:
\begin{equation}
\begin{aligned}
I(A:C\vert \widetilde{B})_{\widetilde{\rho}^{(\alpha\overrightarrow{\alpha}; \beta\overrightarrow{\beta})}} &=0\\
\widetilde{\rho}_{A\widetilde{B}}^{(\alpha\overrightarrow{\alpha}; \beta\overrightarrow{\beta})}&={\rho}_{A\widetilde{B}}^{(\alpha\overrightarrow{\alpha}; \beta\overrightarrow{\beta})}\\
\widetilde{\rho}_{\widetilde{B}C}^{(\alpha\overrightarrow{\alpha}; \beta\overrightarrow{\beta})}&={\rho}_{\widetilde{B}C}^{(\alpha\overrightarrow{\alpha}; \beta\overrightarrow{\beta})}.
\end{aligned} \label{eq:properties_widetilde_rho}
\end{equation}
Equations~\eqref{eq:properties_widetilde_rho} and \eqref{eq:4-hole_ABC_reduced} imply that 
\begin{equation}
\rho^{(\alpha\overrightarrow{\alpha}; \beta\overrightarrow{\beta})}_{\widetilde{W}}= \widetilde{\rho}^{(\alpha\overrightarrow{\alpha}; \beta\overrightarrow{\beta})}_{\widetilde{W}}. \label{eq:eq_of_rho_and_tilde_rho}
\end{equation}
 This is because any two tripartite states over $A, B, $ and $C$ obeying $I(A:C|B)=0$ and having identical reduced density matrices over $AB$ and $BC$ are equal~\cite{Kim2014a}.

Next, we observe that 
\begin{equation}
\Tr_{\widetilde{W}\setminus G_L} \widetilde{\rho}_{\widetilde{W}}^{(\alpha\overrightarrow{\alpha}; \beta\overrightarrow{\beta})}=\rho_{G_L}^{\alpha\merge \beta}. \label{eq:tilde_rho_independence_G_L}
\end{equation}
This statement follows from two facts. (i) The state $\widetilde{\rho}_{\widetilde{W}}^{(\alpha\overrightarrow{\alpha}; \beta\overrightarrow{\beta})}$  is conditionally independent with respect to the partition in Fig.~\ref{fig:4_hole_combined_appendix}(c), where the conditioned subsystem is that between the triple line. (ii) We can apply a partial trace on the unconditioned subsystems to connect the holes with sectors $\overrightarrow{\beta}, 1 , \overrightarrow{\alpha}$ and the outer boundary of $\widetilde{W}$. This partial trace reduces $A\to \widetilde{A}$ and $C \to \widetilde{C}$. Thus, $I(\widetilde{A}: \widetilde{C} \vert \widetilde{B})_{\widetilde{\rho}^{(\alpha\overrightarrow{\alpha}; \beta\overrightarrow{\beta})}}=0$. Furthermore, $\widetilde{A}\widetilde{B}\widetilde{C}$ can be smoothly deformed into $G_L$. 

These two facts imply that the state $\widetilde{\rho}^{(\alpha\overrightarrow{\alpha}; \beta\overrightarrow{\beta})}$, after reduced to $G_L$, must be the maximum-entropy state with the sector choice $\alpha$ and $\beta$. This implies Eq.~\eqref{eq:tilde_rho_independence_G_L}.  

Finally, it follows from Eq.~\eqref{eq:eq_of_rho_and_tilde_rho} and~\eqref{eq:tilde_rho_independence_G_L} that Eq.~\eqref{eq:_rho_independence_G_L} is true. This completes the proof.
\end{proof}

\newpage

\bibliography{ref}

\begin{thebibliography}{55}%
\makeatletter
\providecommand \@ifxundefined [1]{%
 \@ifx{#1\undefined}
}%
\providecommand \@ifnum [1]{%
 \ifnum #1\expandafter \@firstoftwo
 \else \expandafter \@secondoftwo
 \fi
}%
\providecommand \@ifx [1]{%
 \ifx #1\expandafter \@firstoftwo
 \else \expandafter \@secondoftwo
 \fi
}%
\providecommand \natexlab [1]{#1}%
\providecommand \enquote  [1]{``#1''}%
\providecommand \bibnamefont  [1]{#1}%
\providecommand \bibfnamefont [1]{#1}%
\providecommand \citenamefont [1]{#1}%
\providecommand \href@noop [0]{\@secondoftwo}%
\providecommand \href [0]{\begingroup \@sanitize@url \@href}%
\providecommand \@href[1]{\@@startlink{#1}\@@href}%
\providecommand \@@href[1]{\endgroup#1\@@endlink}%
\providecommand \@sanitize@url [0]{\catcode `\\12\catcode `\$12\catcode
  `\&12\catcode `\#12\catcode `\^12\catcode `\_12\catcode `\%12\relax}%
\providecommand \@@startlink[1]{}%
\providecommand \@@endlink[0]{}%
\providecommand \url  [0]{\begingroup\@sanitize@url \@url }%
\providecommand \@url [1]{\endgroup\@href {#1}{\urlprefix }}%
\providecommand \urlprefix  [0]{URL }%
\providecommand \Eprint [0]{\href }%
\providecommand \doibase [0]{https://doi.org/}%
\providecommand \selectlanguage [0]{\@gobble}%
\providecommand \bibinfo  [0]{\@secondoftwo}%
\providecommand \bibfield  [0]{\@secondoftwo}%
\providecommand \translation [1]{[#1]}%
\providecommand \BibitemOpen [0]{}%
\providecommand \bibitemStop [0]{}%
\providecommand \bibitemNoStop [0]{.\EOS\space}%
\providecommand \EOS [0]{\spacefactor3000\relax}%
\providecommand \BibitemShut  [1]{\csname bibitem#1\endcsname}%
\let\auto@bib@innerbib\@empty
%</preamble>
\bibitem [{\citenamefont {Wen}(2004)}]{Wen2004}%
  \BibitemOpen
  \bibfield  {author} {\bibinfo {author} {\bibfnamefont {X.-G.}\ \bibnamefont
  {Wen}},\ }\href@noop {} {\emph {\bibinfo {title} {Quantum Field Theory of
  Many-Body Systems}}}\ (\bibinfo  {publisher} {Oxford Univ. Press, Oxford},\
  \bibinfo {year} {2004})\BibitemShut {NoStop}%
\bibitem [{\citenamefont {Leinaas}\ and\ \citenamefont
  {Myrheim}(1977)}]{Leinaas1977}%
  \BibitemOpen
  \bibfield  {author} {\bibinfo {author} {\bibfnamefont {J.~M.}\ \bibnamefont
  {Leinaas}}\ and\ \bibinfo {author} {\bibfnamefont {J.}~\bibnamefont
  {Myrheim}},\ }\bibfield  {title} {\bibinfo {title} {On the theory of
  identical particles},\ }\href {https://doi.org/10.1007/BF02727953} {\bibfield
   {journal} {\bibinfo  {journal} {Il Nuovo Cimento B (1971-1996)}\ }\textbf
  {\bibinfo {volume} {37}},\ \bibinfo {pages} {1} (\bibinfo {year}
  {1977})}\BibitemShut {NoStop}%
\bibitem [{\citenamefont {Wilczek}(1982)}]{PhysRevLett.48.1144}%
  \BibitemOpen
  \bibfield  {author} {\bibinfo {author} {\bibfnamefont {F.}~\bibnamefont
  {Wilczek}},\ }\bibfield  {title} {\bibinfo {title} {Magnetic flux, angular
  momentum, and statistics},\ }\href
  {https://doi.org/10.1103/PhysRevLett.48.1144} {\bibfield  {journal} {\bibinfo
   {journal} {Phys. Rev. Lett.}\ }\textbf {\bibinfo {volume} {48}},\ \bibinfo
  {pages} {1144} (\bibinfo {year} {1982})}\BibitemShut {NoStop}%
\bibitem [{\citenamefont {Arovas}\ \emph {et~al.}(1984)\citenamefont {Arovas},
  \citenamefont {Schrieffer},\ and\ \citenamefont
  {Wilczek}}]{PhysRevLett.53.722}%
  \BibitemOpen
  \bibfield  {author} {\bibinfo {author} {\bibfnamefont {D.}~\bibnamefont
  {Arovas}}, \bibinfo {author} {\bibfnamefont {J.~R.}\ \bibnamefont
  {Schrieffer}},\ and\ \bibinfo {author} {\bibfnamefont {F.}~\bibnamefont
  {Wilczek}},\ }\bibfield  {title} {\bibinfo {title} {Fractional statistics and
  the quantum hall effect},\ }\href
  {https://doi.org/10.1103/PhysRevLett.53.722} {\bibfield  {journal} {\bibinfo
  {journal} {Phys. Rev. Lett.}\ }\textbf {\bibinfo {volume} {53}},\ \bibinfo
  {pages} {722} (\bibinfo {year} {1984})}\BibitemShut {NoStop}%
\bibitem [{\citenamefont {Kalmeyer}\ and\ \citenamefont
  {Laughlin}(1987)}]{Kalmeyer1987}%
  \BibitemOpen
  \bibfield  {author} {\bibinfo {author} {\bibfnamefont {V.}~\bibnamefont
  {Kalmeyer}}\ and\ \bibinfo {author} {\bibfnamefont {R.~B.}\ \bibnamefont
  {Laughlin}},\ }\bibfield  {title} {\bibinfo {title} {Equivalence of the
  resonating-valence-bond and fractional quantum hall states},\ }\href
  {https://doi.org/10.1103/PhysRevLett.59.2095} {\bibfield  {journal} {\bibinfo
   {journal} {Phys. Rev. Lett.}\ }\textbf {\bibinfo {volume} {59}},\ \bibinfo
  {pages} {2095} (\bibinfo {year} {1987})}\BibitemShut {NoStop}%
\bibitem [{\citenamefont {Moore}\ and\ \citenamefont {Read}(1991)}]{Moore1991}%
  \BibitemOpen
  \bibfield  {author} {\bibinfo {author} {\bibfnamefont {G.}~\bibnamefont
  {Moore}}\ and\ \bibinfo {author} {\bibfnamefont {N.}~\bibnamefont {Read}},\
  }\bibfield  {title} {\bibinfo {title} {Nonabelions in the fractional quantum
  hall effect},\ }\href@noop {} {\bibfield  {journal} {\bibinfo  {journal}
  {Nucl. Phys. B}\ }\textbf {\bibinfo {volume} {360}},\ \bibinfo {pages} {362 }
  (\bibinfo {year} {1991})}\BibitemShut {NoStop}%
\bibitem [{\citenamefont {Tsui}\ \emph {et~al.}(1982)\citenamefont {Tsui},
  \citenamefont {Stormer},\ and\ \citenamefont {Gossard}}]{Tsui1982}%
  \BibitemOpen
  \bibfield  {author} {\bibinfo {author} {\bibfnamefont {D.~C.}\ \bibnamefont
  {Tsui}}, \bibinfo {author} {\bibfnamefont {H.~L.}\ \bibnamefont {Stormer}},\
  and\ \bibinfo {author} {\bibfnamefont {A.~C.}\ \bibnamefont {Gossard}},\
  }\bibfield  {title} {\bibinfo {title} {Two-dimensional magnetotransport in
  the extreme quantum limit},\ }\href@noop {} {\bibfield  {journal} {\bibinfo
  {journal} {Phys. Rev. Lett.}\ }\textbf {\bibinfo {volume} {48}},\ \bibinfo
  {pages} {1559} (\bibinfo {year} {1982})}\BibitemShut {NoStop}%
\bibitem [{\citenamefont {{Kane}}\ and\ \citenamefont
  {{Fisher}}(1997)}]{1997PhRvB..5515832K}%
  \BibitemOpen
  \bibfield  {author} {\bibinfo {author} {\bibfnamefont {C.~L.}\ \bibnamefont
  {{Kane}}}\ and\ \bibinfo {author} {\bibfnamefont {M.~P.~A.}\ \bibnamefont
  {{Fisher}}},\ }\bibfield  {title} {\bibinfo {title} {{Quantized thermal
  transport in the fractional quantum Hall effect}},\ }\href
  {https://doi.org/10.1103/PhysRevB.55.15832} {\bibfield  {journal} {\bibinfo
  {journal} {\prb}\ }\textbf {\bibinfo {volume} {55}},\ \bibinfo {pages}
  {15832} (\bibinfo {year} {1997})},\ \Eprint
  {https://arxiv.org/abs/cond-mat/9603118} {arXiv:cond-mat/9603118 [cond-mat]}
  \BibitemShut {NoStop}%
\bibitem [{\citenamefont {{Kitaev}}(2006)}]{Kitaev2006solo}%
  \BibitemOpen
  \bibfield  {author} {\bibinfo {author} {\bibfnamefont {A.}~\bibnamefont
  {{Kitaev}}},\ }\bibfield  {title} {\bibinfo {title} {{Anyons in an exactly
  solved model and beyond}},\ }\href
  {https://doi.org/10.1016/j.aop.2005.10.005} {\bibfield  {journal} {\bibinfo
  {journal} {Ann. Phys.}\ }\textbf {\bibinfo {volume} {321}},\ \bibinfo {pages}
  {2} (\bibinfo {year} {2006})},\ \Eprint
  {https://arxiv.org/abs/cond-mat/0506438} {arXiv:cond-mat/0506438
  [cond-mat.mes-hall]} \BibitemShut {NoStop}%
\bibitem [{\citenamefont {Bravyi}\ and\ \citenamefont
  {Kitaev}(1998)}]{Bravyi1998}%
  \BibitemOpen
  \bibfield  {author} {\bibinfo {author} {\bibfnamefont {S.}~\bibnamefont
  {Bravyi}}\ and\ \bibinfo {author} {\bibfnamefont {A.}~\bibnamefont
  {Kitaev}},\ }\bibfield  {title} {\bibinfo {title} {{Quantum codes on a
  lattice with boundary}},\ }\href@noop {} {\bibfield  {journal} {\bibinfo
  {journal} {ArXiv e-prints}\ } (\bibinfo {year} {1998})},\ \Eprint
  {https://arxiv.org/abs/quant-ph/9811052} {arXiv:quant-ph/9811052
  [cond-mat.str-el]} \BibitemShut {NoStop}%
\bibitem [{\citenamefont {Beigi}\ \emph {et~al.}(2011)\citenamefont {Beigi},
  \citenamefont {Shor},\ and\ \citenamefont {Whalen}}]{Beigi2010}%
  \BibitemOpen
  \bibfield  {author} {\bibinfo {author} {\bibfnamefont {S.}~\bibnamefont
  {Beigi}}, \bibinfo {author} {\bibfnamefont {P.~W.}\ \bibnamefont {Shor}},\
  and\ \bibinfo {author} {\bibfnamefont {D.}~\bibnamefont {Whalen}},\
  }\bibfield  {title} {\bibinfo {title} {The quantum double model with
  boundary: Condensations and symmetries},\ }\href
  {https://doi.org/10.1007/s00220-011-1294-x} {\bibfield  {journal} {\bibinfo
  {journal} {Comm. Math. Phys.}\ }\textbf {\bibinfo {volume} {306}},\ \bibinfo
  {pages} {663} (\bibinfo {year} {2011})},\ \Eprint
  {https://arxiv.org/abs/1006.5479v5} {1006.5479v5} \BibitemShut {NoStop}%
\bibitem [{\citenamefont {{Kitaev}}\ and\ \citenamefont
  {{Kong}}(2012)}]{KitaevKong2012}%
  \BibitemOpen
  \bibfield  {author} {\bibinfo {author} {\bibfnamefont {A.}~\bibnamefont
  {{Kitaev}}}\ and\ \bibinfo {author} {\bibfnamefont {L.}~\bibnamefont
  {{Kong}}},\ }\bibfield  {title} {\bibinfo {title} {{Models for Gapped
  Boundaries and Domain Walls}},\ }\href
  {https://doi.org/10.1007/s00220-012-1500-5} {\bibfield  {journal} {\bibinfo
  {journal} {Comm. Math. Phys.}\ }\textbf {\bibinfo {volume} {313}},\ \bibinfo
  {pages} {351} (\bibinfo {year} {2012})},\ \Eprint
  {https://arxiv.org/abs/1104.5047} {arXiv:1104.5047 [cond-mat.str-el]}
  \BibitemShut {NoStop}%
\bibitem [{\citenamefont {{Levin}}(2013)}]{2013PhRvX...3b1009L}%
  \BibitemOpen
  \bibfield  {author} {\bibinfo {author} {\bibfnamefont {M.}~\bibnamefont
  {{Levin}}},\ }\bibfield  {title} {\bibinfo {title} {{Protected Edge Modes
  without Symmetry}},\ }\href {https://doi.org/10.1103/PhysRevX.3.021009}
  {\bibfield  {journal} {\bibinfo  {journal} {Phys. Rev. X}\ }\textbf {\bibinfo
  {volume} {3}},\ \bibinfo {eid} {021009} (\bibinfo {year} {2013})},\ \Eprint
  {https://arxiv.org/abs/1301.7355} {arXiv:1301.7355 [cond-mat.str-el]}
  \BibitemShut {NoStop}%
\bibitem [{\citenamefont {Barkeshli}\ \emph {et~al.}(2013)\citenamefont
  {Barkeshli}, \citenamefont {Jian},\ and\ \citenamefont
  {Qi}}]{Barkeshli2013a}%
  \BibitemOpen
  \bibfield  {author} {\bibinfo {author} {\bibfnamefont {M.}~\bibnamefont
  {Barkeshli}}, \bibinfo {author} {\bibfnamefont {C.-M.}\ \bibnamefont
  {Jian}},\ and\ \bibinfo {author} {\bibfnamefont {X.-L.}\ \bibnamefont {Qi}},\
  }\bibfield  {title} {\bibinfo {title} {Theory of defects in abelian
  topological states},\ }\href@noop {} {\bibfield  {journal} {\bibinfo
  {journal} {Phys. Rev. B}\ }\textbf {\bibinfo {volume} {88}},\ \bibinfo
  {pages} {235103} (\bibinfo {year} {2013})}\BibitemShut {NoStop}%
\bibitem [{\citenamefont {{Kong}}(2014)}]{Kong2014}%
  \BibitemOpen
  \bibfield  {author} {\bibinfo {author} {\bibfnamefont {L.}~\bibnamefont
  {{Kong}}},\ }\bibfield  {title} {\bibinfo {title} {{Anyon condensation and
  tensor categories}},\ }\href
  {https://doi.org/10.1016/j.nuclphysb.2014.07.003} {\bibfield  {journal}
  {\bibinfo  {journal} {Nucl. Phys. B}\ }\textbf {\bibinfo {volume} {886}},\
  \bibinfo {pages} {436} (\bibinfo {year} {2014})},\ \Eprint
  {https://arxiv.org/abs/1307.8244} {arXiv:1307.8244 [cond-mat.str-el]}
  \BibitemShut {NoStop}%
\bibitem [{\citenamefont {Lan}\ \emph {et~al.}(2015)\citenamefont {Lan},
  \citenamefont {Wang},\ and\ \citenamefont {Wen}}]{Lan2015}%
  \BibitemOpen
  \bibfield  {author} {\bibinfo {author} {\bibfnamefont {T.}~\bibnamefont
  {Lan}}, \bibinfo {author} {\bibfnamefont {J.~C.}\ \bibnamefont {Wang}},\ and\
  \bibinfo {author} {\bibfnamefont {X.-G.}\ \bibnamefont {Wen}},\ }\bibfield
  {title} {\bibinfo {title} {Gapped domain walls, gapped boundaries, and
  topological degeneracy},\ }\href@noop {} {\bibfield  {journal} {\bibinfo
  {journal} {Phys. Rev. Lett.}\ }\textbf {\bibinfo {volume} {114}},\ \bibinfo
  {pages} {076402} (\bibinfo {year} {2015})}\BibitemShut {NoStop}%
\bibitem [{\citenamefont {Hung}\ and\ \citenamefont
  {Wan}(2015{\natexlab{a}})}]{Hung2015}%
  \BibitemOpen
  \bibfield  {author} {\bibinfo {author} {\bibfnamefont {L.-Y.}\ \bibnamefont
  {Hung}}\ and\ \bibinfo {author} {\bibfnamefont {Y.}~\bibnamefont {Wan}},\
  }\bibfield  {title} {\bibinfo {title} {Ground-state degeneracy of topological
  phases on open surfaces},\ }\href@noop {} {\bibfield  {journal} {\bibinfo
  {journal} {Phys. Rev. Lett.}\ }\textbf {\bibinfo {volume} {114}},\ \bibinfo
  {pages} {076401} (\bibinfo {year} {2015}{\natexlab{a}})}\BibitemShut
  {NoStop}%
\bibitem [{\citenamefont {Cong}\ \emph
  {et~al.}(2017{\natexlab{a}})\citenamefont {Cong}, \citenamefont {Cheng},\
  and\ \citenamefont {Wang}}]{Cong2017}%
  \BibitemOpen
  \bibfield  {author} {\bibinfo {author} {\bibfnamefont {I.}~\bibnamefont
  {Cong}}, \bibinfo {author} {\bibfnamefont {M.}~\bibnamefont {Cheng}},\ and\
  \bibinfo {author} {\bibfnamefont {Z.}~\bibnamefont {Wang}},\ }\bibfield
  {title} {\bibinfo {title} {Universal quantum computation with gapped
  boundaries},\ }\href@noop {} {\bibfield  {journal} {\bibinfo  {journal}
  {Phys. Rev. Lett.}\ }\textbf {\bibinfo {volume} {119}},\ \bibinfo {pages}
  {170504} (\bibinfo {year} {2017}{\natexlab{a}})}\BibitemShut {NoStop}%
\bibitem [{\citenamefont {Kitaev}(2003)}]{Kitaev2003}%
  \BibitemOpen
  \bibfield  {author} {\bibinfo {author} {\bibfnamefont {A.}~\bibnamefont
  {Kitaev}},\ }\bibfield  {title} {\bibinfo {title} {Fault-tolerant quantum
  computation by anyons},\ }\href
  {https://doi.org/https://doi.org/10.1016/S0003-4916(02)00018-0} {\bibfield
  {journal} {\bibinfo  {journal} {Ann. Phys.}\ }\textbf {\bibinfo {volume}
  {303}},\ \bibinfo {pages} {2 } (\bibinfo {year} {2003})}\BibitemShut
  {NoStop}%
\bibitem [{\citenamefont {{Kapustin}}\ and\ \citenamefont
  {{Saulina}}(2011)}]{2011NuPhB.845..393K}%
  \BibitemOpen
  \bibfield  {author} {\bibinfo {author} {\bibfnamefont {A.}~\bibnamefont
  {{Kapustin}}}\ and\ \bibinfo {author} {\bibfnamefont {N.}~\bibnamefont
  {{Saulina}}},\ }\bibfield  {title} {\bibinfo {title} {{Topological boundary
  conditions in abelian Chern-Simons theory}},\ }\href
  {https://doi.org/10.1016/j.nuclphysb.2010.12.017} {\bibfield  {journal}
  {\bibinfo  {journal} {Nucl. Phys. B}\ }\textbf {\bibinfo {volume} {845}},\
  \bibinfo {pages} {393} (\bibinfo {year} {2011})},\ \Eprint
  {https://arxiv.org/abs/1008.0654} {arXiv:1008.0654 [hep-th]} \BibitemShut
  {NoStop}%
\bibitem [{\citenamefont {{Fuchs}}\ \emph {et~al.}(2013)\citenamefont
  {{Fuchs}}, \citenamefont {{Schweigert}},\ and\ \citenamefont
  {{Valentino}}}]{2013CMaPh.321..543F}%
  \BibitemOpen
  \bibfield  {author} {\bibinfo {author} {\bibfnamefont {J.}~\bibnamefont
  {{Fuchs}}}, \bibinfo {author} {\bibfnamefont {C.}~\bibnamefont
  {{Schweigert}}},\ and\ \bibinfo {author} {\bibfnamefont {A.~r.}\ \bibnamefont
  {{Valentino}}},\ }\bibfield  {title} {\bibinfo {title} {{Bicategories for
  Boundary Conditions and for Surface Defects in 3-d TFT}},\ }\href
  {https://doi.org/10.1007/s00220-013-1723-0} {\bibfield  {journal} {\bibinfo
  {journal} {Comm. Math. Phys.}\ }\textbf {\bibinfo {volume} {321}},\ \bibinfo
  {pages} {543} (\bibinfo {year} {2013})},\ \Eprint
  {https://arxiv.org/abs/1203.4568} {arXiv:1203.4568 [hep-th]} \BibitemShut
  {NoStop}%
\bibitem [{\citenamefont {Bais}\ and\ \citenamefont
  {Slingerland}(2009)}]{Bais2009}%
  \BibitemOpen
  \bibfield  {author} {\bibinfo {author} {\bibfnamefont {F.~A.}\ \bibnamefont
  {Bais}}\ and\ \bibinfo {author} {\bibfnamefont {J.~K.}\ \bibnamefont
  {Slingerland}},\ }\bibfield  {title} {\bibinfo {title} {Condensate-induced
  transitions between topologically ordered phases},\ }\href@noop {} {\bibfield
   {journal} {\bibinfo  {journal} {Phys. Rev. B}\ }\textbf {\bibinfo {volume}
  {79}},\ \bibinfo {pages} {045316} (\bibinfo {year} {2009})}\BibitemShut
  {NoStop}%
\bibitem [{\citenamefont {{Wang}}\ and\ \citenamefont
  {{Wen}}(2015)}]{2015PhRvB..91l5124W}%
  \BibitemOpen
  \bibfield  {author} {\bibinfo {author} {\bibfnamefont {J.~C.}\ \bibnamefont
  {{Wang}}}\ and\ \bibinfo {author} {\bibfnamefont {X.-G.}\ \bibnamefont
  {{Wen}}},\ }\bibfield  {title} {\bibinfo {title} {{Boundary degeneracy of
  topological order}},\ }\href {https://doi.org/10.1103/PhysRevB.91.125124}
  {\bibfield  {journal} {\bibinfo  {journal} {\prb}\ }\textbf {\bibinfo
  {volume} {91}},\ \bibinfo {eid} {125124} (\bibinfo {year} {2015})},\ \Eprint
  {https://arxiv.org/abs/1212.4863} {arXiv:1212.4863 [cond-mat.str-el]}
  \BibitemShut {NoStop}%
\bibitem [{\citenamefont {Hung}\ and\ \citenamefont
  {Wan}(2015{\natexlab{b}})}]{Hung2015a}%
  \BibitemOpen
  \bibfield  {author} {\bibinfo {author} {\bibfnamefont {L.-Y.}\ \bibnamefont
  {Hung}}\ and\ \bibinfo {author} {\bibfnamefont {Y.}~\bibnamefont {Wan}},\
  }\bibfield  {title} {\bibinfo {title} {Generalized ade classification of
  topological boundaries and anyon condensation},\ }\href@noop {} {\bibfield
  {journal} {\bibinfo  {journal} {JHEP}\ }\textbf {\bibinfo {volume} {1507}},\
  \bibinfo {pages} {120}}\BibitemShut {NoStop}%
\bibitem [{\citenamefont {Neupert}\ \emph
  {et~al.}(2016{\natexlab{a}})\citenamefont {Neupert}, \citenamefont {He},
  \citenamefont {von Keyserlingk}, \citenamefont {Sierra},\ and\ \citenamefont
  {Bernevig}}]{Neupert2016}%
  \BibitemOpen
  \bibfield  {author} {\bibinfo {author} {\bibfnamefont {T.}~\bibnamefont
  {Neupert}}, \bibinfo {author} {\bibfnamefont {H.}~\bibnamefont {He}},
  \bibinfo {author} {\bibfnamefont {C.}~\bibnamefont {von Keyserlingk}},
  \bibinfo {author} {\bibfnamefont {G.}~\bibnamefont {Sierra}},\ and\ \bibinfo
  {author} {\bibfnamefont {B.~A.}\ \bibnamefont {Bernevig}},\ }\bibfield
  {title} {\bibinfo {title} {Boson condensation in topologically ordered
  quantum liquids},\ }\href@noop {} {\bibfield  {journal} {\bibinfo  {journal}
  {Phys. Rev. B}\ }\textbf {\bibinfo {volume} {93}},\ \bibinfo {pages} {115103}
  (\bibinfo {year} {2016}{\natexlab{a}})}\BibitemShut {NoStop}%
\bibitem [{\citenamefont {Neupert}\ \emph
  {et~al.}(2016{\natexlab{b}})\citenamefont {Neupert}, \citenamefont {He},
  \citenamefont {von Keyserlingk}, \citenamefont {Sierra},\ and\ \citenamefont
  {Bernevig}}]{Neupert2016a}%
  \BibitemOpen
  \bibfield  {author} {\bibinfo {author} {\bibfnamefont {T.}~\bibnamefont
  {Neupert}}, \bibinfo {author} {\bibfnamefont {H.}~\bibnamefont {He}},
  \bibinfo {author} {\bibfnamefont {C.}~\bibnamefont {von Keyserlingk}},
  \bibinfo {author} {\bibfnamefont {G.}~\bibnamefont {Sierra}},\ and\ \bibinfo
  {author} {\bibfnamefont {B.~A.}\ \bibnamefont {Bernevig}},\ }\bibfield
  {title} {\bibinfo {title} {No-go theorem for boson condensation in
  topologically ordered quantum liquids},\ }\href
  {https://doi.org/10.1088/1367-2630/18/12/123009} {\bibfield  {journal}
  {\bibinfo  {journal} {New J. Phys.}\ }\textbf {\bibinfo {volume} {18}},\
  \bibinfo {pages} {123009} (\bibinfo {year} {2016}{\natexlab{b}})}\BibitemShut
  {NoStop}%
\bibitem [{\citenamefont {Cong}\ \emph
  {et~al.}(2017{\natexlab{b}})\citenamefont {Cong}, \citenamefont {Cheng},\
  and\ \citenamefont {Wang}}]{Cong2017a}%
  \BibitemOpen
  \bibfield  {author} {\bibinfo {author} {\bibfnamefont {I.}~\bibnamefont
  {Cong}}, \bibinfo {author} {\bibfnamefont {M.}~\bibnamefont {Cheng}},\ and\
  \bibinfo {author} {\bibfnamefont {Z.}~\bibnamefont {Wang}},\ }\bibfield
  {title} {\bibinfo {title} {Hamiltonian and algebraic theories of gapped
  boundaries in topological phases of matter},\ }\href@noop {} {\bibfield
  {journal} {\bibinfo  {journal} {Comm. Math. Phys.}\ }\textbf {\bibinfo
  {volume} {355}},\ \bibinfo {pages} {645} (\bibinfo {year}
  {2017}{\natexlab{b}})}\BibitemShut {NoStop}%
\bibitem [{\citenamefont {{Hu}}\ \emph {et~al.}(2018)\citenamefont {{Hu}},
  \citenamefont {{Luo}}, \citenamefont {{Pankovich}}, \citenamefont {{Wan}},\
  and\ \citenamefont {{Wu}}}]{2018JHEP...01..134H}%
  \BibitemOpen
  \bibfield  {author} {\bibinfo {author} {\bibfnamefont {Y.}~\bibnamefont
  {{Hu}}}, \bibinfo {author} {\bibfnamefont {Z.-X.}\ \bibnamefont {{Luo}}},
  \bibinfo {author} {\bibfnamefont {R.}~\bibnamefont {{Pankovich}}}, \bibinfo
  {author} {\bibfnamefont {Y.}~\bibnamefont {{Wan}}},\ and\ \bibinfo {author}
  {\bibfnamefont {Y.-S.}\ \bibnamefont {{Wu}}},\ }\bibfield  {title} {\bibinfo
  {title} {{Boundary Hamiltonian theory for gapped topological phases on an
  open surface}},\ }\href {https://doi.org/10.1007/JHEP01(2018)134} {\bibfield
  {journal} {\bibinfo  {journal} {Journal of High Energy Physics}\ }\textbf
  {\bibinfo {volume} {2018}},\ \bibinfo {eid} {134} (\bibinfo {year} {2018})},\
  \Eprint {https://arxiv.org/abs/1706.03329} {arXiv:1706.03329
  [cond-mat.str-el]} \BibitemShut {NoStop}%
\bibitem [{\citenamefont {{May-Mann}}\ and\ \citenamefont
  {{Hughes}}(2019)}]{2019PhRvB..99o5134M}%
  \BibitemOpen
  \bibfield  {author} {\bibinfo {author} {\bibfnamefont {J.}~\bibnamefont
  {{May-Mann}}}\ and\ \bibinfo {author} {\bibfnamefont {T.~L.}\ \bibnamefont
  {{Hughes}}},\ }\bibfield  {title} {\bibinfo {title} {{Families of gapped
  interfaces between fractional quantum Hall states}},\ }\href
  {https://doi.org/10.1103/PhysRevB.99.155134} {\bibfield  {journal} {\bibinfo
  {journal} {\prb}\ }\textbf {\bibinfo {volume} {99}},\ \bibinfo {eid} {155134}
  (\bibinfo {year} {2019})},\ \Eprint {https://arxiv.org/abs/1810.03673}
  {arXiv:1810.03673 [cond-mat.str-el]} \BibitemShut {NoStop}%
\bibitem [{\citenamefont {Shen}\ and\ \citenamefont
  {Hung}(2019)}]{2019arXiv190108285S}%
  \BibitemOpen
  \bibfield  {author} {\bibinfo {author} {\bibfnamefont {C.}~\bibnamefont
  {Shen}}\ and\ \bibinfo {author} {\bibfnamefont {L.-Y.}\ \bibnamefont
  {Hung}},\ }\bibfield  {title} {\bibinfo {title} {Defect verlinde formula for
  edge excitations in topological order},\ }\href
  {https://doi.org/10.1103/PhysRevLett.123.051602} {\bibfield  {journal}
  {\bibinfo  {journal} {Phys. Rev. Lett.}\ }\textbf {\bibinfo {volume} {123}},\
  \bibinfo {pages} {051602} (\bibinfo {year} {2019})},\ \Eprint
  {https://arxiv.org/abs/1901.08285} {1901.08285} \BibitemShut {NoStop}%
\bibitem [{\citenamefont {{Hu}}\ and\ \citenamefont
  {{Wan}}(2019)}]{2019JHEP...05..110H}%
  \BibitemOpen
  \bibfield  {author} {\bibinfo {author} {\bibfnamefont {Y.}~\bibnamefont
  {{Hu}}}\ and\ \bibinfo {author} {\bibfnamefont {Y.}~\bibnamefont {{Wan}}},\
  }\bibfield  {title} {\bibinfo {title} {{Entanglement entropy, quantum
  fluctuations, and thermal entropy in topological phases}},\ }\href
  {https://doi.org/10.1007/JHEP05(2019)110} {\bibfield  {journal} {\bibinfo
  {journal} {Journal of High Energy Physics}\ }\textbf {\bibinfo {volume}
  {2019}},\ \bibinfo {eid} {110} (\bibinfo {year} {2019})},\ \Eprint
  {https://arxiv.org/abs/1901.09033} {arXiv:1901.09033 [cond-mat.str-el]}
  \BibitemShut {NoStop}%
\bibitem [{\citenamefont {Bridgeman}\ and\ \citenamefont
  {Barter}(2020)}]{2019arXiv190706692B}%
  \BibitemOpen
  \bibfield  {author} {\bibinfo {author} {\bibfnamefont {J.~C.}\ \bibnamefont
  {Bridgeman}}\ and\ \bibinfo {author} {\bibfnamefont {D.}~\bibnamefont
  {Barter}},\ }\bibfield  {title} {\bibinfo {title} {Computing data for
  {L}evin-{W}en with defects},\ }\href
  {https://doi.org/10.22331/q-2020-06-04-277} {\bibfield  {journal} {\bibinfo
  {journal} {{Quantum}}\ }\textbf {\bibinfo {volume} {4}},\ \bibinfo {pages}
  {277} (\bibinfo {year} {2020})},\ \Eprint {https://arxiv.org/abs/1907.06692}
  {1907.06692} \BibitemShut {NoStop}%
\bibitem [{\citenamefont {{Lan}}\ \emph {et~al.}(2020)\citenamefont {{Lan}},
  \citenamefont {{Wen}}, \citenamefont {{Kong}},\ and\ \citenamefont
  {{Wen}}}]{Lan2019}%
  \BibitemOpen
  \bibfield  {author} {\bibinfo {author} {\bibfnamefont {T.}~\bibnamefont
  {{Lan}}}, \bibinfo {author} {\bibfnamefont {X.}~\bibnamefont {{Wen}}},
  \bibinfo {author} {\bibfnamefont {L.}~\bibnamefont {{Kong}}},\ and\ \bibinfo
  {author} {\bibfnamefont {X.-G.}\ \bibnamefont {{Wen}}},\ }\bibfield  {title}
  {\bibinfo {title} {{Gapped domain walls between 2+1D topologically ordered
  states}},\ }\href {https://doi.org/10.1103/PhysRevResearch.2.023331}
  {\bibfield  {journal} {\bibinfo  {journal} {Phys. Rev. Res.}\ }\textbf
  {\bibinfo {volume} {2}},\ \bibinfo {eid} {023331} (\bibinfo {year} {2020})},\
  \Eprint {https://arxiv.org/abs/1911.08470} {arXiv:1911.08470
  [cond-mat.str-el]} \BibitemShut {NoStop}%
\bibitem [{\citenamefont {{Lootens}}\ \emph {et~al.}(2020)\citenamefont
  {{Lootens}}, \citenamefont {{Fuchs}}, \citenamefont {{Haegeman}},
  \citenamefont {{Schweigert}},\ and\ \citenamefont
  {{Verstraete}}}]{2020arXiv200811187L}%
  \BibitemOpen
  \bibfield  {author} {\bibinfo {author} {\bibfnamefont {L.}~\bibnamefont
  {{Lootens}}}, \bibinfo {author} {\bibfnamefont {J.}~\bibnamefont {{Fuchs}}},
  \bibinfo {author} {\bibfnamefont {J.}~\bibnamefont {{Haegeman}}}, \bibinfo
  {author} {\bibfnamefont {C.}~\bibnamefont {{Schweigert}}},\ and\ \bibinfo
  {author} {\bibfnamefont {F.}~\bibnamefont {{Verstraete}}},\ }\bibfield
  {title} {\bibinfo {title} {{Matrix product operator symmetries and
  intertwiners in string-nets with domain walls}},\ }\href@noop {} {\bibfield
  {journal} {\bibinfo  {journal} {arXiv e-prints}\ ,\ \bibinfo {eid}
  {arXiv:2008.11187}} (\bibinfo {year} {2020})},\ \Eprint
  {https://arxiv.org/abs/2008.11187} {arXiv:2008.11187 [quant-ph]} \BibitemShut
  {NoStop}%
\bibitem [{\citenamefont {Kitaev}\ and\ \citenamefont
  {Preskill}(2006)}]{Kitaev2006}%
  \BibitemOpen
  \bibfield  {author} {\bibinfo {author} {\bibfnamefont {A.}~\bibnamefont
  {Kitaev}}\ and\ \bibinfo {author} {\bibfnamefont {J.}~\bibnamefont
  {Preskill}},\ }\bibfield  {title} {\bibinfo {title} {Topological entanglement
  entropy},\ }\href {https://doi.org/10.1103/PhysRevLett.96.110404} {\bibfield
  {journal} {\bibinfo  {journal} {Phys. Rev. Lett.}\ }\textbf {\bibinfo
  {volume} {96}},\ \bibinfo {pages} {110404} (\bibinfo {year}
  {2006})}\BibitemShut {NoStop}%
\bibitem [{\citenamefont {Levin}\ and\ \citenamefont {Wen}(2006)}]{Levin2006}%
  \BibitemOpen
  \bibfield  {author} {\bibinfo {author} {\bibfnamefont {M.}~\bibnamefont
  {Levin}}\ and\ \bibinfo {author} {\bibfnamefont {X.-G.}\ \bibnamefont
  {Wen}},\ }\bibfield  {title} {\bibinfo {title} {Detecting topological order
  in a ground state wave function},\ }\href
  {https://doi.org/10.1103/PhysRevLett.96.110405} {\bibfield  {journal}
  {\bibinfo  {journal} {Phys. Rev. Lett.}\ }\textbf {\bibinfo {volume} {96}},\
  \bibinfo {pages} {110405} (\bibinfo {year} {2006})}\BibitemShut {NoStop}%
\bibitem [{\citenamefont {Zhang}\ \emph {et~al.}(2012)\citenamefont {Zhang},
  \citenamefont {Grover}, \citenamefont {Turner}, \citenamefont {Oshikawa},\
  and\ \citenamefont {Vishwanath}}]{Zhang2012}%
  \BibitemOpen
  \bibfield  {author} {\bibinfo {author} {\bibfnamefont {Y.}~\bibnamefont
  {Zhang}}, \bibinfo {author} {\bibfnamefont {T.}~\bibnamefont {Grover}},
  \bibinfo {author} {\bibfnamefont {A.}~\bibnamefont {Turner}}, \bibinfo
  {author} {\bibfnamefont {M.}~\bibnamefont {Oshikawa}},\ and\ \bibinfo
  {author} {\bibfnamefont {A.}~\bibnamefont {Vishwanath}},\ }\bibfield  {title}
  {\bibinfo {title} {Quasiparticle statistics and braiding from ground-state
  entanglement},\ }\href {https://doi.org/10.1103/PhysRevB.85.235151}
  {\bibfield  {journal} {\bibinfo  {journal} {Phys. Rev. B}\ }\textbf {\bibinfo
  {volume} {85}},\ \bibinfo {pages} {235151} (\bibinfo {year}
  {2012})}\BibitemShut {NoStop}%
\bibitem [{\citenamefont {Isakov}\ \emph {et~al.}(2011)\citenamefont {Isakov},
  \citenamefont {Hastings},\ and\ \citenamefont {Melko}}]{Isakov2011}%
  \BibitemOpen
  \bibfield  {author} {\bibinfo {author} {\bibfnamefont {S.~V.}\ \bibnamefont
  {Isakov}}, \bibinfo {author} {\bibfnamefont {M.~B.}\ \bibnamefont
  {Hastings}},\ and\ \bibinfo {author} {\bibfnamefont {R.~G.}\ \bibnamefont
  {Melko}},\ }\bibfield  {title} {\bibinfo {title} {Topological entanglement
  entropy of a bose-hubbard spin liquid},\ }\href@noop {} {\bibfield  {journal}
  {\bibinfo  {journal} {Nature Phys.}\ }\textbf {\bibinfo {volume} {7}},\
  \bibinfo {pages} {772} (\bibinfo {year} {2011})}\BibitemShut {NoStop}%
\bibitem [{\citenamefont {Jiang}\ \emph {et~al.}(2012)\citenamefont {Jiang},
  \citenamefont {Wang},\ and\ \citenamefont {Balents}}]{Jiang2012}%
  \BibitemOpen
  \bibfield  {author} {\bibinfo {author} {\bibfnamefont {H.-C.}\ \bibnamefont
  {Jiang}}, \bibinfo {author} {\bibfnamefont {Z.}~\bibnamefont {Wang}},\ and\
  \bibinfo {author} {\bibfnamefont {L.}~\bibnamefont {Balents}},\ }\bibfield
  {title} {\bibinfo {title} {Identifying topological order by entanglement
  entropy},\ }\href@noop {} {\bibfield  {journal} {\bibinfo  {journal} {Nature
  Phys.}\ }\textbf {\bibinfo {volume} {8}},\ \bibinfo {pages} {902} (\bibinfo
  {year} {2012})}\BibitemShut {NoStop}%
\bibitem [{\citenamefont {Cincio}\ and\ \citenamefont
  {Vidal}(2013)}]{Cincio2013}%
  \BibitemOpen
  \bibfield  {author} {\bibinfo {author} {\bibfnamefont {L.}~\bibnamefont
  {Cincio}}\ and\ \bibinfo {author} {\bibfnamefont {G.}~\bibnamefont {Vidal}},\
  }\bibfield  {title} {\bibinfo {title} {Characterizing topological order by
  studying the ground states on an infinite cylinder},\ }\href
  {https://doi.org/10.1103/PhysRevLett.110.067208} {\bibfield  {journal}
  {\bibinfo  {journal} {Phys. Rev. Lett.}\ }\textbf {\bibinfo {volume} {110}},\
  \bibinfo {pages} {067208} (\bibinfo {year} {2013})}\BibitemShut {NoStop}%
\bibitem [{\citenamefont {{Shi}}\ \emph {et~al.}(2020)\citenamefont {{Shi}},
  \citenamefont {{Kato}},\ and\ \citenamefont {{Kim}}}]{SKK2019}%
  \BibitemOpen
  \bibfield  {author} {\bibinfo {author} {\bibfnamefont {B.}~\bibnamefont
  {{Shi}}}, \bibinfo {author} {\bibfnamefont {K.}~\bibnamefont {{Kato}}},\ and\
  \bibinfo {author} {\bibfnamefont {I.~H.}\ \bibnamefont {{Kim}}},\ }\bibfield
  {title} {\bibinfo {title} {{Fusion rules from entanglement}},\ }\href
  {https://doi.org/10.1016/j.aop.2020.168164} {\bibfield  {journal} {\bibinfo
  {journal} {Ann. Phys.}\ }\textbf {\bibinfo {volume} {418}},\ \bibinfo {eid}
  {168164} (\bibinfo {year} {2020})},\ \Eprint
  {https://arxiv.org/abs/1906.09376} {arXiv:1906.09376 [cond-mat.str-el]}
  \BibitemShut {NoStop}%
\bibitem [{\citenamefont {{Shi}}\ and\ \citenamefont
  {{Kim}}(2020)}]{EntanglementBootstrap_short}%
  \BibitemOpen
  \bibfield  {author} {\bibinfo {author} {\bibfnamefont {B.}~\bibnamefont
  {{Shi}}}\ and\ \bibinfo {author} {\bibfnamefont {I.~H.}\ \bibnamefont
  {{Kim}}},\ }\bibfield  {title} {\bibinfo {title} {{Domain wall topological
  entanglement entropy}},\ }\href@noop {} {\bibfield  {journal} {\bibinfo
  {journal} {arXiv e-prints}\ ,\ \bibinfo {eid} {arXiv:2008.11794}} (\bibinfo
  {year} {2020})},\ \Eprint {https://arxiv.org/abs/2008.11794}
  {arXiv:2008.11794 [cond-mat.str-el]} \BibitemShut {NoStop}%
\bibitem [{\citenamefont {Ferrara}\ \emph {et~al.}(1973)\citenamefont
  {Ferrara}, \citenamefont {Grillo},\ and\ \citenamefont
  {Gatto}}]{Ferrara1973}%
  \BibitemOpen
  \bibfield  {author} {\bibinfo {author} {\bibfnamefont {S.}~\bibnamefont
  {Ferrara}}, \bibinfo {author} {\bibfnamefont {A.}~\bibnamefont {Grillo}},\
  and\ \bibinfo {author} {\bibfnamefont {R.}~\bibnamefont {Gatto}},\ }\bibfield
   {title} {\bibinfo {title} {Tensor representations of conformal algebra and
  conformally covariant operator product expansion},\ }\href
  {https://doi.org/https://doi.org/10.1016/0003-4916(73)90446-6} {\bibfield
  {journal} {\bibinfo  {journal} {Ann. Phys.}\ }\textbf {\bibinfo {volume}
  {76}},\ \bibinfo {pages} {161 } (\bibinfo {year} {1973})}\BibitemShut
  {NoStop}%
\bibitem [{\citenamefont {Polyakov}(1974)}]{Polyakov1974}%
  \BibitemOpen
  \bibfield  {author} {\bibinfo {author} {\bibfnamefont {A.~M.}\ \bibnamefont
  {Polyakov}},\ }\bibfield  {title} {\bibinfo {title} {{Nonhamiltonian approach
  to conformal quantum field theory}},\ }\href@noop {} {\bibfield  {journal}
  {\bibinfo  {journal} {Zh. Eksp. Teor. Fiz.}\ }\textbf {\bibinfo {volume}
  {66}},\ \bibinfo {pages} {23} (\bibinfo {year} {1974})},\ \bibinfo {note}
  {[Sov. Phys. JETP39,9(1974)]}\BibitemShut {NoStop}%
%%CITATION = ZETFA,66,23;%%
\bibitem [{\citenamefont {Kim}(2015)}]{Kim2015sydney}%
  \BibitemOpen
  \bibfield  {author} {\bibinfo {author} {\bibfnamefont {I.~H.}\ \bibnamefont
  {Kim}},\ }\bibfield  {title} {\bibinfo {title} {{Conservation laws from
  entanglement}},\ }\href
  {www.physics.usyd.edu.au/quantum/Coogee2015/Presentations/Kim.pdf} {\bibfield
   {journal} {\bibinfo  {journal} {Sydney Quantum Information Theory
  Workshop,}\ } (\bibinfo {year} {Jan. 22, 2015})}\BibitemShut {NoStop}%
\bibitem [{\citenamefont {{Shi}}(2019)}]{Shi2018}%
  \BibitemOpen
  \bibfield  {author} {\bibinfo {author} {\bibfnamefont {B.}~\bibnamefont
  {{Shi}}},\ }\bibfield  {title} {\bibinfo {title} {{Seeing topological
  entanglement through the information convex}},\ }\href
  {https://doi.org/10.1103/PhysRevResearch.1.033048} {\bibfield  {journal}
  {\bibinfo  {journal} {Phys. Rev. Res.}\ }\textbf {\bibinfo {volume} {1}},\
  \bibinfo {eid} {033048} (\bibinfo {year} {2019})},\ \Eprint
  {https://arxiv.org/abs/1810.01986} {arXiv:1810.01986 [cond-mat.str-el]}
  \BibitemShut {NoStop}%
\bibitem [{\citenamefont {Williamson}\ \emph {et~al.}(2019)\citenamefont
  {Williamson}, \citenamefont {Dua},\ and\ \citenamefont
  {Cheng}}]{Williamson2019}%
  \BibitemOpen
  \bibfield  {author} {\bibinfo {author} {\bibfnamefont {D.~J.}\ \bibnamefont
  {Williamson}}, \bibinfo {author} {\bibfnamefont {A.}~\bibnamefont {Dua}},\
  and\ \bibinfo {author} {\bibfnamefont {M.}~\bibnamefont {Cheng}},\ }\bibfield
   {title} {\bibinfo {title} {Spurious topological entanglement entropy from
  subsystem symmetries},\ }\href
  {https://doi.org/10.1103/PhysRevLett.122.140506} {\bibfield  {journal}
  {\bibinfo  {journal} {Phys. Rev. Lett.}\ }\textbf {\bibinfo {volume} {122}},\
  \bibinfo {pages} {140506} (\bibinfo {year} {2019})}\BibitemShut {NoStop}%
\bibitem [{\citenamefont {Lieb}\ and\ \citenamefont {Ruskai}(1973)}]{Lieb1973}%
  \BibitemOpen
  \bibfield  {author} {\bibinfo {author} {\bibfnamefont {E.~H.}\ \bibnamefont
  {Lieb}}\ and\ \bibinfo {author} {\bibfnamefont {M.~B.}\ \bibnamefont
  {Ruskai}},\ }\bibfield  {title} {\bibinfo {title} {Proof of the strong
  subadditivity of quantum mechanical entropy},\ }\href
  {https://doi.org/10.1063/1.1666274} {\bibfield  {journal} {\bibinfo
  {journal} {J. Math. Phys.}\ }\textbf {\bibinfo {volume} {14}},\ \bibinfo
  {pages} {1938} (\bibinfo {year} {1973})},\ \Eprint
  {https://arxiv.org/abs/https://doi.org/10.1063/1.1666274}
  {https://doi.org/10.1063/1.1666274} \BibitemShut {NoStop}%
\bibitem [{\citenamefont {{Kato}}\ \emph {et~al.}(2016)\citenamefont {{Kato}},
  \citenamefont {{Furrer}},\ and\ \citenamefont {{Murao}}}]{Kato2016}%
  \BibitemOpen
  \bibfield  {author} {\bibinfo {author} {\bibfnamefont {K.}~\bibnamefont
  {{Kato}}}, \bibinfo {author} {\bibfnamefont {F.}~\bibnamefont {{Furrer}}},\
  and\ \bibinfo {author} {\bibfnamefont {M.}~\bibnamefont {{Murao}}},\
  }\bibfield  {title} {\bibinfo {title} {{Information-theoretical analysis of
  topological entanglement entropy and multipartite correlations}},\ }\href
  {https://doi.org/10.1103/PhysRevA.93.022317} {\bibfield  {journal} {\bibinfo
  {journal} {\pra}\ }\textbf {\bibinfo {volume} {93}},\ \bibinfo {eid} {022317}
  (\bibinfo {year} {2016})},\ \Eprint {https://arxiv.org/abs/1505.01917}
  {arXiv:1505.01917 [quant-ph]} \BibitemShut {NoStop}%
\bibitem [{\citenamefont {Bombin}(2010)}]{Bombin2010}%
  \BibitemOpen
  \bibfield  {author} {\bibinfo {author} {\bibfnamefont {H.}~\bibnamefont
  {Bombin}},\ }\bibfield  {title} {\bibinfo {title} {Topological order with a
  twist: Ising anyons from an abelian model},\ }\href
  {https://doi.org/10.1103/PhysRevLett.105.030403} {\bibfield  {journal}
  {\bibinfo  {journal} {Phys. Rev. Lett.}\ }\textbf {\bibinfo {volume} {105}},\
  \bibinfo {pages} {030403} (\bibinfo {year} {2010})}\BibitemShut {NoStop}%
\bibitem [{\citenamefont {{Shi}}(2020)}]{2020PhRvR...2b3132S}%
  \BibitemOpen
  \bibfield  {author} {\bibinfo {author} {\bibfnamefont {B.}~\bibnamefont
  {{Shi}}},\ }\bibfield  {title} {\bibinfo {title} {{Verlinde formula from
  entanglement}},\ }\href {https://doi.org/10.1103/PhysRevResearch.2.023132}
  {\bibfield  {journal} {\bibinfo  {journal} {Phys. Rev. Res.}\ }\textbf
  {\bibinfo {volume} {2}},\ \bibinfo {eid} {023132} (\bibinfo {year} {2020})},\
  \Eprint {https://arxiv.org/abs/1911.01470} {arXiv:1911.01470
  [cond-mat.str-el]} \BibitemShut {NoStop}%
\bibitem [{\citenamefont {Petz}(1988)}]{Petz1987}%
  \BibitemOpen
  \bibfield  {author} {\bibinfo {author} {\bibfnamefont {D.}~\bibnamefont
  {Petz}},\ }\bibfield  {title} {\bibinfo {title} {{Sufficiency of channels
  over von Neumann algebras}},\ }\href {https://doi.org/10.1093/qmath/39.1.97}
  {\bibfield  {journal} {\bibinfo  {journal} {Q. J. Math.}\ }\textbf {\bibinfo
  {volume} {39}},\ \bibinfo {pages} {97} (\bibinfo {year} {1988})}\BibitemShut
  {NoStop}%
\bibitem [{\citenamefont {{Hayden}}\ \emph {et~al.}(2004)\citenamefont
  {{Hayden}}, \citenamefont {{Jozsa}}, \citenamefont {{Petz}},\ and\
  \citenamefont {{Winter}}}]{2004CMaPh.246..359H}%
  \BibitemOpen
  \bibfield  {author} {\bibinfo {author} {\bibfnamefont {P.}~\bibnamefont
  {{Hayden}}}, \bibinfo {author} {\bibfnamefont {R.}~\bibnamefont {{Jozsa}}},
  \bibinfo {author} {\bibfnamefont {D.}~\bibnamefont {{Petz}}},\ and\ \bibinfo
  {author} {\bibfnamefont {A.}~\bibnamefont {{Winter}}},\ }\bibfield  {title}
  {\bibinfo {title} {{Structure of States Which Satisfy Strong Subadditivity of
  Quantum Entropy with Equality}},\ }\href
  {https://doi.org/10.1007/s00220-004-1049-z} {\bibfield  {journal} {\bibinfo
  {journal} {Comm. Math. Phys.}\ }\textbf {\bibinfo {volume} {246}},\ \bibinfo
  {pages} {359} (\bibinfo {year} {2004})},\ \Eprint
  {https://arxiv.org/abs/quant-ph/0304007} {quant-ph/0304007} \BibitemShut
  {NoStop}%
\bibitem [{\citenamefont {Ibinson}\ \emph {et~al.}(2008)\citenamefont
  {Ibinson}, \citenamefont {Linden},\ and\ \citenamefont
  {Winter}}]{Ibinson2008}%
  \BibitemOpen
  \bibfield  {author} {\bibinfo {author} {\bibfnamefont {B.}~\bibnamefont
  {Ibinson}}, \bibinfo {author} {\bibfnamefont {N.}~\bibnamefont {Linden}},\
  and\ \bibinfo {author} {\bibfnamefont {A.}~\bibnamefont {Winter}},\
  }\bibfield  {title} {\bibinfo {title} {Robustness of quantum markov chains},\
  }\href {https://doi.org/10.1007/s00220-007-0362-8} {\bibfield  {journal}
  {\bibinfo  {journal} {Comm. Math. Phys.}\ }\textbf {\bibinfo {volume}
  {277}},\ \bibinfo {pages} {289} (\bibinfo {year} {2008})}\BibitemShut
  {NoStop}%
\bibitem [{\citenamefont {Kim}(2014)}]{Kim2014a}%
  \BibitemOpen
  \bibfield  {author} {\bibinfo {author} {\bibfnamefont {I.~H.}\ \bibnamefont
  {Kim}},\ }\bibfield  {title} {\bibinfo {title} {On the informational
  completeness of local observables},\ }\href@noop {} {\bibfield  {journal}
  {\bibinfo  {journal} {arXiv:1405.0137}\ } (\bibinfo {year}
  {2014})}\BibitemShut {NoStop}%
\end{thebibliography}%
\end{document}